\documentclass{article}
\usepackage{graphicx}
\usepackage[american]{babel}
\usepackage{tikz}
\usepackage{multirow}
\usetikzlibrary{patterns.meta}
\usepackage{float}
\usepackage{geometry}
\usepackage{tabularx}
\usepackage{booktabs}
\usepackage{tabularray}
\usepackage{xurl}
\usepackage{algorithm}
\usepackage{algpseudocode}
\usepackage{caption}
\usepackage{subcaption}
\def\checkmark{\tikz\fill[scale=0.4](0,.35) -- (.25,0) -- (1,.7) -- (.25,.15) -- cycle;} 
\newcommand{\xmark}{%
\tikz[scale=0.23] {
    \draw[line width=0.7,line cap=round] (0,0) to [bend left=6] (1,1);
    \draw[line width=0.7,line cap=round] (0.2,0.95) to [bend right=3] (0.8,0.05);
}}
\usepackage{makecell}
\newcolumntype{C}[1]{%
 >{\vbox to 5ex\bgroup\vfill\centering}%
 p{#1}%
 <{\egroup}}  
\tikzset{
dot4/.style = {circle, fill, minimum size=#1,
              inner sep=0pt, outer sep=0pt,
              postaction={
pattern={
Hatch[distance=1mm,  line width=3mm],
}, 
pattern color=white,
}},
dot4/.default = 6pt,
dot3/.style = {circle, fill, minimum size=#1,
              inner sep=0pt, outer sep=0pt,
              postaction={
pattern={
Lines[angle=00, distance=1mm,  line width=3mm]
},    pattern color=white,
}},
dot3/.default = 6pt,
dot2/.style = {circle, fill, minimum size=#1,
              inner sep=0pt, outer sep=0pt,
              postaction={
pattern={
Lines[angle=90, distance=1mm,  line width=3mm]
},    pattern color=white,
}},
dot2/.default = 6pt,
dot/.style = {circle, fill, minimum size=#1,
              inner sep=0pt, outer sep=0pt},
dot/.default = 6pt 
}

\usepackage{natbib}
\usepackage{mathtools}
\usepackage{booktabs}
\usepackage{bm}
\usepackage{bibentry}
\usepackage{amssymb}
\usepackage{amsthm}
\usepackage{wrapfig}
\usepackage{enumitem}
\usetikzlibrary{patterns}
\newcommand\xqed[1]{%
	\leavevmode\unskip\penalty9999 \hbox{}\nobreak\hfill
	\quad\hbox{#1}}
\newcommand\demo{\xqed{$\triangle$}}

\newcommand\numberthis{\addtocounter{equation}{1}\tag{\theequation}}
\newtheorem{theorem}{Theorem}
\newtheorem{proposition}[theorem]{Proposition}
\newtheorem{assumption}[theorem]{Assumption} 
\newtheorem{remark}[theorem]{Remark}
\newtheorem{corollary}[theorem]{Corollary}
\newtheorem{lemma}[theorem]{Lemma}
\theoremstyle{definition}
\newtheorem{example}[theorem]{Example}
\newtheorem{definition}[theorem]{Definition}
\DeclareCaptionFormat{algor}{%
  \hrulefill\par\offinterlineskip\vskip1pt%
    \textbf{#1#2}#3\offinterlineskip\hrulefill}
\DeclareCaptionStyle{algori}{singlelinecheck=off,format=algor,labelsep=space}
\definecolor{c1}{RGB}{74,143,222}
\definecolor{c2}{RGB}{195,135,44}

\title{Using Time Structure to Estimate Causal Effects}

\begin{document}
\author{Tom Hochsprung$^{1, 2}$\footnote{Email address of corresponding author: tom.hochsprung@dlr.de} \footnote{Code to reproduce results available at \protect\url{https://gitlab.com/dlr-dw/using_time_structure_to_estimate_causal_effects_code}} \and Jakob Runge$^{3, 2, 1}$ \and Andreas Gerhardus$^1$}
\date{\small
    $^1$German Aerospace Center (DLR), Institute of Data Science, Jena, Germany\\%
    $^2$Technische Universität Berlin, Berlin, Germany\\%
    $^3$Center for Scalable Data Analytics and Artificial Intelligence (ScaDS.AI) Dresden/Leipzig, TU Dresden, Dresden, Germany \\[2ex]%
    \normalsize
    \today}
\maketitle	
\begin{abstract}
There exist several approaches for estimating causal effects in time series when latent confounding is present. Many of these approaches rely on additional auxiliary observed variables or time series such as instruments, negative controls or time series that satisfy the front- or backdoor criterion in certain graphs. In this paper, we present a novel approach for estimating direct (and via Wright's path rule total) causal effects in a time series setup which does not rely on additional auxiliary observed variables or time series. This approach assumes that the underlying time series is a Structural Vector Autoregressive (SVAR) process and estimates direct causal effects by solving certain linear equation systems made up of different covariances and model parameters. We state sufficient graphical criteria in terms of the so-called full time graph under which these linear equations systems are uniquely solvable and under which their solutions contain the to-be-identified direct causal effects as components. We also state sufficient lag-based criteria under which the previously mentioned graphical conditions are satisfied and, thus, under which direct causal effects are identifiable. Several numerical experiments underline the correctness and applicability of our results.

    \end{abstract}
\section{Introduction}
\label{sec_intro}
Identifying causal effects is an important task in numerous research areas such as social science [e.g., \citealp{angrist1996identification, morgan2014counterfactuals, imbens2015causal}], epidemiology [e.g., \citealp{hernan2006instruments, didelez2010assumptions}] and Earth system science [e.g., \citealp{runge2023causal}]. Unfortunately, this identification is often not possible due to latent confounding [e.g., \citealp{pearl2009causality, peters2017elements}].
If a latent variable confounds cause and effect, then one can attribute any portion of dependence between cause and effect to spurious dependencies introduced by the latent confounder. Therefore, without further assumptions, it is in general not possible to identify causal effects in the presence of latent confounding.
 
In this paper, we consider the case where cause, effect and latent confounder are time series. In particular, we assume that all these time series are components of a stable structural vector autoregressive (SVAR) process [e.g., \citealp{lutkepohl2005new, hyvarinen2010estimation, moneta2011causal, malinsky2018causal}] with finite-second-moment noise.
As we will show, direct causal effects can in this setting often be identified by just using covariances between cause and effect plus qualitative knowledge of the structural relationships between cause, effect and latent confounder. Together with Wright's path rule \citep{wright1934method}, we thus implicitly also show that \emph{total} causal effects can often be identified from the same assumptions (which, however, we do not explicitly discuss in our paper).
 
Our results stand in stark contrast to most existing identification results for latently confounded causal effects in a time series setup: These results frequently resemble results in the i.i.d.\ setting and typically require \emph{additional} auxiliary variables or time series such as negative controls \citep{hu2023using}, instruments \citep{mogensen2023instrumental, michael2024instrumental, thams2022identifying}, or time series that satisfy the backdoor or frontdoor criterion in certain graphs \citep{eichler2007causal,eichler2010granger,eichler2012causal}. Our results, however, do \emph{not} rely on additional auxiliary observed variables or time series. 

Within the domain of time series analysis, we are only aware of a few other works that also not require additional variables or time series for identification of latently confounded causal effects: One approach for SVAR processes is to combine biased estimators of various structural coefficients into an unbiased estimator of the direct causal effect [\citealp{malinvaud1961estimation}, particularly Section 6 in \citealp{bercu2013sharp,proia2013further}]. This approach, however, has only been developed for rather simple SVAR processes and, in particular, for a rather simple confounding structure and is much less general than our approach.\footnote{Namely, this approach is \emph{not} applicable to Examples \ref{ex_app_1}, \ref{ex_app_2} and \ref{ex_app_3} in Section \ref{sec_examples} of the Appendix and to the real world example in Section \ref{real_world_example}. In our paper, this approach is only applicable to Example \ref{example1} in Section \ref{sec_main_identifiability_result}.} Furthermore, this approach results in to-be-solved equation systems for which it is not immediately clear from the SVAR process alone whether direct causal effects are identifiable or not. 
Besides this approach, there is the recent work from  \citet{schur2024decor} (and the related and more application-oriented work of \citet{mahecha2010global}). \citet{schur2024decor} transform the entire problem to the frequency domain and show that if the latent confounder is sparse in the frequency domain and some further technical assumptions are satisfied, then causal effects can be identified without further observed variables or time series.
However, besides that their results are proven for continuous bounded time, the assumption of sparsity in the frequency domain does not seem to be easily translatable to SVAR processes. 

Our identification results are formulated in the time domain and for discrete time: For identifying the direct causal effect of some variable $X_{t-h}$ on another variable $Y_t$, we rely on solving a certain linear equation system of the form
\begin{align*}
\Gamma_{R,Y_t} = \Gamma_{R, C} \cdot v,
\end{align*}
where $R$ and $C$ are sets of observed variables at various time points, $\Gamma_{R,C}$ denotes the matrix of covariances between the variables in $R$ and the variables in $C$, and $\Gamma_{R,Y_t}$ denotes the vector of covariances between the variables in $R$ and $Y_t$. Hereby, in contrast to approaches such as instrumental time series \citep{thams2022identifying}, the set $C$ and thus this linear equation system is also allowed to contain variables with strictly larger time index than $t$ because these ``future'' variables allow us to express the latent covariance structure just in terms of covariances between observed variables.

We also derive graphical conditions that are sufficient for this linear equation system to be uniquely solvable and the solution vector $v$ having the to-be-identified causal effect as a certain component. These graphical conditions
are stated in terms of the so-called \emph{full time graph} [e.g., \citealp{peters2013causal}, Section 10 in \citealp{peters2017elements, malinsky2018causal, runge2023causal, gerhardus2024characterization}], which is a graph that contains vertices for each component time series at each particular time point and edges between these vertices if there is a direct causal effect from one vertex onto another vertex---see Figure \ref{example_graph_1} in the main paper and Figures \ref{ex_full_time_graph_app_1}, \ref{ex_full_time_graph_app_2} and \ref{ex_full_time_graph_app_3} in the Appendix for examples.

When articulating and deriving these graphical conditions, we will frequently use the notion of \textit{treks}, that is, paths without colliders, which we have borrowed from the Linear Structural Equation Models (LSEMs) literature \citep{sullivant2010trek,foygel2012half,barber2022half}.\footnote{This borrowing of terms should not be too surprising given that SVAR processes, depending on the precise time index set, either are or closely resemble Linear Structural Equation Models (LSEMs).} This borrowing of terms does not, however, imply that our results are trivial extensions of existing LSEM identification results: Our identification results are neither a trivial extension of results that rely on projecting away latent variables [e.g., \citealp{foygel2012half}], because an autodependent latent confounder typically implies that the latent confounding is non-sparse after projecting it away, thus hindering identification using these approaches. Nor are our results a trivial extension of an LSEM-identification approach that explicitly models latent variables \citep{barber2022half} as we do, because this approach assumes that the latent variables are source variables and independent, which is an assumption that we do not make and which, in addition, is already violated for any autodependent latent time series. (We do assume that the latent time series do not have observed parents, but this is a much weaker assumption.)
Furthermore, we work with \emph{infinite} graphs, whereas the LSEM literature usually only applies to finite graphs. Moreover, in contrast to LSEMs, we can also exploit the fact that full time graphs exhibit regular edge patterns (which is also known as \emph{causal stationarity} [e.g., \citealp{gerhardus2024characterization}]), and thus graphical conditions can to some degree be translated into assumptions on lags, which one can rather easily implement on a computer or use to study classes of full time graphs for which causal effects are identifiable.

The structure of our paper is as follows:  We start in \textbf{Section \ref{sec_var_processes}} by reviewing some preliminaries on SVAR processes, their graphical representation via full time graphs and the notion of generic identifiability.  We then explain our graph- and in particular trek-based identifiability results in \textbf{Section \ref{sec_main_identifiability_result}}. While we focus on identifiability and not on statistical estimation, we will also present a simple consistency result in this section. In \textbf{Section \ref{sec_sufficient_lag_based_criteria}}, we then state sufficient lag-based criteria for identifiability. Afterwards, in  \textbf{Section \ref{sec_numerical_experiments}}, we empirically underline the correctness of our theory by demonstrating convergence of the estimated to the true direct causal effects for randomly drawn full time graphs and parameters.  
In \textbf{Section \ref{real_world_example}}, we then further illustrate our approach by providing numerical simulations for the recent electricity market example from \citet{tiedemann2024identifying}. For this example, we do not rely on further observed instrumental time series for identification unlike \citet{tiedemann2024identifying}. In \textbf{Section \ref{sec_proof_of_main_result}}, we then present the proof of our main theorem (Theorem \ref{main_theorem}). In \textbf{Section \ref{sec_conclusion}}, we finish with a conclusion and an outlook. \textbf{Section \ref{sec_glossary}} in the Appendix contains a glossary providing an overview of our frequently used notation and the \textbf{Appendix} in general contains proof details and further explanations and examples.

\section{SVAR processes and their full time graphs}
\label{sec_var_processes}
\textbf{Basic notation and (linear) SVAR equation:}	In this paper\footnote{We frequently complement the symbols $\mathbb{N}$, $\mathbb{Z}$ or $\mathbb{R}$ by indices indicating which numbers are excluded or where these particular sets start. For example, $\mathbb{N}_{>0}$ stands for the natural numbers starting at $1$, and $\mathbb{R}_{\neq 0}$ stands for the real numbers excluding zero.}, we consider a stochastic process $\{S_t\}_{t\in\mathbb{Z}}$ with unobserved component time series $\{U^1_t\}_{t\in\mathbb{Z}},\allowbreak\ldots,\allowbreak\{U^{d_U}_t\}_{t\in\mathbb{Z}}$ where $d_U\in \mathbb{N}_{\geq 0}$ and observed component time series $\{O^1_t\}_{t\in\mathbb{Z}},\allowbreak\ldots,\allowbreak\{O^{d_O}_t\}_{t\in\mathbb{Z}}$ where $d_O\in \mathbb{N}_{>0}$, that is, $\{S_t\}_{t\in\mathbb{Z}}=\allowbreak\{(S^1_t,\allowbreak\ldots,\allowbreak S^{d}_t)\}_{t\in\mathbb{Z}}=\allowbreak\{(U^1_t,\allowbreak\ldots,\allowbreak U^{d_U}_t,\allowbreak O^1_t,\allowbreak\ldots,\allowbreak O^{d_O}_t)\}_{t\in\mathbb{Z}}$ where $d:=d_U+d_O$. If $d_U=1$ or $d_O=1$, then we sometimes drop the superscript and just write $U_t$ instead of $U^1_t$ and $O_t$ instead of $O^1_t$, respectively.
We assume that $U^i_t,O^i_t\in \mathbb{R}$ for all $t\in\mathbb{Z}$ and $i\in [d_U]_1:=\{1,\ldots, d_U\}$ respectively $i\in [d_O]_1:=\{1,\ldots,d_O\}$. Without loss of generality, we let $\{Y_t\}_{t\in\mathbb{Z}}:=\{O^1_t\}_{t\in\mathbb{Z}}$ be the target time series, and we are interested in identifying the direct causal effects of some other observed component time series $\{X_t\}_{t\in\mathbb{Z}}:=\{O^i_t\}_{t\in\mathbb{Z}}$ with $i\in[d_O]_1$ onto $\{Y_t\}_{t\in\mathbb{Z}}$. (Note that $\{X_t\}_{t\in\mathbb{Z}}$ is allowed to equal $\{Y_t\}_{t\in\mathbb{Z}}$).

We further assume that
the joint process $\{S_t\}_{t\in\mathbb{Z}}$ is a (linear) SVAR process of some order $p\in\mathbb{N}_{\geq 0}$ [e.g., \citealp{hyvarinen2010estimation, moneta2011causal, malinsky2018causal}, Chapters 2 and 9 in \citealp{lutkepohl2005new}]: This SVAR-assumption means that one can write every $S_t$ as a linear function of $S_t,\allowbreak S_{t-1},\ldots,\allowbreak S_{t-p}$ and a noise vector $\epsilon_t\in\mathbb{R}^d$ via the coefficient matrices $A^{(0)},\allowbreak A^{(1)},\ldots,\allowbreak A^{(p)}\in \mathbb{R}^{d\times d}$ with $A^{(p)}\neq 0$ if $p>0$. That is, (almost surely) for all $t\in\mathbb{Z}$, 
  
	\begin{align}
	\label{var_process}
S_t	= A^{(0)}S_t + A^{(1)}S_{t-1}+\ldots+A^{(p)}S_{t-p} + \epsilon_t.
	\end{align}
	We furthermore assume that $\{\epsilon_{t}\}_{t\in\mathbb{Z}}$ is an i.i.d.\ process and that each $\epsilon_t$ has mean zero (for notational simplicity) and diagonal covariance matrix $\Sigma$ with non-negative and finite diagonal entries. Moreover, Assumptions \ref{assumption_no_instantaneous_self_edges} and \ref{assumption_acyclicity} later in this section
	ensure that $I_d-A^{(0)}$ (where $I_d$ denotes the $d\times d$-identity matrix) is invertible.
	This invertibility assumption allows one to write equation \eqref{var_process} as a classical VAR process (with, in general, non-diagonal noise-covariance matrix) that is, as
	\begin{align}
	    \label{var_process_2}
	    S_t	= (I_d-A^{(0)})^{-1}A^{(1)}S_{t-1}+\ldots+(I_d-A^{(0)})^{-1}A^{(p)}S_{t-p} + (I_d-A^{(0)})^{-1}\epsilon_t.
	\end{align}
    Throughout the paper, we assume that $\{S_t\}_{t\in\mathbb{Z}}$ is stable, that is, that $\{S_t\}_{t\in\mathbb{Z}}$ interpreted as a classical VAR process is stable, that is, $\det(I_d - (I_d-A^{(0)})^{-1}A^{(1)}\lambda-(I_d-A^{(0)})^{-1}A^{(2)}\lambda^2-\cdots -(I_d-A^{(0)})^{-1}A^{(p)}\lambda^p)=0$ implies $|\lambda| >  1$. This stability assumption is often made and helpful because, together with finite-second-moment noise, it implies weak stationarity, that is, that the mean of $S_t$ does not depend on $t$ and that the covariance matrices $\mathbb{E}[S_tS_{t-h}^T]$ only depend on $h\in\mathbb{Z}$ and not on $t$.\footnote{In fact, stability plus finite-second-moment i.i.d.\ noise also implies strict stationarity, that is, that all distributions of finitely many variables are time-shift invariant---see for example Section 2.1.2 in \citet{fan2008nonlinear} or Section 4 in \citet{hochsprungglobal} for an explanation of this fact (in the latter reference the further assumption that each $\epsilon_t$ has independent components is not required for the explanation).} In addition, stable SVAR processes have a moving average representation, that is, one can write $S_t=\sum_{i=0}^\infty \Phi_i\epsilon_{t-i}$ where each $\Phi_i\in \mathbb{R}^{d\times d}$ and each component series $\{(\Phi_i)_{k,l}\}_{i\in\mathbb{N}_{\geq 0}}$ with $k,l\in[d]_1$ is absolutely summable,  and where $\sum_{i=0}^\infty \Phi_i\epsilon_{t-i}$ is understood as a limit in mean square.

\textbf{Autocovariance function:} The autocovariance function $\Gamma_S:\mathbb{Z}\rightarrow \mathbb{R}^{d\times d}$  is defined by
\begin{align*}
\Gamma_S(h)&:=
\mathbb{E}[S_tS_{t-h}^T].
\end{align*} 
Due to weak stationarity, the missing time index in the notation $\Gamma_S(h)$ is not a problem. Also due to weak stationarity, $\Gamma_S(-h)=(\Gamma_S(h))^T$ (where $^T$ denotes the transpose).

When referring to a specific entry of $\Gamma_S(h)$, say the covariance between $O^1_{t-h}$ and $U^1_t$, we typically write $\Gamma_{U^1_tO^1_{t-h}}$ or $\Gamma_{O^1_{t-h}U^1_t}$ instead of $(\Gamma_S(h))_{1d_U+1}$. Besides, for finite sets of variables $R:=\{r_1,r_2,\ldots\}$ and $S:=\{s_1,s_2,\ldots\}$, we write $\Gamma_{R,C}$ to denote the matrix with entries $(\Gamma_{R, S})_{ij} = \Gamma_{r_i, s_j}$. Lastly, we write $\Gamma^{\textnormal{obs}}_S(h)$ to denote the submatrix of $\Gamma_S(h)$ just refererring to the observed time series.

\textbf{Full time graph:} 
One can represent Equation \eqref{var_process} by the so-called \emph{full time graph} [e.g., \citealp{peters2013causal}, Section 10 in \citealp{peters2017elements, malinsky2018causal, runge2023causal, gerhardus2024characterization}].\footnote{One could of course also use Equation \eqref{var_process_2} as a basis for the graphical representation, however, we are not interested in doing so in this paper.} In our paper, the full time graph formally is a tuple $(V,E)$ where $V$ is the vertex/node set and $E$ is the edge set. The elements of $V$ are called vertices/nodes (or, in slight abuse of notation, variables) and the elements of $E$ (directed) edges. In our paper, the full time graph has infinitely many vertices, namely $S^i_t$ for all $i\in[d]_1$ and $t\in\mathbb{Z}$. There is a directed edge from $S^{i}_{t-h}$ to $S^{j}_t$ for all $t\in \mathbb{Z}$, written as $S^{i}_{t-h}\rightarrow S^{j}_t$, if and only if the ($j,i$)-th entry of $A^{(h)}$ is nonzero. For example, if $d_U=1$ and $A^{(4)}_{21}\neq 0,$ then there is a directed edge from $U^1_{t-4}$ to $O^1_t$ for all $t\in \mathbb{Z}$. If the $(j,i)$-th entry of $A^{(h)}$ is zero or if $h>p$, then there is no directed edge from $S^{i}_{t-h}$ to $S^{j}_t$. We call an edge \emph{instantaneous} or \emph{contemporaneous}, if the edge is from some $S^i_t$ to some $S^j_t$, so if the corresponding time lag is $h=0$. For examples of full time graphs, see Figure \ref{example_graph_1} in this section and Figures \ref{ex_full_time_graph_app_1}, \ref{ex_full_time_graph_app_2} and \ref{ex_full_time_graph_app_3} in Section \ref{sec_examples} of the Appendix.

 We utilize the following terminology and notation when talking about the full time graph: 
 To denote a vertex in the full time graph without paying respect to the time step or the particular component time series, we use the lower-case letters $v$, $a$, $b$, $q$ and so on. We sometimes also index these lower-case letters by some index $i\in \mathbb{N}_{\geq 1}$, so write $v_i$, $a_i$, $b_i$, $q_i$ and so on when we refer to several such variables; the index $i$ in this case is not a time index but rather refers to some underlying enumeration of these vertices. If $a\rightarrow b$ is a directed edge in the full time graph, then $\textnormal{lag}(a\rightarrow b)\in\mathbb{N}_{\geq 0}$ denotes the lag between $a$ and $b$, that is, the time index of $b$ minus the time index of $a$. For a set of vertices $B\subseteq V$, we write $t_{\textnormal{inf}}(B)$ and $t_{\textnormal{sup}}(B)$ to denote the (potentially infinite) infimum and supremum time indices of $B$, respectively. Similarly, for a vertex $v$, we write $t(v)$ for its time index. Furthermore, for a set of variables $B$ and a component time series $\{S^i_t\}_{t\in\mathbb{Z}}$, we write $B_{S^i}:=B\cap \{S^i_t\}_{t\in\mathbb{Z}}$.
 Besides, we sometimes refer to variables that are part of the time series $\{S^{i}_t\}_{t\in\mathbb{Z}}$ as $S^{i}$-variables, $S^{i}$-nodes or $S^{i}$-vertices. We sometimes also use the notation $A^{(h)}_{w v}:=A^{(h)}_{wS^k_{t-h}}:=A^{(h)}_{S^j_t v}:=A^{(h)}_{S^j_tS^k_{t-h}}:=A^{(h)}_{S^jS^k}:=A^{(h)}_{jk}$ where $v=S^k_{t-h}$ and $w=S^j_{t}$.
 
 We now introduce some standard graphical terminology [e.g., \citealp{lauritzen1996graphical, pearl2009causality}]: We call for a vertex $v\in V$ the set of all vertices from which edges point towards $v$ the \emph{parents} of $v$. We denote the parents of $v$ by $\textnormal{pa}(v)$ and more formally, $\textnormal{pa}(v):=\bigcup_{a\in V:\;a\rightarrow v} \{a\}$. For a set of vertices $B\subseteq V$, we write $\textnormal{pa}(B):=\bigcup_{v\in B} \textnormal{pa}(v)$. Additionally, we call for a vertex $v\in V$ the set of all vertices to which $v$ points the \emph{children} of $v$. We denote the children of $v$ by $\textnormal{ch}(v)$ and more formally, $\textnormal{ch}(v):=\bigcup_{a\in V:\;v\rightarrow a} \{a\}$. For a set of vertices $B\subseteq V$, we write $\textnormal{ch}(B):=\bigcup_{v\in B} \textnormal{ch}(B)$. From Assumption \ref{assumption_no_instantaneous_self_edges} later in this section, it follows that vertices are never parents or children of themselves in our paper.
 Moreover, we occasionally call the parents of $v$ that are part of the time series $\{S^{j}_t\}_{t\in\mathbb{Z}}$ the $S^{j}$-parents of $v$; similarly for children. For example, we call the parents of $v$ that are part of $\{O^1_t\}_{t\in\mathbb{Z}}$ the $O^1$-parents of $v$. Furthermore, to denote the latent and observed parents of some vertex $v\in V$, we write $\textnormal{pa}^{\textnormal{lat}}(v)$ and $\textnormal{pa}^{\textnormal{obs}}(v)$, respectively; similarly, we write $\textnormal{pa}^{\textnormal{lat}}(B)$ and $\textnormal{pa}^{\textnormal{obs}}(B)$ for a set of vertices $B\subseteq V$. The notation for children is analogous.
 
 We call a not necessarily finite sequence of not necessarily distinct vertices $(v_1,v_2,\ldots)$, such that every pair of successive vertices in this sequence is connected by a directed edge pointing in any direction, a \emph{walk}. If all vertices with the potential exception of the endpoint vertices in a walk are unique, then we call a walk a \emph{path}. If the path is just a sequence of one vertex, then we call that path \emph{trivial}. If every edge in a path is pointing in the same direction and away from the starting vertex $v_1$ (towards the potential end-vertex), then we call a path \emph{directed}. If both endpoint vertices in a directed path are equal, then we call a path a \emph{directed cycle}. (Assumption \ref{assumption_acyclicity} later in this section ensures that no directed cycles exist in our paper). We call a vertex $v\in V$ an ancestor of another vertex $w\in V$ and $w$ a descendant of $v$ if there is a directed path from $v$ to $w$ or if $v=w$. For a vertex $v\in V$, we write $\textnormal{an}(v)$ and $\textnormal{dec}(v)$ to denote the set of all ancestors and descendants, respectively; similarly for a set of vertices $B\subseteq V$, we write $\textnormal{an}(B):=\bigcup_{v\in B}\textnormal{an}(v)$ and $\textnormal{dec}(B):=\bigcup_{v\in B}\textnormal{dec}(v)$, respectively. To just denote the latent or observed ancestors or descendants, we write $\textnormal{an}^{\textnormal{lat}}(v)$, $\textnormal{an}^{\textnormal{obs}}(v)$, $\textnormal{dec}^{\textnormal{lat}}(v)$ and $\textnormal{dec}^{\textnormal{obs}}(v)$, respectively; similarly for a set of vertices $B\subseteq V$. Besides, for a triplet of subsequent vertices $(v_{i-1}, v_i, v_{i+1})$ in a walk $\pi$ where $i\geq 2$ is the enumeration index of $\pi$ and $v_i$ is a non-endpoint vertex of $\pi$, we call $v_i$ a collider in $\pi$ if that triplet is $v_{i-1}\rightarrow v_i\leftarrow v_{i+1}$. Otherwise, we call a non-endpoint vertex $v_i$  of $\pi$ a non-collider in $\pi$ in that triplet. For a path $\pi$, we typically only write ``collider/non-collider in $\pi$" without writing ``in that triplet" as each vertex except for the potential exception of the end-point vertices in $\pi$ is unique. Furthermore, for a walk $\pi$ and a set of vertices $C\subseteq V$, we say that $C$ $d$-blocks $\pi$ if there is a non-collider in $\pi$ that is an element of $C$ or if there is a collider in $\pi$ that does not have a descendant in $C$. For two disjoint sets of vertices $B_1\subseteq V$ and $B_2\subseteq V$, we say that $C$ $d$-blocks $B_1$ and $B_2$ if every path between $B_1$ and $B_2$ is $d$-blocked by $C$.

Full time graphs exhibit certain regularity patterns (see, for example, \citet{gerhardus2024characterization}). For example, they have a \textit{repeating edges property}:  The existence of an edge $S^{i}_{t-h}\rightarrow S^{j}_t$ for some $t\in \mathbb{Z}$ implies the existence of edges $S^{i}_{\Tilde{t}-h}\rightarrow S^{j}_{\Tilde{t}}$ for all other $\Tilde{t}\in\mathbb{Z}$.

For the remainder of this paper, we make the following assumptions.
\begin{assumption}
\label{assumption_no_instantaneous_self_edges}
For all $i\in[d]_1$, we assume that $(A^{(0)})_{ii}=0$.
\end{assumption}

\begin{assumption}
\label{assumption_acyclicity}
We always assume that $A^{(0)}$ is such that the full time graph does not have directed cycles.
\end{assumption}
Note that Assumption \ref{assumption_no_instantaneous_self_edges} means that there are no (contemporaneous) self-edges. Also note that Assumption \ref{assumption_acyclicity} is always satisfied if there are no instantantaneous edges in the full time graph, since by construction there are never edges from a strictly larger time index to a strictly smaller time index in the full time graph. Moreover, note that Assumptions \ref{assumption_no_instantaneous_self_edges} and \ref{assumption_acyclicity} imply that $A^{(0)}$ can be permuted into a strictly-lower triangular matrix, thus showing that $I_d-A^{(0)}$ is invertible.

We further assume the following.
\begin{assumption}
\label{assumption_latents_have_no_observed_parents}
There are no directed edges from observed variables to latent variables. That is, $A^{(h)}_{U^jO^k}=0$ for all $h\in[p]_0$ and $j\in [d_U]_1$ and $k\in [d_O]_1$.
\end{assumption}

Assumption \ref{assumption_latents_have_no_observed_parents} or similar assumptions also occur in related work: The instrumental time series approach from \citet{thams2022identifying} also makes this assumption, and an approach to identify LSEMs by explicity modelling latent variables \citep{barber2022half} makes an even stronger assumption by assuming that all latent variables are source nodes and independent of each other, which in the time series setting is already violated for any autodependent latent confounder. In addition, we allow different latent time series to be parents of each other, which \citet{barber2022half} also exclude.

\textbf{Lag notation:} For two component time series $\{S^i_t\}_{t\in\mathbb{Z}}$ and $\{S^j_t\}_{t\in\mathbb{Z}}$ with $i,j\in[d]_1$, we write $m_{S^iS^j}$ to denote the number of lags $h\in [p]_0$ for which $A^{(h)}_{S^iS^j}\neq 0$  and, if $m_{S^iS^j}>0$, we write $l^{S^iS^j}_1,\ldots, l^{S^iS^j}_{m_{S^iS^j}}$ for the particular lags $h\in [p]_0$  in increasing order for which $A^{(h)}_{S^iS^j}\neq 0$.
In the full time graph, we call the corresponding edges $l^{S^iS^j}_k$-edges. For the special case $i=j$, we typically just write $m_{S^i}$ and $l^{S^i}_1,\ldots, l^{S^i}_{m_{S^i}}$.

For an $l^{S^iS^j}_k$-edge and a vertex $v\in \{S^i_t\}_{t\in\mathbb{Z}}$, we write $\textnormal{pa}(v,l^{S^iS^j}_k)$ for the (unique) $S^j$-variable that is a parent of $v$ and connected to $v$ via an $l^{S^iS^j}_k$-edge; similarly, for a set of vertices $B\subseteq \{S^i_t\}_{t\in\mathbb{Z}}$, we write $\textnormal{pa}(B,l^{S^iS^j}_k)$ for the $S^j$-variables that are parents of $B$ and connected to some element of $B$ via an $l^{S^iS^j}_k$-edge.
\addtocounter{theorem}{-3}
\begin{example}
\label{example_lag_notation}
For the full time graph from Figure \ref{example_graph_1} it holds that $m_U=1$, $l^U_1=1$, $m_Y=1$, $l^Y_1=3$, $m_{YU}=1$ and $l^{YU}_1=1$. \demo
\end{example}
\addtocounter{theorem}{2}
\begin{figure}[h]
	\centering
	\includegraphics[scale=0.57]{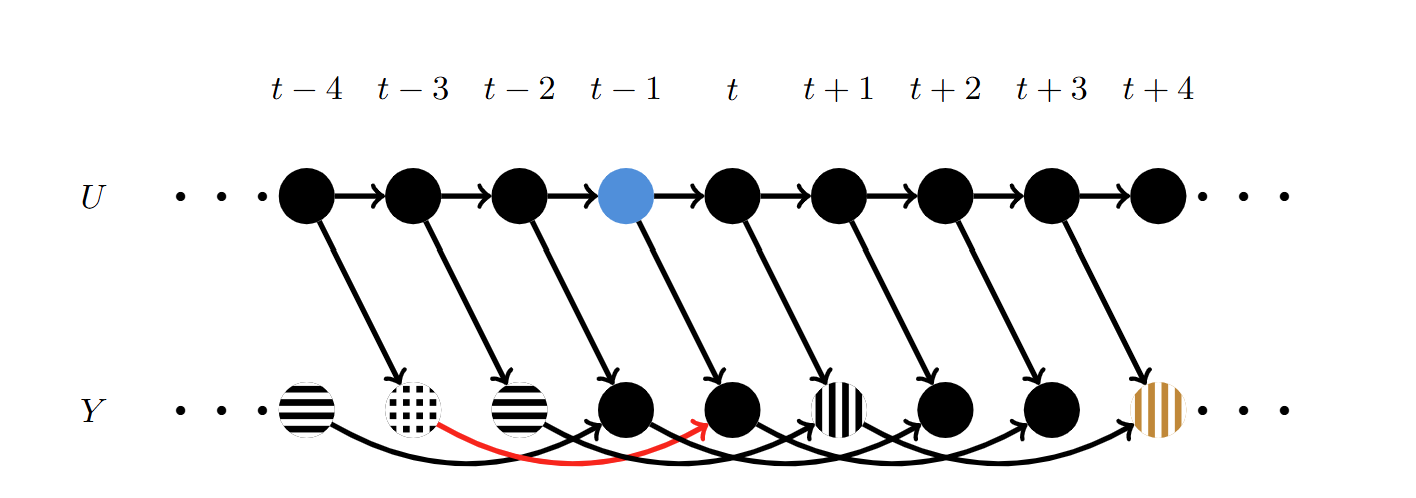}
	\caption{Example full time graph for Examples \ref{example_lag_notation}, \ref{ex_treks} and \ref{example1}. The different colors and hatchings are only relevant for Example \ref{example1}. For Example \ref{example1}: The \color{red} red \color{black} edge corresponds to the parameter $A^{(3)}_{YY}$ that one wants to identify. The \color{c1} blue \color{black} vertex is the only element of $B_U$ and the \color{c2} yellow vertex \color{black} is the only element of $F^{\textnormal{obs}}$. The vertices with vertical hatching are elements of $C$ and the vertices with horizontal hatchings are elements of $R$ (vertices with grid hatchings are both in $R$ and $C$).}
	\label{example_graph_1}
\end{figure}
\textbf{Generic identifiability:} In this paper, we study identifiability of direct causal effects in \emph{generic settings}---a notion which we have also borrowed (and slightly modified) from the LSEM literature [e.g., \citealp{sullivant2010trek}, \citealp{foygel2012half}, \citealp{barber2022half}].
Roughly speaking, a direct causal effect $A^{(h)}_{jk}$ for a given full time graph is identifiable for a certain fixed value if all other values for this direct causal effect yield different observed covariances. And roughly speaking, genericity means that the respective causal effect is identifiable for nearly all parameter choices that yield stable SVAR processes.

For a more precise definition of identifiability, let $\mathbb{R}^{d_O\times d_O}$ denote the set of all real-valued $d_O\times d_O$-matrices and let $\mathbb{R}^d_{\geq 0}$ denote the $d$-times Cartesian product of the non-negative real numbers.
	For a given full time graph $G$, enumerate the edges (so one edge for each
	non-zero entry of a coefficient matrix $A^{(0)},\ldots, A^{(p)}$) by $(h_1,j_1,k_1),\allowbreak \ldots,\allowbreak (h_N,j_N,k_N)$ where the $h$-index corresponds to the lag of the edge, the $j$-index corresponds to the component time series of the end-point vertex of that edge, and the $k$-index corresponds to the component time series of the start vertex of that edge.
Define the set $\Theta_{\textnormal{stable}}(G)$ 
	by
	\begin{align*}
	\Theta_{\textnormal{stable}}(G)&:=\{(\theta_1,\ldots,\theta_N)\in \mathbb{R}_{\neq 0}^{N}:\;\theta_i= A^{(h_i)}_{j_ik_i} \textnormal{ for all $i\in[N]_1$} \\
	& \hspace{1cm}\textnormal{ and } A^{(0)},\ldots,A^{(p)} \textnormal{ satisfy the stability condition}\},
	\end{align*}
where $\mathbb{R}_{\neq 0}^{N}$ denotes the $N$-times Cartesian product of the real numbers excluding zero.
Define the map 
	\begin{align*}
	\phi_G:\;&\Theta_{\textnormal{stable}}(G)\times \mathbb{R}^d_{\geq 0}\rightarrow \{(B_1,B_2,\ldots)\in \mathbb{R}^{d_O\times d_O}\times \mathbb{R}^{d_O\times d_O} \times \cdots  \},\\
	&(\theta_1,\ldots, \theta_N,\Sigma_{11},\ldots, \Sigma_{dd})\mapsto (\Gamma^{\textnormal{obs}}_S(0),\Gamma^{\textnormal{obs}}_S(1),\Gamma^{\textnormal{obs}}_S(2),\ldots),
	\end{align*}
	where $\{S_t\}_{t\in\mathbb{Z}}$ is any SVAR process
	with coefficient matrix entries $\theta_1,\allowbreak\ldots,\allowbreak\theta_N$  and variances  $\Sigma_{11},\allowbreak\ldots,\allowbreak \Sigma_{dd}$.
	For given $(\theta_1,\ldots,\theta_N,\Sigma_{11},\ldots, \Sigma_{dd})\in \Theta_{\textnormal{stable}}(G)\times \mathbb{R}^d_{\geq 0}$, we say that some direct causal effect $A^{(h)}_{jk}=:\theta_{i_0}$ with $i_0\in[N]_1$ is identifiable if for all $(\Tilde{\theta}_1,\ldots,\Tilde{\theta}_N,\Tilde{\Sigma}_{11},\ldots, \Tilde{\Sigma}_{dd})\in \Theta_{\textnormal{stable}}(G)\times \mathbb{R}^d_{\geq 0}$ 
	such that $\phi_G(\Tilde{\theta}_1,\ldots,\Tilde{\theta}_N,\Tilde{\Sigma}_{11},\ldots, \Tilde{\Sigma}_{dd})=\phi_G(\theta_1,\ldots,\theta_N,\Sigma_{11},\ldots, \Sigma_{dd})$ it holds that $\Tilde{\theta}_{i_0}=\theta_{i_0}$.

Regarding a more precise definition of genericity: For a given full time graph $G$, we say that some statement holds in generic settings if it holds for all $(\theta_1,\ldots,\theta_N,\Sigma_{11},\ldots, \Sigma_{dd})\in (\Theta_{\textnormal{stable}}(G)\times \mathbb{R}^d_{\geq 0})\setminus W$ where $W$ is a set with Lebesgue measure zero. Here, note that $\Theta_{\textnormal{stable}}(G)\times \mathbb{R}^d_{\geq 0}$ itself has positive Lebesgue measure---we refer to Section \ref{sec_proofs_sec_causal_effects} in the Appendix for a proof of this fact.

\section{Main identifiability result}
\label{sec_main_identifiability_result}
As mentioned in the introduction, for identifying the direct causal effect of some
$X_{t-h}$ on $Y_t$, we rely on solving a linear equation system of the form
\begin{align}
\label{main_lin_eq}
\Gamma_{R,Y_t} = \Gamma_{R, C} \cdot v,
\end{align}
where $R$ and $C$ are sets of observed variables at various time points.
Before introducing the (sufficient) theorem stating for which choices of $R$ and $C$ and under which further assumptions the direct causal effect of question can be identified from Equation \eqref{main_lin_eq}, we introduce further required graphical concepts. Our exposition for these graphical concepts partially follows the existing LSEM-paper from \citet{foygel2012half}.

	\begin{definition} [Treks]
		\label{def_trek_rule_var_process}
	Let $t_1,t_2\in\mathbb{Z}$ be arbitrary but fixed time points. A \emph{trek} in the full time graph from \emph{source} $S^{i}_{t_1}$ to \emph{target} $S^{j}_{t_2}$ is a \emph{finite} albeit arbitrary long walk that does not have any colliders.  That is, a trek from $S^{i}_{t_1}$ to $S^{j}_{t_2}$ takes the form
	\begin{align*}
S^{i}_{t_1}=v^L_l\leftarrow v^L_{l-1}\leftarrow \cdots \leftarrow v^L_1 \leftarrow v^T\rightarrow v^R_1\rightarrow \cdots \rightarrow v^R_{r-1}\rightarrow v^R_{r}=S^{j}_{t_2}.
\end{align*}	
For a trek $\pi$, we define the left-part of $\pi$ by $\textnormal{Left}(\pi):=\{v^T,v^L_1,\ldots,\allowbreak v^L_l\}$ and the right-part of $\pi$ by $\textnormal{Right}(\pi):=\{v^T,v^R_1,\ldots,v^R_r\}$. We allow $\textnormal{Left}(\pi)=\{v^T\}$ or $\textnormal{Right}(\pi)=\{v^T\}$. If 
$\textnormal{Left}(\pi)=\{v^T\}=\textnormal{Right}(\pi),$ then we say that $\pi$ is trivial. 
The top node of the trek is $v^T$.
We denote the set of all treks from $S^{j}_{t_1}$ to $S^{i}_{t_2}$ by $\mathcal{T}(S^{i}_{t_1},S^{j}_{t_2})$ and the set of all treks from $S^{j}_{t_1}$ to $S^{i}_{t_2}$ with top node $v^T$ by $\mathcal{T}(S^{i}_{t_1},S^{j}_{t_2},v^T)$. 
\end{definition}
\addtocounter{theorem}{-3}
\begin{example}
\label{ex_treks}
	Consider the full time graph in Figure \ref{example_graph_1}. Two examples of treks from $U_{t-1}$ to $Y_t$ are
	\begin{align*}
&\pi_1: U_{t-1}\rightarrow Y_t \textnormal{ and }\\
&\pi_2: U_{t-1}\leftarrow U_{t-2}\leftarrow U_{t-3}\leftarrow U_{t-4}\rightarrow Y_{t-3}\rightarrow Y_{t}.
	\end{align*}
	The top node of $\pi_1$ is $U_{t-1}$. Moreover, $\textnormal{Left}(\pi_1)=\{U_{t-1}\}$ and $\textnormal{Right}(\pi_1)=\{U_{t-1},\allowbreak Y_t\}$. Similarly, the top node of $\pi_2$ is $U_{t-4}$. Moreover, $\textnormal{Left}(\pi_2)=\{U_{t-1},\allowbreak U_{t-2},\allowbreak U_{t-3},\allowbreak U_{t-4}\}$ and $\textnormal{Right}(\pi_2)= \{U_{t-4},\allowbreak Y_{t-3},\allowbreak Y_t\}$.

		The path $U_{t-1}\rightarrow Y_t\leftarrow Y_{t+3}$
	is not a trek because it contains a collider (namely $Y_{t}$).
	\demo
\end{example}
\addtocounter{theorem}{2}
\begin{remark}
\label{remark_collider}
	Once a trek went forward in time or moved forward within the same time point, so once it passed an edge $S^{i_1}_{t_1}\rightarrow S^{i_2}_{t_2}$ (in that order) with $i_1,i_2\in[d]_1$ and $t_1 \leq t_2$, it cannot afterwards move backwards in time, so it cannot after the previous edge contain an edge $S^{i_3}_{t_3}\leftarrow S^{i_4}_{t_4}$ (in that order) with $i_3,i_4\in[d]_1$ and $t_3> t_4$, as that would create a collider. 
\end{remark}

\begin{definition} [Walk/Path \& Trek monomial]
	For a walk/path $\pi$, the \emph{walk/path monomial} of $\pi$ is 
	\begin{align*}
	\pi(A) := \prod_{a\rightarrow b\in \pi}A^{(\textnormal{lag}(a \rightarrow b))}_{ba}.
	\end{align*}
		Here, the symbol $A$ in $\pi(A)$ is shorthand for the coefficient matrices $A^{(0)},\ldots,A^{(p)}$.	 By convention, we set $\pi(A)=1$ if $\pi$ is trivial.
	
	Similarly, for a trek $\pi$ with top node $v$, the \emph{trek monomial} of $\pi$ is 
	\begin{align*}
	\pi(A, \Sigma) := \Sigma_{vv}\cdot\pi(A).
	\end{align*}

\end{definition}
\addtocounter{theorem}{-5}
\begin{example}[continued]
 The path monomial of $\pi_1$ is $\pi(A)=A^{(1)}_{21}$ and the trek monomial of $\pi_1$ is $\pi_1(A,\Sigma)=\Sigma_{11}A^{(1)}_{21}$.
	Similarly, the path monomial of $\pi_2$ is $\pi_2(A)=(A^{(1)}_{11})^{3}A^{(1)}_{21}A^{(3)}_{22}$ and the trek monomial of $\pi_2$ is $\pi_2(A,\Sigma)=\Sigma_{11}(A^{(1)}_{11})^{3}A^{(1)}_{21}A^{(3)}_{22}$.
	\demo
\end{example}
\addtocounter{theorem}{4}
\begin{definition}[System of treks/directed paths]
	Let $R$ and $C$ be finite sets of vertices such that $n:=|R|=|C|$. A \textit{system of treks/directed paths} $\Pi$ from sources $R:=\{r_1,\ldots,r_n\}$ to targets $C:=\{c_1,\ldots,c_n\}$, abbreviated by $\Pi:R\rightrightarrows C$ respectively $\Pi:R\rightarrow C$, is a set of treks/directed paths $\Pi:=\{\pi_1,\ldots,\pi_n\}$ for which each $\pi_i$ goes from $r_i$ to $c_{\sigma(i)}$ for some permutation $\sigma \in \mathcal{S}_n$.\footnote{We write $\mathcal{S}_n$ to denote the symmetric group of order $n$.}  
	 We also say that $\Pi$ has permutation $\sigma$ and sometimes, to make the order induced by $\sigma$ clear, write $\Pi:R\rightrightarrows(c_{\sigma(1)},\ldots,c_{\sigma(n)})$ respectively $\Pi:R\rightarrow(c_{\sigma(1)},\ldots,c_{\sigma(n)})$.  The sign of the system of treks/directed paths is the sign of the permutation $\sigma$, that is, $\textnormal{sgn}(\Pi):=\textnormal{sgn}(\sigma)$.
	
	We say that a system of directed paths $\Pi$ has \textit{no intersections} if the vertices of every directed path in $\Pi$ do not occur in any other directed path in $\Pi$.
	
	We say that a system of treks $\Pi$ has \textit{no sided intersection} if 
	\begin{align*}
	\textnormal{Left}(\pi_i)\cap \textnormal{Left}(\pi_j)=\emptyset = \textnormal{Right}(\pi_i)\cap \textnormal{Right}(\pi_j)\;\;\; \forall i\neq j.
	\end{align*}
	
	If $\Pi$ is a system of directed paths, then we define the monomial of $\Pi$ by
	\begin{align*}
	\Pi(A):=\prod_{\pi\in\Pi}\pi(A).
	\end{align*}
	
	If $\Pi$ is a system of treks, then we define the monomial of $\Pi$ by
	\begin{align*}
	\Pi(A,\Sigma):=\prod_{\pi\in\Pi}\pi(A,\Sigma).
	\end{align*}
\end{definition}
We are now ready to state the main theorem of this paper. Note that this theorem identifies all direct causal effects of all observed parents of $Y_t$ onto $Y_t$ \emph{at once}, and thus this theorem does not contain the symbol $X_{t-h}$ which represents some \emph{specific} observed parent.

\begin{theorem} (Main identifiability result)
\label{main_theorem}
Assume a stable SVAR process satisfying Assumptions \ref{assumption_no_instantaneous_self_edges}, \ref{assumption_acyclicity} and \ref{assumption_latents_have_no_observed_parents}.
Furthermore, assume that in the full time graph one has
\begin{enumerate}
    \item a finite set of latent variables $B_U$ such that for each latent parent $q$ of $Y_t$ it either holds that
    \begin{enumerate}
        \item $q\in B_U$, or
        \item all directed paths from the latent vertices with time index strictly smaller than $t_{\textnormal{inf}}(B_U)$ to $q$ are d-blocked by $B_U$; and
    \end{enumerate}
    \item    a finite set of observed variables $F^{\textnormal{obs}}$ such that
    \begin{enumerate}
        \item $|F^{\textnormal{obs}}|=|B_U|$, and
        \item such that for each latent parent $\Tilde{q}$ of $F^{\textnormal{obs}}$ it either holds that
        \begin{itemize}
            \item $\Tilde{q} \in B_U$, or
            \item all directed paths from the latent vertices with time index strictly smaller than $t_{\textnormal{inf}}(B_U)$ to $\Tilde{q}$ are d-blocked by $B_U$; and
        \end{itemize}
        \item there exists a system of directed paths $\Upsilon:B_U\rightarrow F^{\textnormal{obs}}$ with no intersections  such that
        \begin{itemize}
            \item all vertices in each directed path except the last one are latent, and
            \item for every other system of directed paths $\Pi:B_U\rightarrow F^{\textnormal{obs}}$ with no intersections for which all vertices in each directed path except that last one are latent it holds that $\Pi(A)\neq \Upsilon(A)$.
        \end{itemize}
    \end{enumerate}
\end{enumerate}
Furthermore, assume that 
\begin{enumerate}
\setcounter{enumi}{2}
    \item $\bigr(F^{\textnormal{obs}}\cup \textnormal{pa}^{\textnormal{obs}}(F^{\textnormal{obs}})\bigr)\cap \bigr(\{Y_t\}\cup\textnormal{pa}^{\textnormal{obs}}(Y_t)\bigr)=\emptyset.$
\end{enumerate}

Define $C:=F^{\textnormal{obs}}\cup \textnormal{pa}^{\textnormal{obs}}(F^{\textnormal{obs}}) \cup \textnormal{pa}^{\textnormal{obs}}(Y_t)$. Let $R$ be a finite set of variables such that $|R|=|C|$ and such that
\begin{enumerate}
\setcounter{enumi}{3}
    \item no element in $R$ is a descendant of $\textnormal{ForbAn}:=F^{\textnormal{obs}}\cup \{Y_t\}\,\cup\, \bigr(\textnormal{an}^{\textnormal{lat}}\bigr(F^{\textnormal{obs}}\cup \{Y_t\})\setminus \textnormal{an}^{\textnormal{lat}}(B_U)\bigr)$.
\end{enumerate}

Then, 
\begin{align*}
\Gamma_{R,Y_t} = \Gamma_{R, C} \cdot v
\end{align*}
holds in generic settings for a vector $v$
 for which, letting $i$ be such that the $i$-th colum of $\Gamma_{R,C}$ corresponds to covariances induced by the observed parent $q$ of $Y_t$, the $i$-th component of $v$ contains the direct causal effect of $q$ on $Y_t$.
 (The other components of $v$ not related to observed parents of $Y_t$ are certain functions of the parameters, which are, however, not interesting for the remainder of this work).
 
 If there in addition is 
 \begin{enumerate}
\setcounter{enumi}{4}
\item a system of treks $\Psi:R \rightrightarrows C$ with no sided intersection such that for all other systems of treks $\Pi:R \rightrightarrows C$ with no sided intersection it holds that $\Pi(A,\Sigma)\neq \Psi(A,\Sigma)$,
\end{enumerate}
  then, in generic settings, the matrix $\Gamma_{R,C}$ in equation \eqref{main_lin_eq} is invertible and hence, this solution $v$ is unique.
\end{theorem}

\addtocounter{theorem}{-6}
\begin{example}
\label{example1}
    In this example, let $m_U=1$,  $l^U_1=1$, $m_Y=1$, $l^Y_1> 1$, $m_{YU}=1$ and $l^{YU}_1=1$. A possible full time graph for this example is presented in Figure \ref{example_graph_1}; in this full time graph, $l^Y_1=3$.
    
    Assume that one wants to identify the direct causal effect corresponding to the $l^Y_1$-edge, that is, $A^{(l^Y_1)}_{YY}$.
    Note that in this example and for this task, no auxiliary time series (such as an instrument) is available for identification. 
    
    In the following, we manually discuss why Theorem \ref{main_theorem} applies (Corollary \ref{corollary_identifiability} later in Section \ref{sec_sufficient_lag_based_criteria} directly implies this fact).
    For that, define 
    \begin{align*}
    B_U&:=\{U_{t-1}\},\\
    F^{\textnormal{obs}}&:=\{Y_{t+l^Y_1+1}\}.
    \end{align*}
    Note that $\textnormal{pa}^{\textnormal{lat}}(Y_t)=\{U_{t-1}\}$ and $\textnormal{pa}^{\textnormal{lat}}(F^{\textnormal{obs}})=\{U_{t+l^Y_1}\}$ and thus, $B_U$ and $F^{\textnormal{obs}}$ satisfy requirements 1 and 2a and 2b from Theorem \ref{main_theorem}. For 2c, consider $\Upsilon=\{\upsilon_1\}$ with
    \begin{align*}
        \upsilon_1:U_{t-1}\rightarrow U_t \rightarrow \cdots \rightarrow U_{t+l^Y_1}\rightarrow Y_{t+l^Y_1+1}.
    \end{align*}
    Note that every other system of directed paths $\Pi=\{\pi_1\}:B_U\rightarrow F^{\textnormal{obs}}$ with $\Pi(A)=\Upsilon(A)$ needs to have exactly $t+l^Y_1-(t-1)=l^Y_1+1$-many $l^U_1$-edges and exactly one $l^{YU}_1$-edge. As the last edge of $\pi_1$ must be an $l^{YU}_1$-edge (by requirement 2c in Theorem \ref{main_theorem}), it follows that the last edge of $\pi_1$ is as in $\upsilon_1$. The previous edges in $\pi_1$ can thus only be $l^U_1$-edges as only terms for $l^U_1$-edges are left in $\Pi(A)=\Upsilon(A)$. Therefore,  $\Pi=\Upsilon$. 
    
    Next, note that $\textnormal{pa}^{\textnormal{obs}}(F^{\textnormal{obs}})=\{Y_{t+1}\}$ and $\textnormal{pa}^{\textnormal{obs}}(Y_t)=\{Y_{t-l^Y_1}\}$ and hence,
    \begin{align*}
                C&:=\{Y_{t-l^Y_1}, Y_{t+l^Y_1+1}, Y_{t+1}\}.
    \end{align*}
    Now define
    \begin{align*}
                R&:=\{Y_{t-l^Y_1}, Y_{t-l^Y_1+1}, Y_{t-l^Y_1-1}\}.
    \end{align*}
    Because
    \begin{align*}
        \textnormal{ForbAn}=\{U_t, U_{t+1}, \ldots, U_{t+l^Y_1}, Y_t, Y_{t+l^Y_1+1}\},
    \end{align*}
    $R$ satisfies the non-descendance requirement (condition 4) from Theorem \ref{main_theorem}.
    For the system of treks from $\Psi:R\rightrightarrows C$ take
    \begin{align*}
        &\Psi_1:\;Y_{t-l^Y_1}\\
        &\Psi_2:\;Y_{t-l^Y_1+1}\rightarrow Y_{t+1}\\
        &\Psi_3:\; Y_{t-l^Y_1-1}\leftarrow U_{t-l^Y_1-2}\rightarrow U_{t-l^Y_1-1}\rightarrow \cdots \rightarrow U_{t+l^Y_1}\rightarrow Y_{t+l^Y_1+1}. 
    \end{align*}
    This system of treks has no sided intersection because $l^Y_1>1$. It also holds that $\Psi$ has a unique monomial among all other system of treks $\Pi:R\rightrightarrows C$ with no sided intersection---this fact follows because $\Psi$ equals the constructed system of treks in the proof of Lemma \ref{lemma_resid_class1} from Section \ref{sec_sufficient_lag_based_criteria}. (Alternatively, the existence of a system of treks $\Psi:R\rightrightarrows C$ as required by Theorem \ref{main_theorem} directly follows from Lemma \ref{lemma_resid_class1} as we will explain in Section \ref{sec_sufficient_lag_based_criteria}).
   
    Therefore, all requirements from Theorem \ref{main_theorem} are satisfied and thus, the to-be-identified direct causal effect $A^{(l^Y_1)}_{YY}$ is given (in generic settings) by
    \begin{align}
    \label{est_ex_1}
        \begin{pmatrix}
        \color{red}A^{(l^Y_1)}_{YY}\color{black}\\
        \vdots \\
        \textnormal{other terms}\\
        \vdots
        \end{pmatrix}=\begin{pmatrix}
        \Gamma_{Y_{t-l^Y_1}\color{red}Y_{t-l^Y_1}\color{black}} & \Gamma_{Y_{t-l^Y_1}Y_{t+l^Y_1+1}} & \Gamma_{Y_{t-l^Y_1}Y_{t+1}}\\
        \Gamma_{Y_{t-l^Y_1+1}\color{red}Y_{t-l^Y_1}\color{black}} & \Gamma_{Y_{t-l^Y_1+1}Y_{t+l^Y_1+1}} & \Gamma_{Y_{t-l^Y_1+1}Y_{t+1}}\\
        \Gamma_{Y_{t-l^Y_1-1}\color{red}Y_{t-l^Y_1}\color{black}} & \Gamma_{Y_{t-l^Y_1-1}Y_{t+l^Y_1+1}} & \Gamma_{Y_{t-l^Y_1-1}Y_{t+1}}
        \end{pmatrix}^{-1}\cdot
        \begin{pmatrix}
        \Gamma_{Y_{t-l^Y_1}Y_{t}}\\
        \Gamma_{Y_{t-l^Y_1+1}Y_{t}}\\
        \Gamma_{Y_{t-l^Y_1-1}Y_{t}}
        \end{pmatrix}.
    \end{align}
    Here, the red colour indicates which component corresponds to which column.
    For a numerical validation of this example, see Figure \ref{Figure_num_validation_ex_1}. For further examples where not only the identification of direct $Y$-to-$Y$-effects is of interest but also the identification of direct effects from other observed time series onto $Y$, we refer the reader to Section \ref{sec_examples} in the Appendix. Section \ref{sec_examples} in the Appendix also contains an example where $d_U>1$.
    \demo
    	\begin{figure}
    	\centering
		\includegraphics[scale=0.3]{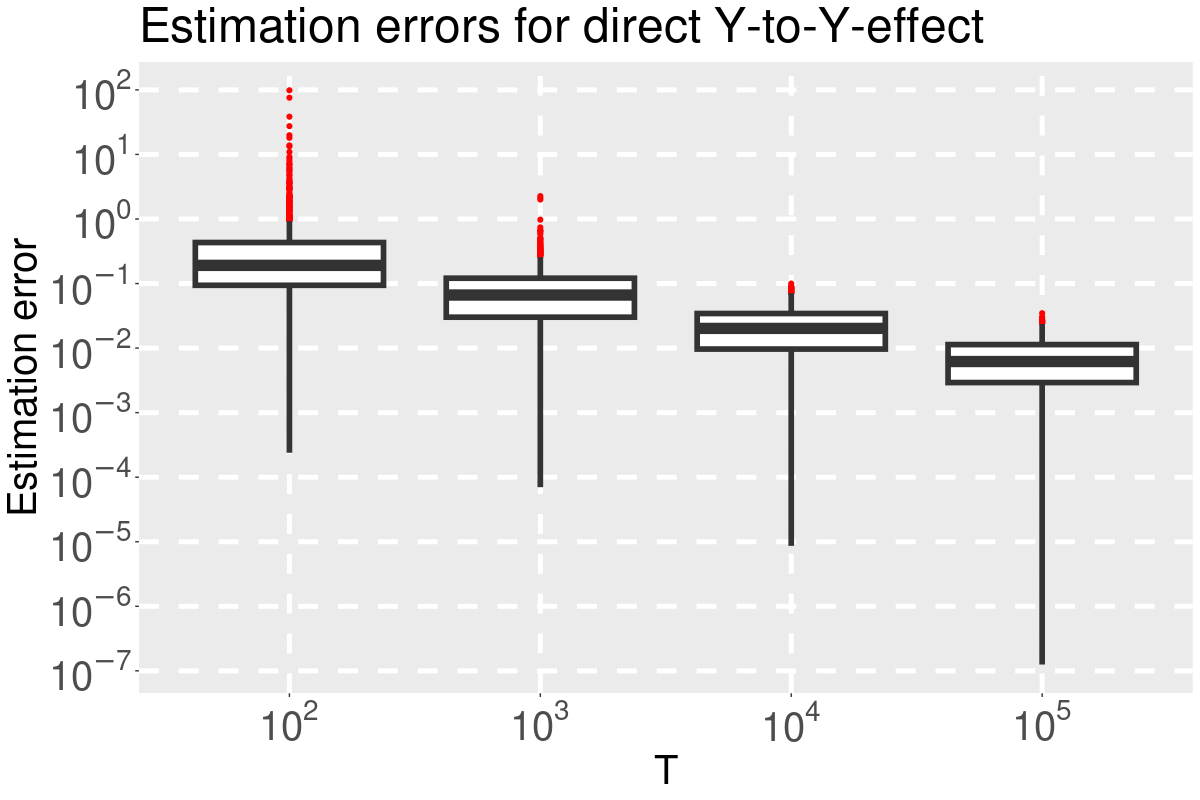}
		\caption{A numerical validation of Example \ref{example1}. For $1000$ different parameters that yield a stable SVAR process inducing the full time graph from Figure \ref{example_graph_1} and time series lengths $T\in\{10^2, 10^3, 10^4, 10^5\}$, we plot the error of the estimate of $A^{(3)}_{YY}$ from equation \eqref{est_ex_1} to the true $A^{(3)}_{YY}$. Remark: In these boxplots, the whisker's outside the boxes correspond to the smallest and largest points within the $1.5$-inner quartile range (calculated on the original scale). Outliers are highlighted in red. The ordinate axis is log10-transformed.}
		\label{Figure_num_validation_ex_1}
	\end{figure}
\end{example}
\addtocounter{theorem}{5}
\begin{proposition}[Consistent Estimation]
\label{proposition_consistency}
One can consistently estimate the solution $v$ in equation \eqref{main_lin_eq} in generic settings assuming that Theorem \ref{main_theorem} applies: For $T_0\in\mathbb{Z}$ and $T\in \mathbb{N}_{\geq 1}$, consider $S_{T_0+1},\ldots,S_{T_0+T}$. Estimate the covariances in $\Gamma_{R,C}$ and $\Gamma_{R,Y_t}$ by using the required entries from the sample autocovariance matrices
\begin{align*}
\hat{\Gamma}_{S_{t-h}S_t}:=\frac{1}{T-h}\sum_{j=h+1 }^TS_{T_0+j-h}S_{T_0+j}^T.
\end{align*}
Writing $\hat{\Gamma}_{R,C}$ and $\hat{\Gamma}_{R,Y_t}$ for the estimators of $\Gamma_{R,C}$ and $\Gamma_{R,Y_t}$, respectively, it then follows that $\hat{v}:=\hat{\Gamma}_{R,C}^{-1}\cdot \hat{\Gamma}_{R,Y_t}$ converges in probability to $v=\Gamma_{R,C}^{-1}\cdot \Gamma_{R,Y_t}$ in generic settings if Theorem \ref{main_theorem} applies.
\end{proposition}

\begin{remark}
\label{remark_generalization_main_theorem}
    In fact, it is possible to replace both $\textnormal{pa}^{\textnormal{obs}}(F^{\textnormal{obs}})$ and $\textnormal{pa}^{\textnormal{obs}}(Y_t)$ in Theorem \ref{main_theorem} by supersets just containing finitely many observed variables (we do the proof of Theorem \ref{main_theorem} in Section \ref{sec_proof_of_main_result} directly for this generalization). The coefficients in the solution $v$ corresponding to the induced columns of the additional variables in the superset of $\textnormal{pa}^{\textnormal{obs}}(Y_t)$ are then all zero.
    Similarly, one can replace $\textnormal{pa}^{\textnormal{lat}}(F^{\textnormal{obs}})$ and $\textnormal{pa}^{\textnormal{lat}}(Y_t)$ by supersets just containing finitely many latent variables (again, we do the proof of Theorem \ref{main_theorem} in Section \ref{sec_proof_of_main_result} directly for this generalization). Taking these supersets might be relevant when the full time graph is not entirely known or when one wants to derive estimators that are valid for multiple full time graphs at once (see Section \ref{real_world_example} for such valid-across-different-full-time-graphs-estimators). Assuming only supersets, however, typically makes it more difficult to satisfy the requirements from Theorem \ref{main_theorem}.
\end{remark}

\begin{remark}
Theorem \ref{main_theorem} has some similarities with Theorem 3.7 in \citet{barber2022half} which is an identification result for LSEMs (with finitely many variables): For example, conditions (i) and (ii) in their Theorem 3.7 resemble conditions 2a and 3 in Theorem \ref{main_theorem}, respectively. However, there are also notable differences: Theorem 3.7 in \citet{barber2022half} is recursive, meaning that to identify certain edge parameters other edge parameters need to have been previously identified---in the time series setting with autodependencies and infinitely many edges this recursive structure seems to be rather prohibitive. Moreover, as already mentioned in Section \ref{sec_intro}, \citet{barber2022half} assume that latent variables are source variables and independent, which we do not assume as this assumption is already too prohibitive when working with latent confounding time series that have autodependencies. Moreover, Theorem 3.7 of \citet{barber2022half} only applies to graphs with \emph{finitely} many vertices whereas we work with graphs containing \emph{infinitely} many vertices.
\end{remark}

\section{Sufficient lag-based criteria}
\label{sec_sufficient_lag_based_criteria}
Theorem \ref{main_theorem} is rather abstract and primarily graphical. As full time graphs can be entirely characterized by their lags, it is possible to derive lag-based sufficient (which are albeit not necessary) criteria that, when satisfied, imply that certain graphical parts of Theorem \ref{main_theorem} hold.  In particular, we will in this section introduce the following lag-based lemmas: Lemma \ref{lemma_resid_class1} states a sufficient condition under which for given sets $R$ and $C$ there exists a system of treks $\Psi:R\rightrightarrows C$ as required by Theorem \ref{main_theorem}. Building on this result, the subsequent lemma (Lemma \ref{lemma_resid_class_new}) then states a sufficient condition under which for a given $C$ a suitable $R$ with a system of treks $\Psi:R\rightrightarrows C$ as required by Theorem \ref{main_theorem} exists. Afterwards, in Lemma \ref{lemma_BU_Fobs_example}, we then look at the case $d_U=1$ and state possible sets $F^{\textnormal{obs}}$ and $B_U$ showing that they satisfy the requirements 1 and 2 from Theorem \ref{main_theorem}. Finally, in Corollary \ref{corollary_identifiability}, we connect Lemma \ref{lemma_resid_class_new} with Lemma \ref{lemma_BU_Fobs_example}, thus yielding a sufficient condition under which one gets identifiability for the case $d_U=1$. Because Lemmas \ref{lemma_resid_class1} and \ref{lemma_resid_class_new} are more broadly applicable than just for the case $d_U=1$ (see for example the full time graph from Example \ref{ex_app_3} in Section \ref{sec_examples} in the Appendix), we will explicitly state them.

In order to explain these lag-based criteria, we need to introduce  further concepts, notation and assumptions.
\begin{definition}
	\label{def_residue_class_new}
	Let $\{S^j_t\}_{t\in\mathbb{Z}}$ be any observed or unobserved component time series. We say that two vertices $S^j_{t_1}$ and $S^j_{t_2}$ are in the same $l^{S^j}_i$-residue class if they are connected by a path of just $l^{S^j}_i$-edges, so if $l^{S^j}_i$ divides the difference $t_1-t_2$.
	 If $m_{S^j}=0$, then we say that there are no residue classes (of any kind) for $\{S^j_t\}_{t\in\mathbb{Z}}$.
\end{definition}
\addtocounter{theorem}{-10}
\begin{example}[continued]
In the full time graph from Figure \ref{example_graph_1}, all $U$-vertices are in the same $l^U_1$-residue class. For the $Y$-vertices, there are $l^Y_1=3$-many $l^Y_1$-residue classes, namely $\{\ldots,\allowbreak Y_{t-3},\allowbreak Y_{t},\allowbreak Y_{t+3},\ldots\}$ and $\{\ldots,\allowbreak Y_{t-2},\allowbreak Y_{t+1}, \allowbreak Y_{t+4}, \allowbreak\ldots\}$ and $\{\ldots,\allowbreak Y_{t-1},\allowbreak Y_{t+2},\allowbreak Y_{t+5},\allowbreak \ldots\}$.\demo
\end{example}
\addtocounter{theorem}{9}

In order to introduce further notation, we now explain why $\textnormal{ForbAn}$ from Theorem \ref{main_theorem} only contains finitely many vertices: First, note that $\{Y_t\}\cup F^{\textnormal{obs}}$ and $B_U$ have only finitely many vertices. Second, note that because $B_U$ $d$-blocks every directed path from the latent vertices with time indices strictly smaller than $t_{\textnormal{inf}}(B_U)$ to $\textnormal{pa}^{\textnormal{lat}}(F^{\textnormal{obs}}\cup \{Y_t\})\setminus B_U$ due to requirements 1 and 2 in Theorem \ref{main_theorem}, it holds that every (latent) ancestor of  $\textnormal{pa}^{\textnormal{lat}}(F^{\textnormal{obs}}\cup \{Y_t\})\setminus B_U$ with time index strictly smaller than $t_{\textnormal{inf}}(B_U)$ is also a (latent) ancestor of $B_U$, because otherwise there would be some directed path from a latent vertex with time index strictly smaller than $t_{\textnormal{inf}}(B_U)$ to $\textnormal{pa}^{\textnormal{lat}}(F^{\textnormal{obs}}\cup \{Y_t\})\setminus B_U$ that is not $d$-blocked by $B_U$. Hence, $\textnormal{ForbAn}$ only contains finitely many vertices.

Because $\textnormal{ForbAn}$ contains only finitely many vertices, it follows that $\textnormal{ForbAn}$ has a finite minimal time index and, thus, it is for each $O^i$ possible to define a (potentially non-unique) $\tau_{O^i}\in\mathbb{Z}$ such that every $O^i$-vertex with time index $t'$ satisfying $t'\leq t-\tau_{O^i}$ is not a descendant of $\textnormal{ForbAn}$.

With this notation and further using the abbreviations $C^{(1)}:=C\setminus F^{\textnormal{obs}}$ and $C^{(2)}:=C\cap F^{\textnormal{obs}}=F^{\textnormal{obs}}$, we are now ready to state the following assumptions which we either all or partially use for the remainder of this section:
   \begin{itemize}
   \item \textbf{(C1):} For a given partition $R=R^{(1)}\dot{\cup} R^{(2)}$, it holds that $|R^{(1)}_{O^i}|= |C^{(1)}_{O^i}|$ and $|R^{(2)}_{O^i}| = |C^{(2)}_{O^i}|$ for each observed component time series $\{O^i_t\}_{t\in \mathbb{Z}}$;
       \item \textbf{(C2):} For each observed component time series $\{O^i_t\}_{t\in \mathbb{Z}}$: If $C^{(1)}_{O^i}\neq \emptyset$, then there exists a lag $l^{O^i}_{j_i}$ (so, in particular, $m_{O^i}>0$) such that all vertices in $C^{(1)}_{O^i}$ with time index in $[t-\tau_{O^i}-(l^{O^i}_{j_i}-1),\infty)$ are in different $l^{O^i}_{j_i}$-residue classes.
        \item \textbf{(C3):} For a given partition $R=R^{(1)}\dot{\cup} R^{(2)}$, it holds for each observed component time series $\{O^i_t\}_{t\in \mathbb{Z}}$ that if $C^{(1)}_{O^i}\neq \emptyset$, then for every $c_i\in C^{(1)}_{O^i}$ with time index in $[t-\tau_{O^i}-(l^{O^i}_{j_i}-1),\infty)$ there exists exactly one $r_i\in R^{(1)}_{O^i}$ with time index in $[t-\tau_{O^i}-(l^{O^i}_{j_i}-1),t-\tau_{O^i}]$ from the same $l^{O^i}_{j_i}$-residue class.
        \item \textbf{(C4):} For a given partition $R=R^{(1)}\dot{\cup} R^{(2)}$, it holds for each observed component time series $\{O^i_t\}_{t\in \mathbb{Z}}$ that if $C^{(1)}_{O^i}\neq \emptyset$, then for every $c_i\in C^{(1)}_{O^i}$ with time index in $(-\infty, t-\tau_{O^i}-l^{O^i}_{j_i}]$ it holds that $c_i\in R^{(1)}_{O^i}$.
       \item \textbf{(C5):} For each observed component time series $\{O^i_t\}_{t\in \mathbb{Z}}$ it holds that
       \begin{itemize}
            \item \textbf{(C5.1):} if $C^{(1)}_{O^i}\neq \emptyset$ and $C^{(2)}_{O^i}\neq \emptyset$, then $t_{\textnormal{sup}}(C^{(1)}_{O^i})<t_{\textnormal{inf}}(C^{(2)}_{O^i})$; and
           \item \textbf{(C5.2):} in \emph{addition} to (C5.1), for a given partition $R=R^{(1)}\dot{\cup} R^{(2)}$, if $R^{(1)}_{O^i}\neq \emptyset$ and $R^{(2)}_{O^i} \neq \emptyset$, then $t_{\textnormal{sup}}(R^{(2)}_{O^i})<t_{\textnormal{inf}}(R^{(1)}_{O^i})$.
       \end{itemize}
   \item \textbf{(C6):} For each $\{O^i_t\}_{t\in \mathbb{Z}}$ for which $C^{(2)}_{O^i}\neq \emptyset$, there exists an $l^{O^iU^{k_i}}_{w_i}$-edge for some $k_i\in[d_U]_1$ linking the latent variables with $\{O^i_t\}_{t\in \mathbb{Z}}$.
   \begin{itemize}
       \item \textbf{(C6.1):} Write $P:=\bigcup_{i\in [d_O]_1:\;C^{(2)}_{O^i}\neq \emptyset} \textnormal{pa}(C^{(2)}_{O^i}, l^{O^iU^{k_i}}_{w_i})$. Assume that for all $\{U^k_t\}_{t\in\mathbb{Z}}$ for which $P_{U^k}\neq \emptyset$ there exists a lag $l^{U^k}_{j_k}$ (so in particular, $m_{U^k}>0$) such that all $p_k\in P_{U^k}$ are in different $l^{U^k}_{j_k}$-residue classes;
       \item \textbf{(C6.2):} \emph{In addition} to assuming Assumption (C6.1): For a given partition $R=R^{(1)}\dot{\cup} R^{(2)}$, write $Q:=\bigcup_{i\in [d_O]_1:\;C^{(2)}_{O^i}\neq \emptyset} \textnormal{pa}(R^{(2)}_{O^i}, l^{O^iU^{k_i}}_{w_i})$ and further assume that for each $p_k\in P_{U^k}$ there exists exactly one $q_k\in Q_{U^k}$ from the same $l^{U^k}_{j_k}$-residue class.
   \end{itemize}
\end{itemize}
We are now ready to state the first lemma.

\begin{lemma}
\label{lemma_resid_class1}
Assume the assumptions from Theorem \ref{main_theorem} except for the existence of a system of treks $\Psi:R\rightrightarrows C$ with no sided intersection and unique monomial among all other systems of treks $\Pi:R\rightrightarrows C$ with no sided intersection (subpoint 5 in Theorem \ref{main_theorem}). Assume that $R$ can be decomposed into $R=R^{(1)}\dot{\cup}R^{(2)}$ such that (C1)--(C6.2) hold (note that for Assumption (C6.2) an $l^{O^iU^{k_i}}_{w_i}$-edge exists because of the assumed existence of $\Upsilon$ in Theorem \ref{main_theorem} and the fact that each directed path in $\Upsilon$ has as last edge an edge linking a latent variable and $C^{(2)}=F^{\textnormal{obs}}$). Then, subpoint 5 in Theorem \ref{main_theorem} is satisfied. 
\end{lemma}
\addtocounter{theorem}{-11}
\begin{example}[continued]
Lemma \ref{lemma_resid_class1} also applies to Example \ref{example1} which corresponded to Figure \ref{example_graph_1}: First note that $\tau_{Y}=1$ is valid because every $Y_{t'}$ with $t'\leq t-\tau_Y=t-1$ is not a descendant of $\textnormal{ForbAn}=\{U_t, U_{t+1}, \ldots, U_{t+l^Y_1}, Y_t, Y_{t+l^Y_1+1}\}$. Partitioning $R=\{Y_{t-l^Y_1}, Y_{t-l^Y_1+1}, Y_{t-l^Y_1-1}\}$ into $R^{(1)}:=\{Y_{t-l^Y_1}, Y_{t-l^Y_1+1}\}$ and $R^{(2)}:=\{Y_{t-l^Y_1-1}\}$ and using the fact that $l^{Y}_1=l^Y_{m_Y}>1$ and $C=\{Y_{t-l^Y_1}, Y_{t+l^Y_1+1}, Y_{t+1}\}$ and $C^{(1)}=\{Y_{t-l^Y_1}, Y_{t+1}\}$ and $C^{(2)}=\{Y_{t+l^Y_1+1}\}$, it follows that (C1)--(C5.2) hold. For (C6.2): Fix as the edge linking latent and observed variables an $l^{YU}_1$-edge. As $P=\{U_t\}$ and $Q=\{U_{t-l^Y_1-2}\}$ and as $l^U_1=1$, clearly (C6.2) holds. Thus, Lemma \ref{lemma_resid_class1} applies. \demo
\end{example}
\addtocounter{theorem}{10}
\begin{remark}
Lemma \ref{lemma_resid_class1} also applies to Examples \ref{ex_app_1}, \ref{ex_app_2} and \ref{ex_app_3} in Section \ref{sec_examples} of the Appendix.
\end{remark}
Lemma \ref{lemma_resid_class1} requires an a priori fixed set $R$ satisfying several assumptions. As the following lemma shows, it is in fact possible to a priori make no assumptions on the set $R$ but still yielding a similar result as Lemma \ref{lemma_resid_class1}.
\begin{lemma}
\label{lemma_resid_class_new}
Assume the assumptions from Theorem \ref{main_theorem} except for the existence of a set $R$ with $|R|=|C|$ satisfying subpoint 4 from Theorem \ref{main_theorem} and a system of treks $\Psi:R\rightrightarrows C$ with no sided intersection and unique monomial among all other systems of treks $\Pi:R\rightrightarrows C$ with no sided intersection (subpoint 5 in Theorem \ref{main_theorem}). Furthermore, assume that Assumptions (C2), (C5.1) and (C6.1) hold (note that for Assumption (C6.1) an $l^{O^iU^{k_i}}_{w_i}$-edge exists because of the assumed existence of $\Upsilon$ in Theorem \ref{main_theorem} and the fact that each directed path in $\Upsilon$ has as last edge an edge linking a latent variable and $C^{(2)}=F^{\textnormal{obs}}$).
Then, there exists a set $R$ with $|R|=|C|$ satisfying requirement 4 from Theorem \ref{main_theorem} and  a system of treks $\Psi:R\rightrightarrows C$ as required by subpoint 5 in Theorem \ref{main_theorem}.
\end{lemma}
\addtocounter{theorem}{-13}
\begin{example}[continued]
From the exact same justification on why Assumptions (C2), (C5.1) and (C6.1) hold for Lemma \ref{lemma_resid_class1}, it also follows that Lemma \ref{lemma_resid_class_new} applies to Example \ref{example1}. \demo
\end{example}
\addtocounter{theorem}{12}
\begin{remark}
Lemma \ref{lemma_resid_class_new} also applies to Examples \ref{ex_app_1}, \ref{ex_app_2} and \ref{ex_app_3} in Section \ref{sec_examples} of the Appendix.
\end{remark}
For the case where $d_U=1$ and $m_U>0$ and for which there exists some $O^{i_0}$ for which there are edges from $U$ to $O^{i_0}$ (if there is no such $O^{i_0}$, then there would be no latent confounding in the first place), one can further specify choices of $B_U$ and $F^{\textnormal{obs}}$ . Namely, one can, for example, choose $F^{\textnormal{obs}}=\{O^{i_0}_{t+\Delta},\allowbreak O^{i_0}_{t+\Delta+1},\allowbreak\ldots,\allowbreak O^{i_0}_{t+\Delta+(l^{U}_{m_{U}}-1)}\}$ where $\Delta\in \mathbb{Z}$ and
$B_U:=\{U_t:\; t\in [t_0,t_0+(l^U_{m_U}-1)]\}$
where $t_0:=t_{\textnormal{inf}}(\textnormal{pa}^{\textnormal{lat}}(Y_t)\cup \textnormal{pa}^{\textnormal{lat}}(F^{\textnormal{obs}}))$. In particular, for these choices of $F^{obs}$ and $B_U$, subpoints 1 and 2 of Theorem \ref{main_theorem} are satisfied as the following lemma asserts.
\begin{lemma}
\label{lemma_BU_Fobs_example}
    Assume a stable SVAR process satisfying Assumptions \ref{assumption_no_instantaneous_self_edges}, \ref{assumption_acyclicity} and \ref{assumption_latents_have_no_observed_parents} from Section \ref{sec_var_processes}. Furthermore, assume $d_U=1$ and $m_U>0$ and the existence of some $O^{i_0}$ for which there are edges from $U$ to $O^{i_0}$. Let $F^{\textnormal{obs}}=\{O^{i_0}_{t+\Delta},\allowbreak O^{i_0}_{t+\Delta+1},\allowbreak\ldots,\allowbreak O^{i_0}_{t+\Delta+(l^U_{m_U}-1)}\}$ for an arbitrary but fixed $\Delta\in\mathbb{Z}$, and let $B_U:=\{U_t:\; t\in [t_0,t_0+(l^U_{m_U}-1)]\}$ where $t_0:=t_{\textnormal{inf}}(\textnormal{pa}^{\textnormal{lat}}(Y_t)\cup \textnormal{pa}^{\textnormal{lat}}(F^{\textnormal{obs}}))$. Then, subpoints 1 and 2 from Theorem \ref{main_theorem} hold.
\end{lemma}
\addtocounter{theorem}{-15}
\begin{example}[continued]
The choice of $F^{\textnormal{obs}}=\{Y_{t+l^Y_1+1}\}$ implies that $t_0=t-1$. Because $B_U=\{U_{t-1}\}$, it follows that $F^{\textnormal{obs}}$ and $B_U$ have the form as required by Lemma \ref{lemma_BU_Fobs_example}, and thus, subpoints 1 and 2 from Theorem  \ref{main_theorem} hold.\demo
\end{example}
\addtocounter{theorem}{14}
\begin{remark}
Lemma \ref{lemma_BU_Fobs_example} also applies to Examples \ref{ex_app_1} and \ref{ex_app_2} but not to Example \ref{ex_app_3} in Section \ref{sec_examples} of the Appendix.
\end{remark}
Combining Lemmas \ref{lemma_resid_class_new} and \ref{lemma_BU_Fobs_example} and using the fact that the $U$-parents of $F^{\textnormal{obs}}$ as specified in Lemma \ref{lemma_BU_Fobs_example} are in different $l^U_{m_U}$-residue classes for any fixed edge linking $O^{i_0}$-variables and $F^{\textnormal{obs}}$, yields the following corollary.

\begin{corollary}
\label{corollary_identifiability}
Assume a stable SVAR process satisfying Assumptions \ref{assumption_no_instantaneous_self_edges}, \ref{assumption_acyclicity} and \ref{assumption_latents_have_no_observed_parents} from Section \ref{sec_var_processes}. Furthermore, assume $d_U=1$ and $m_U>0$ and the existence of at least one $O^{i_0}$ for which there are edges from $U$ to $O^{i_0}$. Let $F^{\textnormal{obs}}:=\{O^{i_0}_{t+\Delta},O^{i_0}_{t+\Delta+1},\ldots, O^{i_0}_{t+\Delta+(l^U_{m_U}-1)}\}$ for some $\Delta\in \mathbb{Z}$ and $B_U:=\{U_t:\; t\in [t_0,t_0+(l^U_{m_U}-1)]\}$ where $t_0:=t_{\textnormal{inf}}(\textnormal{pa}^{\textnormal{lat}}(Y_t)\cup \textnormal{pa}^{\textnormal{lat}}(F^{\textnormal{obs}}))$. Moreover, assume Assumptions (C2) and (C5.1).
If additionally 
\begin{align*}
\bigr(F^{\textnormal{obs}}\cup \textnormal{pa}^{\textnormal{obs}}(F^{\textnormal{obs}})\bigr)\cap \bigr(\{Y_t\}\cup\textnormal{pa}^{\textnormal{obs}}(Y_t)\bigr)=\emptyset,
\end{align*}
then there exists a set $R$ such that all requirements from Theorem \ref{main_theorem} hold.
\end{corollary}
\addtocounter{theorem}{-17}
\begin{example}[continued]
Corollary \ref{corollary_identifiability} also applies to Example \ref{example1}---this fact follows from the same justifications why Lemmas \ref{lemma_resid_class1} and Lemma \ref{lemma_BU_Fobs_example} apply. Therefore, Theorem \ref{main_theorem} applies to this example.\demo
\end{example}
\addtocounter{theorem}{16}
\begin{remark}
Corollary \ref{corollary_identifiability} also applies to Examples \ref{ex_app_1} and \ref{ex_app_2} but not to Example \ref{ex_app_3} in Section \ref{sec_examples} of the Appendix.
\end{remark}
For the case $d_U=1$ and $m_U>0$ and the existence of at least one $O^{i_0}$ for which there are edges from $U$ to $O^{i_0}$, Corollary \ref{corollary_identifiability} thus gives rise to a simple brute-force algorithm, where one for different $\Delta$'s
checks the requirements from Corollary \ref{corollary_identifiability}. Note that checking for \emph{different} $\Delta$'s might yield more identifiability because the observed parents of $F^{\textnormal{obs}}$ change residue-classes with changing $\Delta$, the observed parents of $Y_t$, however, do not change their residue class. Therefore, Assumption (C2) might be satisfied for some $\Delta$'s while not being satisfied for other $\Delta$'s. Similarly, the empty intersection requirement from Corollary \ref{corollary_identifiability} might also be satisfied for some $\Delta$'s but not for other $\Delta$'s.
\section{Synthetic numerical simulations}
\label{sec_numerical_experiments}
	Besides the numerical validation of Example \ref{example1} in Figure \ref{Figure_num_validation_ex_1} (and the similar numerical validations of Examples \ref{ex_app_1}, \ref{ex_app_2} and \ref{ex_app_3} in Section \ref{sec_examples} of the Appendix), we provide further (synthetic) numerical simulations.
	In particular, in this section, we empirically validate Corollary \ref{corollary_identifiability} by showing convergence of the estimated direct to the true direct causal effects for randomly drawn full time graphs and corresponding SVAR processes for the case $d_U=1$ and $d_O=2$.
	In particular, and also taking some inspiration from the simulations in \citet{thams2022identifying}, we evaluate convergence as follows:

	\textbf{Data generating process:} 
Writing $\{X_t\}_{t\in \mathbb{Z}}=\{O^2_t\}_{t\in\mathbb{Z}}$, we generate stable SVAR processes as follows: We uniformly at random draw each $m_{U}$, $m_{X}$, $m_Y$, $m_{XU}$, $m_{YU}$ from $[5]_1$ and set $m_{YX}=1$ and $m_{XY}=0$. We then draw the lags $l^{S^iS^j}_{k}$ for $i\neq j$ from $[5]_0$ and the lags $l^{S^i}_k$ from $[5]_1$ without replacement such that the number of drawn elements equals $m_{S^iS^j}$. For these drawn lags, we then check whether the resulting full time graph is acyclic. We then also check for different $\Delta\in\mathbb{Z}$'s whether Corollary \ref{corollary_identifiability} yields identifiability. 
	
	If the full time graph induced by the previously drawn lags is identifiable for some $\Delta$ via Corollary \ref{corollary_identifiability}, then we draw $10$ different sets of parameters for SVAR processes corresponding to this full time graph. We hereby draw the non-zero entries of the matrices $A^{(0)},\ldots,A^{(p)}$ uniformly at random from $(-0.9,-0.1)\cup(0.1,0.9)$ and we set the covariance matrix $\Sigma$ to the identity matrix $I_3$.  We then check whether the generated SVAR processes for the previously drawn full time graph are stable with enough margin (that is, the maximal modulus of the eigenvalues of the matrix $\mathbf{B}$ in equation \eqref{big_A} in Section \ref{appendix_proof_of_trek_rule} of the Appendix is less or equal than $0.9$). If a generated SVAR process is not stable with enough margin, then we repeat until it is (if these repetitions for all $10$ drawn SVAR processes need in total to be repeated more than $10^5$-times, we start with a new full time graph). 
	
	For a given stable SVAR process, we then generate time series for the different time series lengths $T\in\{10^2, 10^3, 10^4, 10^5, 10^6\}$ by drawing $T+1000$ noise variables from a multivariate normal distribution with the just mentioned covariance matrix $\Sigma=\textnormal{I}_3$ and then repeatedly applying the SVAR equations starting from the initial noise variables (the first $1000$-drawn values are seen as a burnin and subsequently discarded).

	\textbf{Evaluation of estimators:} We generate $1000$ such full time graphs and for each full time graph $10$ stable SVAR processes. For each of these stable SVAR processes, we generate one time series for each length $T\in\{10^2, 10^3, 10^4, 10^5, 10^6\}$. For each drawn time series we calculate the absolute difference of the estimated $X$-to-$Y$-effect to the true direct $X$-to-$Y$-effect. We then take the median out of these $10$ absolute errors. We do this calculation for all $1000$ drawn full time graphs and present these median values in a boxplot for each respective time series length. Figure \ref{Figure_Boxplots} shows the results.
	
	\textbf{Results:}
	Our code was set up such that a non-invertible $\hat{\Gamma}_{R,C}$ yields an error which, however, has not happened. Moreover, one can see that except for a few outliers the errors are typically rather small and decrease with increasing time series length. Also note that the number of points with a median estimation error strictly larger than $10^0$ is (stated from $10^2$ to $10^6$): $151$, $81$, $62$, $44$, $23$.
	
	\begin{figure}
	\centering
		\includegraphics[scale=0.3]{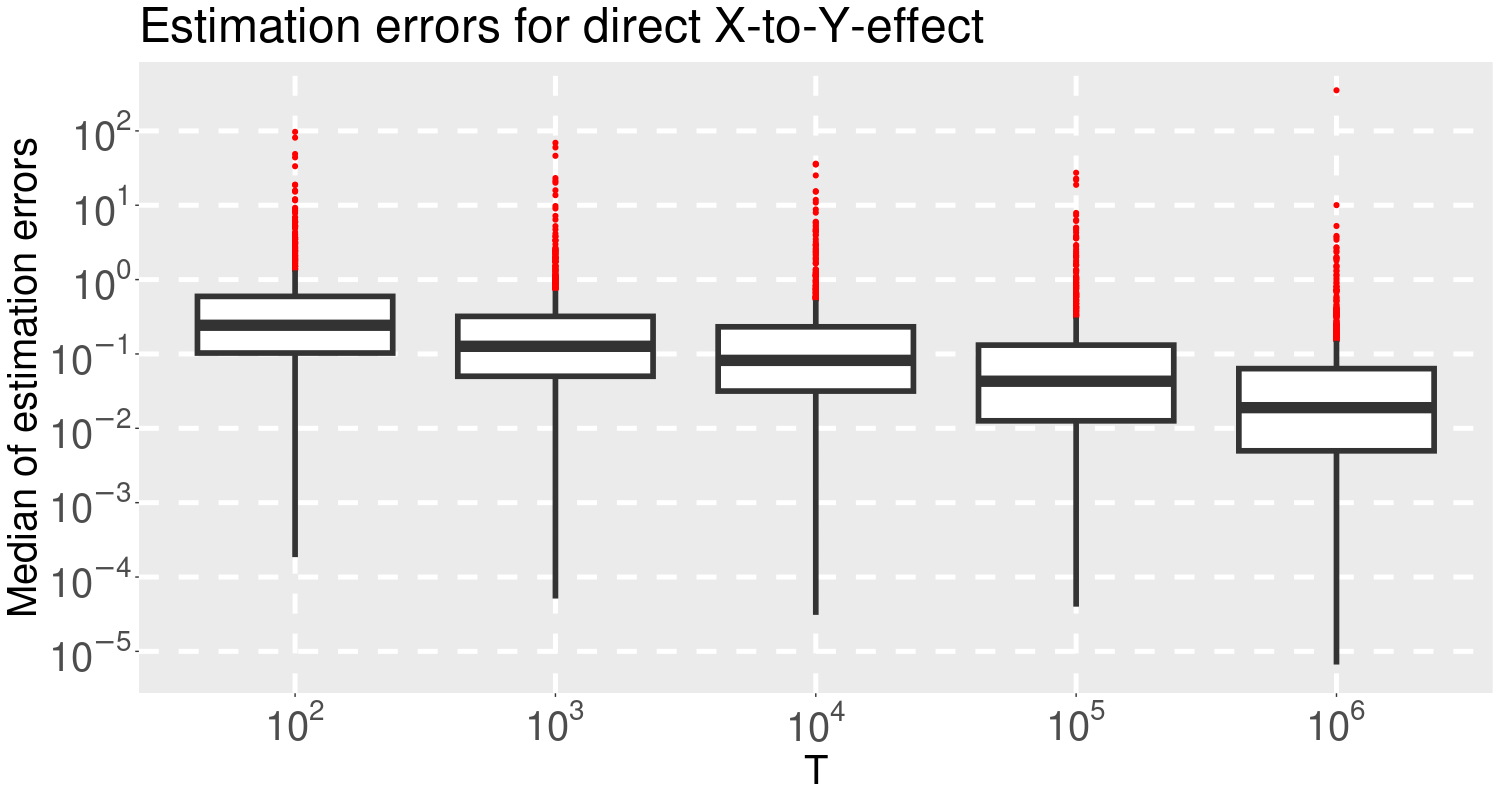}
		\caption{Boxplots for numerical experiments from Section \ref{sec_numerical_experiments}. Remark: In these boxplots, the whisker's outside the boxes correspond to the smallest and largest points within the $1.5$-inner quartile range (calculated on the original scale). Outliers are highlighted in red. The ordinate axis is log10-transformed. The number of points with a median estimation error strictly larger than $10^0$ is (stated from $10^2$ to $10^6$): $151$, $81$, $62$, $44$, $23$.}
		\label{Figure_Boxplots}
	\end{figure}

\section{Real-world example}
\label{real_world_example}
	In this section, we apply Theorem \ref{main_theorem} and Remark \ref{remark_generalization_main_theorem} to the recent electricity market example from \citet{tiedemann2024identifying}. The goal for this example was to identify the instantaneous price elasticity $\beta^P$ in the German-Luxemburgian(-Austrian)\footnote{On October 1st, 2018, the German-Luxemburgian-Austrian market was split into the German-Luxemburgian and Austrian market, and we will from this point, similarly to \citet{tiedemann2024identifying}, just consider the German-Luxemburgian market.} electricity market, that is, the direct causal effect of the electricity price $P$ on the demand of the electricity $D$. This direct causal effect is (assumed to be) latently confounded. To identify $\beta^P$, \citet{tiedemann2024identifying} make use of the instrumental time series approach from \citet{thams2022identifying}. More specifically, \citet{tiedemann2024identifying} use total wind generation $W$ (based on both onshore and offshore wind farms) as an instrumental time series (which implies that total wind generation is observed). We, in contrast, pretend that $W$ is unobserved and only require some qualitative knowledge about $W$ which, however, is rather minimal. 
	
	\textbf{Models:}
	\citet{tiedemann2024identifying} suggest three different models. In all three models, $W$ is assumed to be a (linear) autoregressive (AR) process of some finite order.
	
	For \emph{model 1}, the demand equation and price equation, respectively, are
	\begin{align*}
	D_t &:= D_0 + \beta^PP_t + \beta^{D_1}D_{t-1} + U^D_t,\text{ and}\\
		P_t &:= \frac{S_0-D_0}{\beta^P-\gamma^P} + \frac{\gamma^W}{\beta^P-\gamma^P}W_t-\frac{\beta^{D_1}}{\beta^P-\gamma^P}D_{t-1}+\frac{U^S_t-U^D_t}{\beta^P-\gamma^P}.
	\end{align*}
		To reflect the possibility that only part of the demand is price-sensitive, \citet{tiedemann2024identifying} also suggest \emph{model 2} given by 
	\begin{align*}
	D_t&:=D_0+\beta^P P_t + \beta^{B_1}B_{t-1}+\underbrace{U^A_t + U^B_t}_{=:U^D_t},\text{ and}\\
	P_t &:= \frac{S_0-D_0}{\beta^P-\gamma^P} + \frac{\gamma^W}{\beta^P-\gamma^P}W_t+\frac{\beta^{B_1}}{\beta^P-\gamma^P}B_{t-1}+\frac{U^S_t-U^D_t}{\beta^P-\gamma^P},
	\end{align*}
	where $B_{t}:=B_0+\beta^{B_1}B_{t-1}+U^B_t$.
	
	Finally, to incorporate the possibility that consumers react to price changes by shifting their demand over time, \citet{tiedemann2024identifying} propose \emph{model 3} given by
		\begin{align*}
	D_t &:= D_0 + \beta^PP_t + \beta^{P1}P_{t-1} + \beta^{D_1}D_{t-1} + U^D_t,\text{ and}\\
	P_t &:= \frac{S_0-D_0}{\beta^P-\gamma^P} + \frac{\gamma^W}{\beta^P-\gamma^P}W_t-\frac{\beta^{P1}}{\beta^P-\gamma^P}P_{t-1}+\frac{U^S_t-U^D_t}{\beta^P-\gamma^P}.
	\end{align*}

	In all these models, $D_0$ and $S_0$ are constants and the $U^A_t$-, $U^B_t$-, $U^D_t$- and  $U^S_t$-noise-variables are assumed to be i.i.d.\ over time and independent of each other (except for model 2 and $U^D_t$ which is given by the sum of $U^A_t$ and $U^B_t$).\footnote{We have not explicitly found the independent-of-each-other assumption in \citet{tiedemann2024identifying}. However, they use Theorems 5 and 8 of \citet{thams2022identifying} that require it, and also the semi-synthetic simulations seem to assume it.} Furthermore, it is assumed that $\beta^P-\gamma^P\neq 0$.
	
	In models 1 and 3, latent confounding is due the occurence of $U^D_t$ in both the price equation and demand equation. Thus, as $U^D_t$ is assumed to be i.i.d. over time, in models 1 and 3 the latent confounder does not have autolag. In model 2, the latent confounding is both due to $U^D_t$ and $B_t$. Hence, in model 2  the latent confounding has two components where one component has no autolag and the other has an autolag of $1$.
	
 In Table \ref{table_validity}, we present several estimators and state for which models they are valid and for which they are not. This list is of course not complete and several other not presented estimators are valid as well.

		\begin{table}
    \centering
    \small
    \begin{tabular}{ c c c c c}
    Estimator & \makecell{Valid for\\ Model 1} & \makecell{Valid for\\ Model 2} & \makecell{Valid for\\ Model 3} & \makecell{Further required \\ Assumptions}\\
    \hline
   \makecell{$R=\{P_{t-1},P_{t-2}\}$,\\ $C=\{P_t, D_{t-1}\}$} & \checkmark &  \xmark & \xmark & \makecell{Existence of some lag $l^W_i$}\\
   \hline
   \makecell{$R=\{P_{t-1},P_{t-2}\}$, \\ $C=\{P_{t-1},P_{t}\}$}  &  \xmark &  \xmark &\checkmark & \makecell{Existence of some lag $l^W_i$}\\
   \hline
   \makecell{$R=\{P_{t-1}, P_{t-2},P_{t-k\cdot l^W_i}\}$, \\ $C=\{P_{t-1},P_t, D_{t-1}\}$}  & \checkmark &  \xmark & \checkmark & \makecell{Knowledge and existence\\ of some lag $l^W_i$,\\ $k\cdot l^W_i>2$}\\
    \hline
    \makecell{$R=\{P_{t-1},P_{t-k\cdot l^W_i+1}, P_{t-k\cdot l^W_i}\}$, \\ $C=\{P_{t}, D_{t+1}, P_{t+1}\}$} &  \xmark & \checkmark &  \xmark & \makecell{Knowledge and existence\\ of some lag $l^W_i\geq 2$,\\ $k\cdot l^W_i>2$}\\
   \hline
       \makecell{$R=\{P_{t-1}, P_{t-2},P_{t-3}, P_{t- l^W_i}, P_{t- l^W_i-1}\}$, \\ $C=\{P_{t-1},D_{t-1},P_t,D_{t+l^W_i-2}, P_{t+l^W_i-2}\}$} & \checkmark & \checkmark & \checkmark & \makecell{Knowledge and existence\\ of some lag $l^W_i\geq 4$}\\
  \hline 
  \end{tabular}
  \caption{Overview of different estimators based on Theorem \ref{main_theorem} and Remark \ref{remark_generalization_main_theorem} and their validity for the different models in Section \ref{real_world_example}. Here, $k$ can be any natural number that satisfies the respective constraint in the table. (For the later experiments, we take $l^W_i=1$ and $k = 3$ for the (from the top) third estimator, $l^W_i=2$ and $k=2$ for the fourth estimator, and $l^W_i=4$ for the fifth estimator).}
  \label{table_validity}
  \end{table}
	
	\textbf{Results:}
	We provide results for both semi-synthetic and real data. Our main data source is \citet{Bundesnetzagentur}. The results for the semi-synthetic data are intended to numerically illustrate the validity of the proposed estimators.
	For the semi-synthetic data, we set the true instantaneous price elasticity $\beta^P$ to $-100$ MW/(EUR/MWh). The number of considered time points for both the semi-synthetic and real data is 27072. Further details about the simulation procedure and data sources are deferred to Section \ref{sec_app_real_world_ex} in the Appendix. Note that due to simplicity and accessibility related issues, our data set moderately differs from those of \citet{tiedemann2024identifying}. For better comparison, we thus also present the results for the estimators \#1,\ldots,\#8 from \citet{tiedemann2024identifying} (for the semi-synthetic data only for models for which these estimators were shown to be valid). All results can be found in Table \ref{tabel_real_world}.

	One can see that all estimators except the one from us valid for all three models (and there only for models 1 and 2) yield estimates very close to the true semi-synthetic value of $-100$ MW/(EUR/MWh). The estimator valid for all models applied to models 1 and 2 yields significantly larger confidence intervals for the semi-synthetic data, however, these intervals still look reasonable and somewhat symmetric around the true value of $-100$ MW/(EUR/MWh). It also seems that the larger the cardinalities of $R$ and $C$, the larger the confidence intervals.
	
	For the real-data experiments, 
	our results seem only roughly consistent with the ones from \citet{tiedemann2024identifying}. In particular, the orders of magnitude of the effects ($\approx 10^1$--$10^3$) are typically the same. However, several of our estimators yield positive price elasticities $\hat{\beta}^{P}$, which the estimators from \citet{tiedemann2024identifying} applied to our dataset do not yield. Moreover, our results seem to be a little bit more heavy-tailed than the estimators from \citet{tiedemann2024identifying}. (In particular, the real-data estimate of our estimator valid for just model 2 is 712.99 and of our estimator valid for all three models is 1875.08 and thus both fall out of the range shown in Table \ref{tabel_real_world}).
	
	For our estimators (and to some degree also for the ones from \citet{tiedemann2024identifying}), we observe that estimates differ across different models and that estimators valid for the same model also show rather different results. These facts together with the good semi-synthetic results for both our estimators (and the ones from \citet{tiedemann2024identifying}) suggest  that the underlying models are not fitting perfectly to the data in the first place and can be improved upon.
	
	\begin{table}
    \small
    \begin{tabular}{ c c c c c}
    Estimator & Model  & $\hat{\beta}^P$ (semi-synthetic)  & $\hat{\beta}^P$ (real)  \\
    \hline
   \makecell{\color{white}...\color{black}\\$R=\{P_{t-1},P_{t-2}\}$,\\ $C=\{P_t, D_{t-1}\}$} & \makecell{\color{white}...\color{black}\\1} & \makecell{\color{white}...\color{black}\\\color{white}...\color{black}\\$[-100.19, -99.80]$\\\color{white}...\color{black}} & \multirow{13}{*}{\includegraphics[scale=0.3]{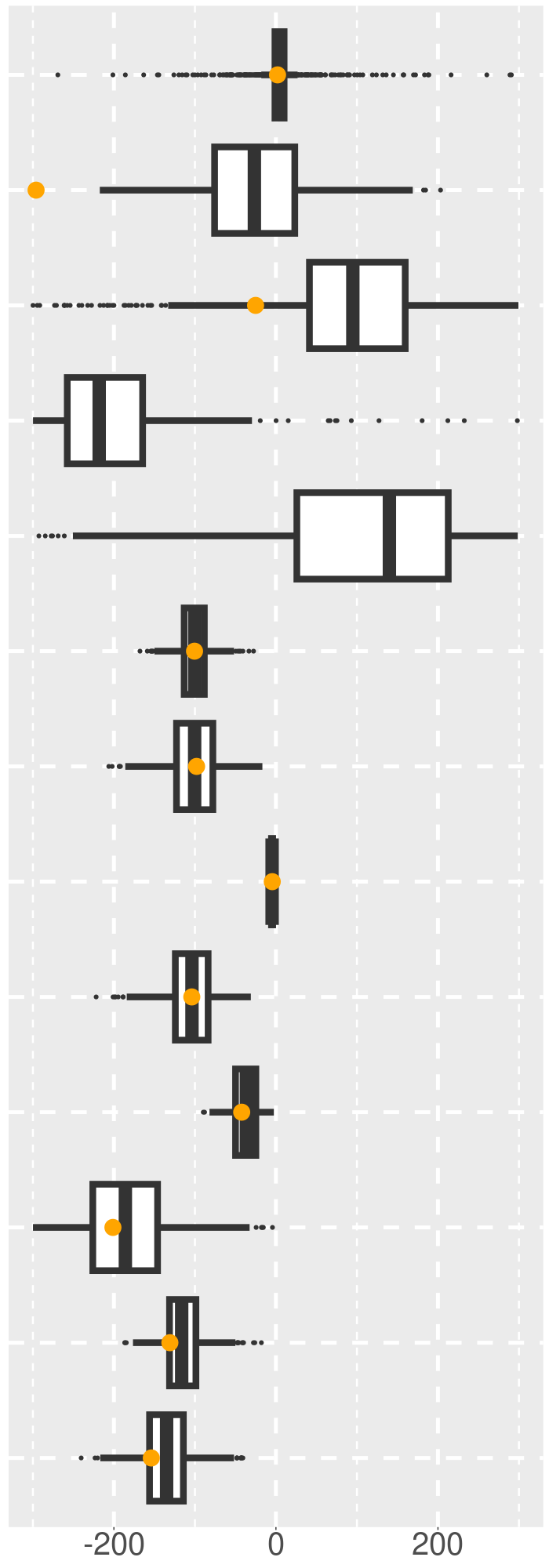}}\\\\[0.1cm]
   \cline{1-3} 
   \makecell{$R=\{P_{t-1},P_{t-2}\}$, \\ $C=\{P_{t-1},P_{t}\}$} & 3 & \makecell{\color{white}...\color{black}\\$[-100.00,-100.00]$\\\color{white}...\color{black}} &  \\
   \cline{1-3} \\[-0.2cm]
   \makecell{$R=\{P_{t-1}, P_{t-2},P_{t-3}\}$, \\ $C=\{P_{t-1},P_t, D_{t-1}\}$} & \makecell{1\\ 3} & \makecell{$[-101.24,-98.84]$\\ $[-100.00,-100.00]$}  &  \\\\[-0.2cm]
    \cline{1-3} 
    \makecell{$R=\{P_{t-1},P_{t-3},P_{t-4}\}$, \\ $C=\{P_{t}, D_{t+1}, P_{t+1}\}$} & 2 & \makecell{\color{white}...\color{black}\\$[-105.73,-94.72]$\\\color{white}...\color{black}} &  \\
   \cline{1-3}
       \makecell{$R=\{P_{t-1}, P_{t-2},P_{t-3},$\\ $P_{t- 4}, P_{t- 5}\}$, $C=\{P_{t-1},D_{t-1}$,\\ $P_t,D_{t+2}, P_{t+2}\}$} & \makecell{1\\ 2 \\ 3} & \makecell{$[-266.42,69.93]$\\ $[-236.99,58.81]$\\ $[-100.00,-100.00]$} & \\
   \cline{1-3} 
       \makecell{\citet{tiedemann2024identifying} \#1} & \makecell{2} & \makecell{\color{white}...\color{black}\\$[-100.23,-99.78]$\\\color{white}...\color{black}} & \\
   \cline{1-3} 
   \makecell{\citet{tiedemann2024identifying} \#2} & \makecell{1\\2\\3} & \makecell{$[-100.47,-99.55]$\\ $[-100.65,-99.31]$\\ $[-100.03,-99.97]$} & \\
   \cline{1-3} 
   \makecell{\citet{tiedemann2024identifying} \#3} & \makecell{1} & \makecell{\color{white}...\color{black}\\$[-100.15,-99.84]$\\\color{white}...\color{black}} & \\
   \cline{1-3}   
   \makecell{\citet{tiedemann2024identifying} \#4} & \makecell{1\\2\\3} & \makecell{$[-100.39,-99.62]$\\ $[-100.44,-99.58]$\\ $[-100.00,100.00]$} & \\
   \cline{1-3}\\[-0.2cm]  
\makecell{\citet{tiedemann2024identifying} \#5} & \makecell{1\\2} & \makecell{$[-101.16,-99.86]$\\ $[-100.90,-99.04]$} & \\\\[-0.2cm]
   \cline{1-3}\\[-0.2cm]
\makecell{\citet{tiedemann2024identifying} \#6} & \makecell{2\\3} & \makecell{ $[-101.27,-98.64]$\\ $[-100.06,-99.96]$} & \\\\[-0.2cm]
   \cline{1-3}\\[-0.2cm]    
\makecell{\citet{tiedemann2024identifying} \#7} & \makecell{1\\3} & \makecell{$[-100.38,-99.61]$\\ $[-100.00,-100.00]$} & \\\\[-0.2cm]
  \cline{1-3}  
\makecell{\citet{tiedemann2024identifying} \#8} & \makecell{1\\2\\3} & \makecell{$[-100.44,-99.56]$\\ $[-100.36,-99.63]$\\ $[-100.00,-100.00]$}&\\ \\\\
   \hline
   \end{tabular}
   \normalsize
    \caption{Simulation results for Section \ref{real_world_example}. Here,  $\hat{\beta}^P$ stands for the estimate of $\beta^P$. For the semi-synthetic experiments, we report $2.5\%$- and $97.5\%$-quantiles based on 1000 repetitions rounded to hundreths. For the real data sets we provide boxplots for $\hat{\beta}^P$ based on the block-bootstrap \citep{kuensch1989} with  $1000$ replicates and fixed block-lengths of $500$. The estimates on the entire real dataset are presented in orange (if they fall within shown interval). The column ``Model" indicates the underlying semi-synthetic data-generating model and is only relevant for the semi-synthetic estimates.}
    \label{tabel_real_world}
    \end{table}

\section{Proof of the main result}
\label{sec_proof_of_main_result}

Before proving the main result  Theorem \ref{main_theorem}, we need to introduce further trek-based lemmas. 
All these lemmas already exist more or less in analogous forms for LSEMs, see for example \citet{foygel2012half}. However, 
because SVAR processes with underlying time index set $\mathbb{Z}$ induce full time graphs with \emph{infinitely} many vertices, the original LSEM lemmas do not directly apply to our setting. Extending these LSEM lemmas to SVAR processes  with underlying time index set $\mathbb{Z}$ is non-trivial and requires some tedious calculations---there are also some subtleties regarding absolute summability and the order of summation which for LSEMs do not occur. However, it is also not surprising that such extensions are possible under suitable assumptions such as stability given that SVAR processes with underlying time index set $\mathbb{Z}$ look like ``infinite LSEMs". Because of this similarity to existing results, we defer the proofs of these lemmas to Section \ref{sec_trek_proofs} of the Appendix.

\begin{lemma} [Trek rule for stable SVAR processes and their corresponding full time graphs]
	\label{lemma_trek_rule_var_processes}
Assume a stable SVAR process $\{S_t\}_{t\in\mathbb{Z}}$ satisfying Assumptions \ref{assumption_no_instantaneous_self_edges} and \ref{assumption_acyclicity}. Let $S^{i}_{t_1}$ and $S^{j}_{t_2}$ with $t_1,t_2\in\mathbb{Z}$ be any two (potentially equal) variables in $\{S_t\}_{t\in\mathbb{Z}}$. Then,
	\begin{align}
	\label{eq_trek_rule_lemma}
\Gamma_{S^{i}_{t_1}S^{j}_{t_2}}& = \sum_{v\in V:\;t(v)\leq \min(t_1,t_2)}\;\;\sum_{\pi \in \mathcal{T}(S^{i}_{t_1},S^{j}_{t_2},v)}\pi(A, \Sigma).
\end{align}
Here, the infinite sum on the right-hand side of equation \eqref{eq_trek_rule_lemma}  has typically infinitely many summands $\sum_{\pi \in \mathcal{T}(S^{i}_{t_1},S^{j}_{t_2},v)}\pi(A, \Sigma)$ and is meant as a power series with respect to the diagonal elements of $\Sigma$ and the non-zero entries of $A^{(0)},\ldots,A^{(p)}$. Also note that each $\sum_{\pi \in \mathcal{T}(S^{i}_{t_1},S^{j}_{t_2},v)}\pi(A, \Sigma)$ is a finite sum as there are just finitely many treks with a given top node due to Remark \ref{remark_collider} and because in a trek a non-top node never has a strictly smaller time index than a top node. Furthermore, it holds that 
\begin{align}
\label{eq_trek_rule_absolute_convergence}
\sum_{v\in V:\;t(v)\leq \min(t_1,t_2)}\;\;\left|\sum_{\pi \in \mathcal{T}(S^{i}_{t_1},S^{j}_{t_2},v)}\pi(A, \Sigma)\right| < \infty,
\end{align}
and, thus, the order of the summands $\sum_{\pi \in \mathcal{T}(S^{i}_{t_1},S^{j}_{t_2},v)}\pi(A, \Sigma)$ with respect to the infinite outer sum in equations \eqref{eq_trek_rule_lemma} and \eqref{eq_trek_rule_absolute_convergence} does not matter.
\end{lemma}

\begin{remark}
	\label{remark_counterexample_trek_rule}
	An assumption like stability is generally required to ensure that a trek rule for SVAR processes holds. For a counterexample, see Section \ref{counter_ex_trek_rule} in the Appendix. 
\end{remark}	

\begin{lemma}[Parent decomposition of the covariances]
	\label{lemma_parent_decomp}
	Assume a stable SVAR process satisfying Assumptions \ref{assumption_no_instantaneous_self_edges} and \ref{assumption_acyclicity}. Let $S^{i}_{t_1}$ and $S^{j}_{t_2}$  be any two vertices in the full time graph such that $S^{i}_{t_1}$ is not a descendant of $S^{j}_{t_2}$.
For each $q\in \textnormal{pa}(S^{j}_{t_2})$, let $k_q:=\textnormal{lag}(q \rightarrow S^{j}_{t_2})$.
	Then, it holds that 	
	\begin{align}
	\label{eq_parent_decomp_new}
	\Gamma_{S^{i}_{t_1}S^{j}_{t_2}}=\sum_{q\in \textnormal{pa}(S^{j}_{t_2})} A^{(k_q)}_{S^{j}_{t_2}q}\Gamma_{S^{i}_{t_1}q}.
	\end{align}
\end{lemma}
\noindent\textit{Proof idea.}
The basic idea behind this proof is to apply the trek rule (Lemma \ref{lemma_trek_rule_var_processes}) both to $\Gamma_{S^{i}_{t_1}S^{j}_{t_2}}$ and every $\Gamma_{S^{i}_{t_1}q}$ and then to realize that every trek from $S^{i}_{t_1}$ to $S^{j}_{t_2}$ has as last edge an edge from $q\in \textnormal{pa}(S^{j}_{t_2})$ to $S^{j}_{t_2}$ because $S^{i}_{t_1}$ is not a descendant of $S^{j}_{t_2}$.
\qedsymbol

\addtocounter{theorem}{-20}
\begin{example}
	Consider the SVAR($3$) process and the full time graph in Figure \ref{example_graph_1}. There, for example, the parent decomposition yields
	\begin{align*}
	\Gamma_{Y_{t-1}Y_t}= A^{(3)}_{YY}\Gamma_{Y_{t-1}Y_{t-3}} + A^{(1)}_{YU}\Gamma_{Y_{t-1}U_{t-1}}.
	\end{align*}
	\demo
\end{example}
\addtocounter{theorem}{19}
\begin{remark}
\label{remark_superset_parent_decomp}
Note that with artificially introducing further coefficient matrices $A^{(k)}=0$ for $k>p$ and $k<0$,  the set $\textnormal{pa}(S^{j}_{t_2})$ in Lemma \ref{lemma_parent_decomp} can be replaced by any finite superset $Q\supseteq \textnormal{pa}(S^{j}_{t_2})$. This fact follows because $A^{(k_q)}_{S^j_{t_2}q}=0$ for all $q\in Q$ which are not parents of $S^j_{t_2}$.
\end{remark}

Now to the proof of Theorem \ref{main_theorem} which we directly do for the generalization mentioned in Remark \ref{remark_generalization_main_theorem}: Consider an arbitrary but fixed element $r_1\in R$. By condition 4 from Theorem \ref{main_theorem}, $r_1$ is not a descendant of $Y_t$ and, thus, we can apply the parent-decomposition (that is, Lemma \ref{lemma_parent_decomp}) to $\Gamma_{r_1 Y_t}$ with respect to the parents of $Y_t$. Writing $P_{O}\supseteq  \textnormal{pa}^{\textnormal{obs}}(Y_t)$ and $Q_L\supseteq \textnormal{pa}^{\textnormal{lat}}(Y_t)$ for arbitrary but fixed supersets just containing observed respectively unobserved variables, and additionally using Remark \ref{remark_superset_parent_decomp}, this yields
\begin{align*}
\label{eq_pa_dec_1}
    \Gamma_{r_1Y_t} &= \underbrace{\sum_{p_i \in P_O}\Gamma_{r_1 p_i} A^{(\textnormal{lag}(p_i\rightarrow Y_t))}_{Y_t p_i}}_{\textnormal{Term for observed parents of $Y_t$}}\\
    &+ \underbrace{\sum_{q_i \in Q_L}\Gamma_{r_1 q_i} A^{(\textnormal{lag}(q_i\rightarrow Y_t))}_{Y_t q_i}}_{\textnormal{Term for latent parents of $Y_t$}}.\numberthis
\end{align*}
As $r_1\in R$ was arbitrary, we get with $\Tilde{C}:=P_O \cup Q_L$ that
\begin{align}
\label{eq_C_Tilde}
    \Gamma_{RY_t} = \Gamma_{R\Tilde{C}} \begin{pmatrix}
    A^{(\textnormal{lag}(p_1\rightarrow Y_t))}_{Y_t p_1}\\
    \vdots \\
    A^{(\textnormal{lag}(q_1\rightarrow Y_t))}_{Y_t q_1}\\
    \vdots 
    \end{pmatrix},
\end{align}
where the elements of $\Tilde{C}$ are labelled such that $\Tilde{c}_1=p_1$, $\Tilde{c}_2=p_2$ and so on, and the $|P_O|+1$-th element of $\Tilde{C}$ equals $q_1$ and so on.
If all variables in equation \eqref{eq_C_Tilde} were observed, then we could directly solve that linear equation system (assuming a unique solution). However, some covariances are with respect to the latent parents, so we cannot directly calculate these covariances.

Because of this issue, we further parent-decompose the covariances $\Gamma_{r_1f_1}$ where $r_1\in R$ and $f_1\in F^{\textnormal{obs}}$ are both arbitrary. By condition 4 from Theorem \ref{main_theorem}, $r_1$ is not a descendant of $f_1$, and thus, we can parent-decompose $\Gamma_{r_1f_1}$ with respect to the parents of $f_1$. Writing $P^f_{O}\supseteq  \textnormal{pa}^{\textnormal{obs}}(F^{\textnormal{obs}})$ and $Q^f_L\supseteq \textnormal{pa}^{\textnormal{lat}}(F^{\textnormal{obs}})$ for arbitrary but fixed supersets just containing observed respectively unobserved variables yields
\begin{align*}
    \Gamma_{r_1f_1} &= \underbrace{\sum_{p^f_i \in P^f_O}\Gamma_{r_1 p^f_i} A^{(\textnormal{lag}(p^f_i\rightarrow f_1))}_{f_1 p^f_i}}_{\textnormal{Term for observed parents of $f_1$}}\\
    &+ \underbrace{\sum_{q^f_i \in Q^f_L}\Gamma_{r_1 q^f_i} A^{(\textnormal{lag}(q^f_i\rightarrow f_1))}_{f_1 q^f_i}}_{\textnormal{Term for latent parents of $f_1$}}.
\end{align*}
Rearranging this equation by putting all covariances between observed variables on one side of the equation yields
\begin{align}
\label{eq_rearrangend_further}
    \Gamma_{r_1f_1} - \underbrace{\sum_{p^f_i \in P^f_O}\Gamma_{r_1 p^f_i} A^{(\textnormal{lag}(p^f_i\rightarrow f_1))}_{f_1 p^f_i}}_{\textnormal{Term for observed parents of $f_1$}}
    &= \underbrace{\sum_{q^f_i \in Q^f_L}\Gamma_{r_1 q^f_i} A^{(\textnormal{lag}(q^f_i\rightarrow f_1))}_{f_1 q^f_i}}_{\textnormal{Term for latent parents of $f_1$}}.
\end{align}

Note that $r_1$ is by assumption not a descendant of $\textnormal{ForbAn}$ and thus, by definition, not of $\textnormal{an}^{\textnormal{lat}}(F^{\textnormal{obs}}\cup \{Y_t\})\setminus \textnormal{an}^{\textnormal{lat}}(B_U)$. Therefore, and because $B_U$ $d$-blocks all directed paths from latent vertices with strictly smaller time index than $t_{\textnormal{inf}}(B_U)$ to $Q^f_L\setminus B_U$ by assumption (subcondition 2b in Theorem \ref{main_theorem}), 
one can now repeatedly apply the parent decomposition to the $\Gamma_{r_1q^f_i}$ (and the components of $\Gamma_{r_1\textnormal{an}^{\textnormal{lat}}(q^f_i)}$) until either all occuring covariances are between $r_1$ and elements of $B_U$ or between $r_1$ and some latent vertex $w$ that does not have further (latent) ancestors except for itself (recall here that latent variables never have observed parents by Assumption \ref{assumption_latents_have_no_observed_parents}).
As in the latter case where $w$ does not have further (latent) ancestors except for itself it also holds by assumption (subcondition 4 in Theorem \ref{main_theorem}) that $r_1$ is not a descendant of $w$, it holds that $r_1$ and $w$ then do not have a common ancestor at all, and thus, there is no trek between $r_1$ and $w$. From the trek rule, it then follows that the covariance between $r_1$ and $w$ is zero and hence the corresponding covariance-term can be dropped from the repeatedly applied parent-decompositions. Therefore, by repeatedly applying the parent-decomposition, one can write each $\Gamma_{r_1 q^f_i}$ with $q^f_i \in Q^f_L$ as a linear combination of the covariances $\Gamma_{r_1b_j}$ where $b_j\in B_U$, that is,
\begin{align*}
    \Gamma_{r_1 q^f_i} = \sum_{b_j\in B_U}\Gamma_{r_1b_j}\lambda_{q^f_i b_j}
\end{align*}
where $\lambda_{q^f_i b_j}$ is the sum of all path monomials of all directed paths from $b_j$ to $q^f_i$ that do not go through any other $b_l$ with $j\neq l$.

Plugging this ``basis-covariance-representation" back into equation \eqref{eq_rearrangend_further}, one gets
\begin{align*}
    \Gamma_{r_1f_1} - \sum_{p^f_i \in P^f_O}\Gamma_{r_1 p^f_i} A^{(\textnormal{lag}(p^f_i\rightarrow f_1))}_{f_1 p^f_i}
    &= \sum_{q^f_i \in Q^f_L} A^{(\textnormal{lag}(q^f_i\rightarrow f_1))}_{f_1q^f_i}\sum_{b_j\in B_U}\Gamma_{r_1b_j}\lambda_{q^f_i b_j}\\
    &=:\sum_{b_j\in B_U}\Gamma_{r_1b_j}\Tilde{\lambda^1}_{b_j}
\end{align*}
where $\Tilde{\lambda^1}_{b_j}$ is the sum of all path monomials of all directed paths from $b_j$ to $f_1$ not going through any other $b_l$ with $j\neq l$ and that only use latent vertices except for the last vertex $f_1$. 

As $f_1$ was arbitrary, one gets in matrix notation 
\begin{align}
\label{eq_matrix_not}
\begin{pmatrix}
 \Gamma_{r_1f_1} - \sum_{p^f_i \in P^f_O}\Gamma_{r_1 p^f_i} A^{(\textnormal{lag}(p^f_i\rightarrow f_1))}_{f_1 p^f_i}\\
  \Gamma_{r_1f_2} - \sum_{p^f_i \in Q^f_L}\Gamma_{r_1 p^f_i} A^{(\textnormal{lag}(p^f_i\rightarrow f_2))}_{f_2 p^f_i}\\
  \vdots
\end{pmatrix}
=
   \Tilde{\Lambda}\cdot 
   \begin{pmatrix}
   \Gamma_{r_1b_1}\\
   \Gamma_{r_1b_2}\\
   \vdots 
   \end{pmatrix} 
\end{align}
where $\Tilde{\Lambda}$ is the matrix with entries defined by $\Tilde{\Lambda}_{kj}=\Tilde{\lambda^k}_{b_j}$ (where the $k$ refers to $f_k$).
Note that $\Tilde{\Lambda}$ is the same for each choice of $r_1$.

From condition 2c in Theorem \ref{main_theorem} together with Lemma \ref{lemma:invertible_lambda_new} in Section \ref{sec_invertible_lambda_new} of the Appendix follows the invertibility of $\Tilde{\Lambda}$ in generic settings and, thus, one obtains the $\Gamma_{r_1b_j}$ in generic settings by solving equation \eqref{eq_matrix_not} for the $\Gamma_{r_1b_j}$.

Next, again note that $r_1$ is by assumption not a descendant of $\textnormal{ForbAn}$ and thus, by definition, not of $\textnormal{an}^{\textnormal{lat}}(F^{\textnormal{obs}}\cup \{Y_t\})\setminus \textnormal{an}^{\textnormal{lat}}(B_U)$. From the exact same argument as for $q^f_i$ (except for using subcondition 1 instead of subcondition 2), one can therefore by repeatedly applying the parent-decomposition write each $\Gamma_{r_1 q_i}$ with $q_i \in Q_L$ as a linear combination of the covariances $\Gamma_{r_1b_j}$ where $b_j\in B_U$, that is,
\begin{align}
\label{eq_parent_decomp_q_i}
    \Gamma_{r_1 q_i} = \sum_{b_j\in B_U}\Gamma_{r_1b_j}\lambda_{q_i b_j}
\end{align}
where $\lambda_{q_i b_j}$ is the sum of all path monomials of directed paths from $b_j$ to $q_i$ that do not go through any other $b_l$ with $j\neq l$.

As one has just solved for the $\Gamma_{r_1b_j}$, one can plug these into equation \eqref{eq_parent_decomp_q_i} and get each $\Gamma_{r_1 q_i}$ as a potentially non-trivial function of the $\Gamma_{r_1f_k}$'s and $\Gamma_{r_1p_i^f}$'s. Thus, one can express the $\Gamma_{r_1 q_i}$'s just in terms of observed covariances, and hence, plugging these expressions of the $\Gamma_{r_1 q_i}$'s back into equation \eqref{eq_pa_dec_1} yields a linear equation just in terms of observed covariances.

Collecting all the variables in that rewritten equation for which one took covariances of $r_1$ with, exactly yields the set $C$ from the theorem. As neither $C$ nor the functions in the linear equation system in front of the respective covariances depend on the choice of $r_1$ as $\Tilde{\Lambda}$ does not depend on $r_1$ as explained previously, one gets the linear equation system from equation \eqref{main_lin_eq}. In particular, the coefficients in front of the covariances that were with respect $P_O$ are the direct causal effects of the respective parent of $Y_t$ onto $Y_t$ (or $0$ if the respective element of $P_O$ is not a parent): This fact follows from the fact that the covariances with respect to $P_O$ are not introduced except in equation \eqref{eq_pa_dec_1} because 
\begin{align*}
  \bigr(F^{\textnormal{obs}}\cup P^f_O\bigr)\cap P_O=\emptyset
\end{align*} 
by requirement 3 in Theorem \ref{main_theorem} and because in equation \eqref{eq_pa_dec_1} the parameter in front of the respective covariance is the direct causal effect of interest.
Finally, using subpoint 5 from the assumptions in Theorem \ref{main_theorem} and Lemma \ref{lemma_unique_solvability_new} in Section \ref{sec_unique_solvability_new} of the Appendix, one gets unique solvability. (The other part of requirement 3 which stated that 
$(F^{\textnormal{obs}}\cup P^f_O)\cap \{Y_t\}=\emptyset$ ensures that a certain trivial solution is excluded---strictly speaking, subpoint 3 is already contained in subpoint 5 as Lemma \ref{lemma_unique_solvability_new} only requires subpoint 5 and not subpoint 3).
\qedsymbol

\section{Conclusion and Outlook}
\label{sec_conclusion}
	In this paper, we have shown that the time structure of an SVAR process can enable the identification of direct causal effects in the presence of latent confounding without making use of additional auxiliary variables or time series such as negative controls \citep{hu2023using}, instruments \citep{mogensen2023instrumental, michael2024instrumental, thams2022identifying} or time series that satisfy the backdoor-criterion in certain graphs \citep{eichler2007causal,eichler2010granger,eichler2012causal}.
	While similar ideas in the frequency domain have existed before \citep{mahecha2010global, schur2024decor}, an approach in the time domain and in particular for SVAR processes such as ours so far only seems to exist for relatively simple full time graphs [\citealp{malinvaud1961estimation}, particularly Section 6 in \citealp{bercu2013sharp,proia2013further}]. Furthermore, we have both derived sufficient graphical and lag-based results and illustrated the validity and applicability of our theory on several synthetic examples and a real-world example.
	
	In future research, it might be interesting to study the statistical performance of these estimators in dependence of the different parameters and lags. It might also be interesting to compare different estimators valid for the same identification problem (both including estimators based on different valid choices of $R$ and $C$ as well as other estimators for example based on instrumental time series) in terms of their statistical performance and to discuss some notion of optimal estimation. Furthermore, it might be interesting to study the robustness of our proposed estimators against model violations or to relax the knowledge of the full time graph that is required---these two points have already been hinted at in Remark \ref{remark_generalization_main_theorem}. Another future avenue might be to derive results that directly enable identification of total causal effects (and not, which is already possible so far, by combining different estimates for direct causal effects via Wright's path rule \citep{wright1934method}). 

\section{Acknowledgements}
We thank Nicolas-Domenic Reiter for helpful discussions and Alexandrine Lanson and Urmi Ninad for pointing us to related literature. J.R.\ received funding from the European Research Council (ERC) Starting Grant CausalEarth under the European Union’s Horizon 2020 research and
innovation program (Grant Agreement No.\ 948112).

	\appendix 
	\section{Glossary}
	\label{sec_glossary}
	\begin{table}[H]
		\begin{center}
			\begin{tabular}{ |c|c| } 
				\hline
				\textbf{Symbol(s)} & \textbf{Meaning}\\
				\hline
				$[k]_0$ & \makecell{$\{0,\ldots,k\}$}\\
				\hline
				$[k]_1$ & \makecell{$\{1,\ldots,k\}$}\\
				\hline
				$\mathbb{N}_{\geq k}$ & \makecell{$\{k, k+1,k+2,\ldots\}$}\\
				\hline 
				$t_{\textnormal{inf}}(B)$ and $t_{\textnormal{sup}}(B)$ & infimum and supremum time index of a set of vertices $B$, resp.\ \\
				\hline
				$\textnormal{pa}^{\textnormal{obs}}(B)$ and $\textnormal{pa}^{\textnormal{lat}}(B)$ & observed and latent parents of a set $B$ of vertices, resp.\ \\
				\hline
				$\mathbb{Z}$ & \makecell{$\{0,1,-1,2,-2,\ldots\}$}\\
				\hline
				$p$ & \makecell{Order of the SVAR-process}\\
				\hline
			 \makecell{$\{S_t\}_{t\in\mathbb{Z}}=\allowbreak\{(S^1_t,\allowbreak\ldots,\allowbreak S^{d}_t)\}_{t\in\mathbb{Z}}$\\$=\allowbreak\{(U^1_t,\allowbreak\ldots,\allowbreak U^{d_U}_t,\allowbreak O^1_t,\allowbreak\ldots,\allowbreak O^{d_O}_t)\}_{t\in\mathbb{Z}}$}  & \makecell{SVAR process of order $p\in\mathbb{N}_{\geq 0}$ with latent component time series\\ $\{(U^1_t,\ldots,U^{d_U}_t)\}_{t\in\mathbb{Z}}$ and observed component time series\\ $\{(O^1_t,\ldots,O^{d_O}_t)\}_{t\in\mathbb{Z}}$ where $\{O^1_t\}_{t\in\mathbb{Z}} = \{Y_t\}_{t\in\mathbb{Z}}$}\\
				\hline
				$A^{(0)},\ldots, A^{(p)}$ & \makecell{Coefficient matrices of the SVAR-process}\\
				\hline
				\makecell{$A^{(h)}_{jk}$ or $A^{(h)}_{S^jS^k}$ or $A^{(h)}_{S^j_tS^k_{t-h}}$}& \makecell{$(j,k)$-th entry of the matrix $A^{(h)}$}\\
				\hline 
				$\Gamma_S(h)$ & Autocovariance function of $\{S_t\}_{t\in\mathbb{Z}}$ at lag $h\in\mathbb{Z}$\\
				\hline
				$\Gamma^{\textnormal{obs}}_S(h)$ & Autocovariance function of  $\{(O^1_t,\ldots,O^{d_O}_t)\}_{t\in\mathbb{Z}}$ at lag $h\in\mathbb{Z}$\\
				\hline
				$\Gamma_{S^i_tS^j_{\Tilde{t}}}$ & Covariance between $S^i_t$ and $S^j_{\Tilde{t}}$ \\
				\hline
				$m_{S^iS^j}$ & \makecell{Number of $h\in[p]_0$ for which $A^{(h)}_{S^iS^j}\neq 0$ }	\\
				\hline
				$m_{S^i}$ & \makecell{Number of $h\in[p]_1$ for which $A^{(h)}_{S^iS^i}\neq 0$ }	\\
				\hline
				$l^{S^iS^j}_{1}, \ldots, l^{S^iS^j}_{m_{S^iS^j}}$ & \makecell{Lags of the causal relationships from the $S^j$- to the $S^i$-variables,\\ that is, the indices $h\in[p]_0$ for which $A^{(h)}_{S^iS^j}\neq 0$} \\
				\hline
				$l^{S^i}_{1}, \ldots, l^{S^i}_{m_{S^i}}$ & \makecell{Lags of the causal relationships from the $S^i$-variables onto itself,\\ that is, the indices $h\in[p]_1$ for which $A^{(h)}_{S^iS^i}\neq 0$} \\
				\hline
			\end{tabular}
		\end{center}
		\caption{Overview of frequently used notation}
	\end{table}
\section{Proof of Proposition \ref{proposition_consistency} from Section \ref{sec_main_identifiability_result}}
\label{sec_consistency}
For $T_0\in\mathbb{Z}$ and $T\in \mathbb{N}_{\geq 1}$, consider $S_{T_0+1},\ldots,S_{t_0+T}$. Then (see Section 10.2 in \citet{hamilton2020time}), for every $t\in \mathbb{Z}$ and $T\rightarrow \infty$,
\begin{align*}
\hat{\mathbb{E}}[S_t]:=\frac{1}{T}\sum_{j=1}^T S_{T_0+j}\overset{P}{\longrightarrow} \mathbb{E}[S_t]
\end{align*}
and 
\begin{align}
\label{eq_cov_est}
\hat{\Gamma}_{S_{t-h}S_t}:=\frac{1}{T-h}\sum_{j=h+1 }^TS_{T_0+j-h}S_{T_0+j}^T \overset{P}{\longrightarrow} \Gamma_{S_{t-h}S_t}.
\end{align}
(Here, $\overset{P}{\longrightarrow}$ denotes convergence in probability). If Theorem \ref{main_theorem} applies, then $\Gamma_{RC}$ is invertible in generic settings and the unique solution of the linear equation system in equation \eqref{main_lin_eq} is given by 
\begin{align*}
v=(\Gamma_{RC})^{-1}\cdot \Gamma_{RY_t}.
\end{align*}

Now, from the convergence in probability result in equation \eqref{eq_cov_est}, it follows that every entry of $\hat{\Gamma}_{RY_t}$ and every entry of $\hat{\Gamma}_{RC}$ converges in probability to the corresponding  entry of $\Gamma_{RY_t}$ and $\Gamma_{RC}$, respectively. As convergence in probability of every entry of a vector or matrix implies that the entire vector or matrix converges in probability to the vector or matrix made up by the individual component-limits, it follows that $\hat{\Gamma}_{RY_t}\overset{P}{\longrightarrow} \Gamma_{RY_t}$ and $\hat{\Gamma}_{RC}\overset{P}{\longrightarrow}\Gamma_{RC}$. As the map (where $\textnormal{GL}_k(\mathbb{R})$ denotes the general linear group of degree $k\in\mathbb{N}_{\geq 1}$ over the field $\mathbb{R}$)
\begin{align*}
f:\; \textnormal{GL}_k (\mathbb{R}) \rightarrow \textnormal{GL}_k(\mathbb{R}),\; B\mapsto B^{-1}
\end{align*}
is continuous for all $k\in\mathbb{N}_{\geq 1}$, and because matrix multiplication is continuous as well, it follows from the continuous mapping theorem \citep{mann1943} that 
\begin{align*}
(\hat{\Gamma}_{RC})^{-1}\cdot \hat{\Gamma}_{RY_t}\overset{P}{\longrightarrow }(\Gamma_{RC})^{-1}\cdot \Gamma_{RY_t}=v.
\end{align*}

\section{Proofs for lag-based lemmas from Section \ref{sec_sufficient_lag_based_criteria}}
\subsection{Proof of Lemma \ref{lemma_resid_class1}}

\begin{proof}
We construct a system of treks $\Psi:R\rightrightarrows C$ as follows: For every observed component time series $\{O^i_t\}_{t\in \mathbb{Z}}$ if $C^{(1)}_{O^i}\neq \emptyset$, then the treks starting in $R^{(1)}_{O^i}$ go to $C^{(1)}_{O^i}$; and if $C^{(2)}_{O^i}\neq \emptyset$, then the treks starting in $R^{(2)}_{O^i}$ go to $F^{\textnormal{obs}}_{O^i}=C^{(2)}_{O^i}$ (this combination of sources and targets is possible because $|R^{(2)}_{O^i}|=|C^{(2)}_{O^i}|$ and $|R^{(1)}_{O^i}|=|C^{(1)}_{O^i}|$ by Assumption (C1)). By Assumptions (C2) and (C3), for each observed component time series $\{O^i_t\}_{t\in \mathbb{Z}}$ for which $C^{(1)}_{O^i}\neq \emptyset$, there exists a lag $l^{O^i}_{j_i}$ such that all vertices in $C^{(1)}_{O^i}$ with time index in $[t-\tau_{O^i}-(l^{O^i}_{j_i}-1),\infty)$ are in different $l^{O^i}_{j_i}$-residue classes and such that for each $c_i\in C^{(1)}_{O^i}$ there exists exactly one $r_i \in R^{(1)}_{O^i}$ with time index in $[t-\tau_{O^i}-(l^{O^i}_{j_i}-1),t-\tau_{O^i}]$ from the same $l^{O^i}_{j_i}$-residue class with equal or strictly smaller time index than $c_i$ (this time index fact holds because the interval $[t-\tau_{O^i}-(l^{O^i}_{j_i}-1),t-\tau_{O^i}]$ contains each $l^{O^i}_{j_i}$-residue class exactly once). Therefore, for the treks starting in $R^{(1)}_{O^i}$ with time index in $[t-\tau_{O^i}-(l^{O^i}_{j_i}-1),t-\tau_{O^i}]$ we can just use treks that are directed paths using $l^{O^i}_{j_i}$-edges and going to their respective targets in $C^{(1)}_{O^i}$ with time index in $[t-\tau_{O^i}-(l^{O^i}_{j_i}-1),\infty)$ from the same $l^{O^i}_{j_i}$-residue class (these directed paths point forward in time by the fact that the time index of $r_i$ is equal or strictly smaller than the time index of $c_i$). For the $c_i\in C^{(1)}_{O^i}$ with time index in $(-\infty, t-\tau_{O^i}-l^{O^i}_{j_i}]$, just choose the trivial treks, which by Assumption (C4) is possible because for such $c_i$ it holds that $c_i\in R^{(1)}_{O^i}$.
Note that all the treks going from $R^{(1)}_{O^i}$ to $C^{(1)}_{O^i}$ clearly do not intersect because the treks starting in the time interval $(-\infty, t-\tau_{O^i}-l^{O^i}_{j_i}]$ are all trivial and never go to vertices with time index strictly larger than $t-\tau_{O^i}-l^{O^i}_{j_i}$, and the treks starting in $[t-\tau_{O^i}-(l^{O^i}_{j_i}-1),t-\tau_{O^i}]$ are either trivial or go to vertices with strictly larger time indices and do so 
by just using vertices from different $l^{O^i}_{j_i}$-residue classes and $l^{O^i}_{j_i}$-edges. Also clearly, the treks going from $R^{(1)}$ to $C^{(1)}$ do not intersect because for each $\{O^i_t\}_{t\in \mathbb{Z}}$ for which $C^{(1)}_{O^i}\neq \emptyset$ the treks going from $R^{(1)}_{O^i}$ to $C^{(1)}_{O^i}$ do not intersect and because each trek from $R^{(1)}$ to $C^{(1)}$ just uses edges and vertices from one particular component time series.

Regarding $R^{(2)}$ and $C^{(2)}$: For each $\{O^i_t\}_{t\in \mathbb{Z}}$ for which $C^{(2)}_{O^i}\neq \emptyset$, the last edge in each trek starting in $R^{(2)}_{O^i}$ and ending in $C^{(2)}_{O^i}$ is the chosen $l^{O^iU^{k_i}}_{w_i}$-edge; the second last vertex for each of these treks is the respective element of $\textnormal{pa}(F^{\textnormal{obs}}_{O^i}, l^{O^iU^{k_i}}_{w_i})$ that is connected via the $l^{O^iU^{k_i}}_{w_i}$-edge to the respective element of $C^{(2)}_{O^i}$. By Assumption (C6.1), for each $\{U^k_t\}_{t\in\mathbb{Z}}$ for which $P_{U^k}\neq \emptyset$, there exists a lag $l^{U^k}_{j_k}$ (so in particular, $m_{U^k}>0$) such that all $p_k\in P_{U^k}$ are in different $l^{U^k}_{j_k}$-residue classes. Moreover, by Assumption (C6.2), for each $p_k\in P_{U^k}$ there exists exactly one $q_k\in Q_{U^k}$ from the same $l^{U^k}_{j_k}$-residue class. By Assumption (C5.2) together with the fact that $t_{\textnormal{inf}}(R^{(1)}_{O^i})\leq t_{\textnormal{inf}}(C^{(1)}_{O^i})$ by Assumptions (C2)--(C4), it follows that $t_{\textnormal{sup}}(R^{(2)}_{O^i})\leq t_{\textnormal{inf}}(C^{(2)}_{O^i})$ and thus that each $q_k$ has strictly smaller time index than its corresponding $p_k$. Thus, it is possible to construct for each $\{U^k_t\}_{t\in\mathbb{Z}}$ and for each $q_k\in Q_{U^k}$ a directed path just using $l^{U^k}_{j_k}$-edges from $q_k$ to the corresponding $p_k\in P_{U^k}$ from the same $l^{U^k}_{j_k}$-residue class. Finally, each $q_k$ can be connected to the respective element from $R^{(2)}_{O^i}$ by the chosen $l^{O^iU^k}_{j_k}$-edge. By construction, the treks from $R^{(2)}$ to $C^{(2)}$ do not intersect as each trek is using different $l^{U^k}_{j_k}$-residue classes and just edges linking latent and observed variables. Furthermore, by Assumption (C5.2), the fact that $t_{\textnormal{inf}}(R^{(1)}_{O^i})\leq t_{\textnormal{inf}}(C^{(1)}_{O^i})$ by Assumptions (C2)--(C4) and by construction, the treks from $R^{(1)}$ to $C^{(1)}$ and the treks from $R^{(2)}$ to $C^{(2)}$ do not intersect. Thus, in total, the treks from $\Psi$ have no-sided intersection.

Now suppose there is another system of treks $\Pi:R\rightrightarrows C$ with no sided intersection and such that $\Pi(A,\Sigma)=\Psi(A,\Sigma)$. Note that in $\Pi(A, \Sigma)=\Psi(A,\Sigma)$ there are exactly as many terms for $l^{O^i}_{j_i}$-edges as the directed paths starting in $R^{(1)}_{O^i}$ and ending in $C^{(1)}_{O^i}$ are in total long. Furthermore, note that there are exactly $2\cdot |R^{(2)}|$-terms in $\Pi(A, \Sigma)=\Psi(A,\Sigma)$ for edges linking unobserved and observed variables as every trek in $\Psi$ from $R^{(2)}$ to $C^{(2)}$ has exactly two such edges and the treks in $\Psi$ from $R^{(1)}$ to $C^{(1)}$ have no such edges. Therefore, and because every trek in $\Pi$ starts and ends in observed variables and because there are no further terms in $\Pi(A, \Sigma)=\Psi(A,\Sigma)$ for edges linking different observed component time series, at least $|R^{(1)}|$-many treks stay in the observed component time series in which they started. From the time index requirement in Assumption (C5.2), we conclude that none of the treks just using $l^{O^i}_{j_i}$-edges can start in $R^{(2)}$ or end in $C^{(2)}$ as simply not enough terms for $l^{O^i}_{j_i}$-edges are available in $\Pi(A, \Sigma)=\Psi(A,\Sigma)$.
Similarly,  by the fact that treks in $\Psi$ starting at vertices in $R^{(1)}_{O^i}$ with time index strictly less than or equal to $t-\tau_{O^i}-l^{O^i}_{j_i}$ are trivial and hence never go to larger time indices, it follows that also the treks in $\Pi$ starting at vertices in $R^{(1)}_{O^i}$ with time index strictly less than or equal to $t-\tau_{O^i}-l^{O^i}_{j_i}$ cannot go to a target from $C^{(1)}_{O^i}$ with time index strictly larger than or equal to $t-\tau_{O^i}-(l^{O^i}_{j_i}-1)$, as for doing so also simply not enough terms for $l^{O^i}_{j_i}$-edges are available in $\Pi(A, \Sigma)=\Psi(A,\Sigma)$.
Thus and because all vertices in $C^{(1)}_{O^i}$ with time index strictly larger than or equal to $\tau_{O^i}-(l^{O^i}_{j_i}-1)$ are in different $l^{O^i}_{j_i}$-residue classes and because the only terms for linking different $O^i$-vertices in $\Pi(A, \Sigma)=\Psi(A,\Sigma)$ are terms for $l^{O^i}_{j_i}$-edges, we conclude that all treks in $\Pi$ starting in $R^{(1)}$ are as in $\Psi$.

Regarding the treks in $\Pi$ starting in $R^{(2)}$: As all terms in $\Pi(A,\Sigma)=\Psi(A,\Sigma)$ for $l^{O^i}_{j_i}$-edges are already required for the treks starting in $R^{(1)}$, the treks starting in $R^{(2)}$ need to go directly to unobserved variables, similarly the last edge of each such trek in $\Pi$ must be from an unobserved to an observed variable. Also note that because there are exactly $2\cdot |R^{(2)}|$-terms in $\Pi(A, \Sigma)=\Psi(A,\Sigma)$ for edges linking unobserved and observed variables and no terms for edges linking different latent component time series, each trek starting in $R^{(2)}$ can just use $l^{U^k}_{j_k}$-edges (corresponding to the $\{U^k_t\}_{t\in\mathbb{Z}}$ they have reached after their first edge) except for their first and last edge. Also note that, because just \emph{one} particular edge type, namely $l^{O^iU^{k_i}}_{w_i}$, for linking each observed component time series $\{O^i_t\}_{t\in\mathbb{Z}}$ for which $C^{(2)}_{O^i}\neq \emptyset$ to the latent variables has been used in $\Psi$, each trek in $\Pi$ starting in $R^{(2)}$ reaches the exact same latent $l^{U^k}_{j_k}$-residue class that the corresponding trek in $\Psi$ starting from the same source reached. The fact that for linking each $q\in Q$ to its correspondent $p\in P$ from the same $l^{U^k}_{j_k}$-residue class only terms in $\Pi(A,\Sigma)=\Psi(A,\Sigma)$ for $l^{U^k}_{j_k}$-edges are left and the number of these terms equals exactly the number of minimal required $l^{U^k}_{j_k}$-edges to get from $q$ to $p$, then implies that the treks in $\Pi$ starting in $R^{(2)}$ are as the treks in $\Psi$ starting in $R^{(2)}$. Therefore, we conclude that $\Pi=\Psi$.
\end{proof}

\subsection{Proof of Lemma \ref{lemma_resid_class_new}}
\begin{proof}
The goal is to apply Lemma \ref{lemma_resid_class1}. For that, we first need to define a set of sources $R$.
We define $R=R^{(1)}\dot{\cup} R^{(2)}$ by defining $R^{(1)}$ and $R^{(2)}$ individually. To begin with, we define $R^{(1)}$ by defining its individual intersections with observed time series $R^{(1)}_{O^i}$: For $R^{(1)}_{O^i}$, take for every $\{O^i_t\}_{t\in \mathbb{Z}}$ for which $C^{(1)}_{O^i}\neq \emptyset$ the set such that for every $c_i\in C^{(1)}_{O^i}$ with time index in $[t-\tau_{O^i}-(l^{O^i}_{j_i}-1),\infty)$ there exactly exists one $r_i\in R^{(1)}_{O^i}$ from the same $l^{O^i}_{j_i}$-residue class with time index in $[t-\tau_{O^i}-(l^{O^i}_{j_i}-1),t-\tau_{O^i}]$ and such that every $c_i\in C^{(1)}_{O^i}$ with time index in $(-\infty,\tau_{O^i}-l^{O^i}_{j_i}]$ is an element of $R^{(1)}_{O^i}$. Doing so is possible because the set of $O^i$-vertices with time index in $[t-\tau_{O^i}-(l^{O^i}_{j_i}-1),t-\tau_{O^i}]$ contains vertices from each $l^{O^i}_{j_i}$-residue class exactly once and because of Assumption (C2). For every $\{O^i_t\}_{t\in \mathbb{Z}}$ for which $C^{(1)}_{O^i}= \emptyset$, let $R^{(1)}_{O^i}=\emptyset$. 
Therefore by construction, conditions (C3) and (C4) are satisfied.

We also define $R^{(2)}$ by its individual intersections with observed time series $R^{(2)}_{O^i}$: If $C^{(2)}_{O^i}=\emptyset$, then let $R^{(2)}_{O^i}=\emptyset$. Otherwise, define each $R^{(2)}_{O^i}$ by time-shifting every element of $C^{(2)}_{O^i}$ by the same amount, which is required to be divisible by $l^{U^k}_{j_k}$, to the past such that $t_{\textnormal{sup}}(R^{(2)}_{O^i})<t_{\textnormal{inf}}(R^{(1)}_{O^i})$. Write $Q:=\bigcup_{i\in [d_O]_1:\;C^{(2)}_{O^i}\neq \emptyset} \textnormal{pa}(R^{(2)}_{O^i}, l^{O^iU^{k_i}}_{w_i})$.
    Note that for each $\{U^k_t\}_{t\in\mathbb{Z}}$ for which $P_{U^k}\neq \emptyset$, 
    the vertices in $Q_{U_k}$ (which is by the repeating edges property then also not empty) cover the exact same $l^{U^k}_{j_k}$-residue classes because the time-shift is divisible by $l^{U_k}_{j_k}$---this fact also follows from the repeating edges property.  In addition, all variables in $P$ are in different $l^{U^k}_{j_k}$-residue classes by Assumption (C6.1). Thus, (C6.2) holds.
    
    Moreover, by construction, Assumption (C1) holds. Similarly by construction and because Assumption (C5.1) holds, Assumption (C5.2) holds.
    
    Besides, every element in each $R^{(1)}_{O^i}$ and $R^{(2)}_{O^i}$ has a time index smaller than or equal to $t-\tau_{O^i}$ by construction, and thus the non-descendance requirement on $R$---so condition 4---from Theorem \ref{main_theorem} is satisfied.
    Thus,  all requirements for Lemma \ref{lemma_resid_class1} are satisfied from which then the result follows.
\end{proof}    

\subsection{Proof of Lemma \ref{lemma_BU_Fobs_example}}

\begin{proof}
    By just counting elements, requirement 2a from Theorem \ref{main_theorem} immediately follows. Next, note that by construction, $B_U$ contains the $U$-vertex from $\textnormal{pa}^{\textnormal{lat}}(Y_t)\cup \textnormal{pa}^{\textnormal{lat}}(F^{\textnormal{obs}})$ with minimal time index and all $l^U_{m_U}$-many subsequent $U$-vertices. Thus, as any set of $l^U_{m_U}$-many subsequent $U$-vertices contains exactly one $U$-vertex from each $l^{U}_{m_U}$-residue class, and hence, every directed path from a $U$-vertex with a time index strictly smaller than $t_0$ to $\textnormal{pa}^{\textnormal{lat}}(Y_t)\cup \textnormal{pa}^{\textnormal{lat}}(F^{\textnormal{obs}})$ goes through $B_U$, it follows that requirements 1 an 2b from Theorem \ref{main_theorem} hold.

Regarding 2c:    By assumption, there exists some $l^{O^{i_0}U}_j$-edge linking $U$ and $O^{i_0}$. Take such an $l^{O^{i_0}U}_j$-edge as the last edge of each directed path in $\Upsilon$. Then, the second-last vertex of each directed path in $\Upsilon$ is latent and from the fact that $F^{\textnormal{obs}}$ consists of $l^{U}_{m_U}$-many subsequent vertices, it follows that all these second-last vertices are in different $l^U_{m_U}$-residue classes. Also note that the vertices of $B_U$ by construction occupy all $l^U_{m_U}$-residue classes. Thus, one can connect the respective elements of $B_U$ with the respective second-last vertices by just using $l^U_{m_U}$-edges. As all these directed paths in $\Upsilon$ use vertices from different $l^U_{m_U}$-residue-classes, no intersections occur.
    
    Now, suppose there exists another system of directed paths $\Pi:B_U\rightarrow F^{\textnormal{obs}}$ with no intersections and such that $\Pi(A)=\Upsilon(A)$. Note that the last edge of each directed path in $\Pi$ needs to be as in $\Upsilon$ by the imposed requirement on $\Pi$ that the last edge in each directed path is from a latent to an observed vertex and the fact that only one edge type linking latent and observed vertices has been used by the directed paths in $\Upsilon$. The previous edges in $\Pi$ then can just be $l^U_{m_U}$-edges, as only terms for these edges are left in $\Pi(A)=\Sigma(A)$ after excluding the $l^{U}_{m_U}$-many terms for edges linking latent and observed vertices. Then, as all second-last-vertices are in different $l^U_{m_U}$-residue classes, each directed path in $\Pi$ needs to be as in $\Upsilon$. Thus, we conclude $\Pi=\Upsilon$ and hence, requirement 2c from Theorem \ref{main_theorem} is satisfied.
\end{proof}

\section{Proofs for Lemmas from Section \ref{sec_proof_of_main_result}}
\label{sec_trek_proofs}
\subsection{Proof of Lemma \ref{lemma_trek_rule_var_processes}}
\label{appendix_proof_of_trek_rule}
\begin{proof}
We start with the proof for SVAR(1) processes.

First, recall that $A^{(0)}$ can be transformed to a strictly lower triangular matrix by permuting rows and columns (due to Assumptions \ref{assumption_no_instantaneous_self_edges} and \ref{assumption_acyclicity}) and, hence, $A^{(0)}$ is in fact nilpotent and its eigenvalues are all $0$. Therefore, we can write $B^{(0)}:=(I-A^{(0)})^{-1}=\sum_{m=0}^\infty (A^{(0)})^m = \sum_{m=0}^M (A^{(0)})^m$ for some $M\in\mathbb{N}_{\geq 0}$. Besides, note that for some $k,l\in[d]_1$ the $(k,l)$-th entry of $A^{(1)}$ equals the edge coefficient for the lag $1$-edge $S^l_{t-1}\rightarrow S^k_t$, the $(k,l)$-th entry of $A^{(0)}A^{(1)}$ equals the sum over all path monomials of directed paths from $S^l_{t-1}$ to $S^k_t$ whose first edge is a lag $1$-edge and who then use exactly one instantaneous edge; and more generally, the $(k,l)$-th entry of $(A^{(0)})^{m}A^{(1)}$ equals the sum over all path monomials of directed paths from $S^l_{t-1}$ to $S^k_t$ whose first edge is a lag $1$-edge and who then use exactly $m$ instantaneous edges. Thus, the $(k,l)$-th entry of $B^{(1)}:=(I-A^{(0)})^{-1}A^{(1)}$ equals the sum of all path monomials of directed paths from $S^l_{t-1}$ to $S^k_t$ whose first edge is a lag $1$-edge and who then use arbitrarily yet finitely many instantaneous edges. 
Therefore, the $(k,l)$-th entry of $(B^{(1)})^mB^{(0)}$ equals the sum of all path monomials of directed paths from $S^{l}_{t-m}$ to $S^k_{t}$.

Now, recall that one can write a stable SVAR process satisfying Assumptions \ref{assumption_no_instantaneous_self_edges} and \ref{assumption_acyclicity} as a stable VAR process---see equation \eqref{var_process_2} in Section \ref{sec_var_processes}---with noise covariance matrix $B^{(0)}\Sigma (B^{(0)})^T$. This stable VAR representation together with equation (2.1.32) in \citet{lutkepohl2005new} which holds for stable VAR processes implies that

 \begin{align}
 \label{vec_1}
    \text{vec}(\Gamma_S(0))=\left(I_{d^2}-B^{(1)}\otimes B^{(1)}\right)^{-1}\text{vec}\left(B^{(0)}\Sigma (B^{(0)})^T\right),
 \end{align}
 where $\text{vec}$ stands for the vector that is constructed by stacking all the columns from left to right of a matrix into a vector starting at the top and ending at the bottom, and $\otimes$ denotes the Kronecker product.
 Because $\{S_t\}_{t\in\mathbb{Z}}$ is stable, it holds (by definition)  that all eigenvalues of $B^{(1)}$ have modulus strictly less than $1$ \citep[Section 2.1.1]{lutkepohl2005new}. From standard rules on the Kronecker product it thus follows that also all eigenvalues of $B^{(1)}\otimes B^{(1)}$ have modulus less than one \citep[Sections 2.1.4 and A.11]{lutkepohl2005new}, which implies that one can write equation \eqref{vec_1} using a von Neumann series, that is, as
 \begin{align*}
 \label{vec_geo}
     \text{vec}(\Gamma_S(0))&=\left(\sum_{m=0}^\infty\left(B^{(1)}\otimes B^{(1)}\right)^m\right)\text{vec}\left(B^{(0)}\Sigma (B^{(0)})^T\right)\\
     &=\sum_{m=0}^\infty\left(B^{(1)}\otimes B^{(1)}\right)^m\text{vec}\left(B^{(0)}\Sigma (B^{(0)})^T\right)\\
     &=\sum_{m=0}^\infty\left((B^{(1)})^m\otimes (B^{(1)})^m\right)\text{vec}\left(B^{(0)}\Sigma (B^{(0)})^T\right)\\
     &=\sum_{m=0}^\infty\left((B^{(1)})^m\otimes (B^{(1)})^m\right)\cdot \left(B^{(0)}\otimes B^{(0)}\right)\text{vec}(\Sigma)\\
     &=\sum_{m=0}^\infty\left((B^{(1)})^mB^{(0)}\otimes (B^{(1)})^mB^{(0)}\right)\text{vec}(\Sigma).\numberthis
 \end{align*}
 Here, the second equality follows from standard rules about the sum of converging infinite sums and the multiplication of a scalar with a converging infinite sum; and the last three equations follow from standard rules about the Kronecker product and the vec-operator \citep[Sections A.11 and A.12]{lutkepohl2005new} and where in the last equation we use the convention that matrix multiplication is to be applied earlier than the Kronecker product.

Now, note that due to the diagonal structure of $\Sigma$, it holds that $\text{vec}(\Sigma)_{k} = 0$ for $k \notin \{1, d + 2, 2d+3,\ldots, (d-1)\cdot d + d = d^2\}$.
 
  Without loss of generality, let now $i = 1$ and start with $t_1=t$ and $t_2=t$.
Note that one can construct all treks between $S^1_t$ and $S^j_t$ with exactly $2m$-many lag-1 edges and arbitrary many lag-0 edges by choosing some top node $S^{k}_{t-m}$ with $k\in [d_1]$ and then concatenating all directed paths from $S^{k}_{t-m}$ to $S^1_t$ with all directed paths from $S^{k}_{t-m}$ to $S^j_t$. Note that the sum over all trek monomials for these treks (that is, treks using exactly $2m$-many lag-1 edges and arbitrary many lag-0 edges) is given by
  \begin{align}
 \label{vec_2}
     \sum_{k=1}^d\bigr((B^{(1)})^mB^{(0)}\bigr)_{1,k}\bigr((B^{(1)})^mB^{(0)}\bigr)_{j,k}\Sigma_{kk}.
 \end{align}
Now, note that the first $d$-entries of $\text{vec}(\Gamma_S(0))$ contain all the covariances between $S^1_t$ and $S^j_t$ and the $j$-th component with $j\in[d]_1$ of $\text{vec}(\Gamma_S(0))$ equals $\Gamma_{S^jS^1}=\Gamma_{S^1S^j}$. Also, note that the $j$-th row of
\begin{align*}
\left((B^{(1)})^mB^{(0)}\otimes (B^{(1)})^mB^{(0)}\right)\text{vec}(\Sigma), 
\end{align*}
 so the $j$-th row of the respective summand on the right-hand side of equation \eqref{vec_geo}, is exactly given by equation \eqref{vec_2}.
Thus, it follows that the $j$-th row of the right side of equation \eqref{vec_geo} 
equals
\begin{align*}
    &\sum_{m=0}^\infty \sum_{k=1}^d\bigr((B^{(1)})^mB^{(0)}\bigr)_{1,k}\bigr((B^{(1)})^mB^{(0)}\bigr)_{j,k}\Sigma_{kk}\\
    &=\sum_{m=0}^\infty \sum_{k=1}^d \sum_{\pi\in \mathcal{T}(S^1_t,S^j_t, S^k_{t-m})}\pi(A,\Sigma)\\
    &=\sum_{v\in V:\;t(v)\leq \min(t_1,t_2)}\;\; \sum_{\pi\in \mathcal{T}(S^1_t,S^j_t, v)}\pi(A,\Sigma),
\end{align*}
where the last line makes sense because the order of summation (over the top nodes) does not matter due to absolute convergence:
Note that due to standard inequalities on scalar products and norms,
\begin{align*}
\label{eq_absolute_summability}
    &\sum_{m=0}^\infty\sum_{k=1}^d\left|\bigr((B^{(1)})^mB^{(0)}\bigr)_{1,k}\bigr((B^{(1)})^mB^{(0)}\bigr)_{j,k}\Sigma_{kk}\right|\\
    &\leq \sum_{m=0}^\infty\sum_{k=1}^d\left(\bigr((B^{(1)})^mB^{(0)}\bigr)_{1,k}\right)^2 + \sum_{m=0}^\infty\sum_{k=1}^d\left(\bigr((B^{(1)})^mB^{(0)}\bigr)_{j,k}\Sigma_{kk}\right)^2\\
    &=\sum_{m=0}^\infty\sum_{k=1}^d\left(\sum_{l=1}^d\bigr((B^{(1)})^m\bigr)_{1,l}\bigr(B^{(0)}\bigr)_{l,k}\right)^2 + \sum_{m=0}^\infty\sum_{k=1}^d\left(\sum_{l=1}^d\bigr((B^{(1)})^m\bigr)_{j,l}\bigr(B^{(0)}\bigr)_{l,k}\Sigma_{kk}\right)^2\\
    &\leq 2\sum_{m=0}^\infty\sum_{k=1}^d\sum_{l=1}^d\left(\bigr((B^{(1)})^m\bigr)_{1,l}\bigr(B^{(0)}\bigr)_{l,k}\right)^2 + 2\sum_{m=0}^\infty\sum_{k=1}^d\sum_{l=1}^d\left(\bigr((B^{(1)})^m\bigr)_{j,l}\bigr(B^{(0)}\bigr)_{l,k}\Sigma_{kk}\right)^2\\
&= 2\sum_{k=1}^d\sum_{l=1}^d\left(\bigr(B^{(0)}\bigr)_{l,k}\right)^2\sum_{m=0}^\infty\left(\bigr((B^{(1)})^m\bigr)_{1,l}\right)^2 + 2\sum_{k=1}^d\sum_{l=1}^d\left(\bigr(B^{(0)}\bigr)_{l,k}\Sigma_{kk}\right)^2\sum_{m=0}^\infty\left(\bigr((B^{(1)})^m\bigr)_{j,l}\right)^2\\
&<\infty,\numberthis
\end{align*}
where the last equality follows again from a standard rule about finite sums of converging infinite series which is applicable due to the "$<\infty$" in the last line and where the "$<\infty$" follows because the stability assumption implies that each of these entry sequences is absolutely summable due to stability and thus square-summable \citep[Sections 2.1.1 and A9.1]{lutkepohl2005new}.

To get the trek rule for general  start and end points $S^1_t$ and $S^j_{t+h}$ with $h\in\mathbb{Z}$ (so $t_1=t$ and $t_2=t+h$), use the Yule-Walker equations \citep[Equation (2.1.31)]{lutkepohl2005new}
 \begin{align*}
     \Gamma_S(h):=(B^{(1)})^h \Gamma_S(0), \;\;\; \forall h\geq 0.
 \end{align*}
For $h\geq 0$, use that all treks from  $S^1_t$ going to $S^j_{t+h}$ can be constructed by taking all treks from $S^1_t$ to all $S^{l}_t$ with $l\in[d]_1$ concatenated with all directed paths starting with a lag-1 edge from the respective $S^{l}_t$ to $S^j_{t+h}$ (these directed paths need to start with a lag 1 edge, as otherwise, certain treks would be double counted because one considers \emph{all} $l\in[d]_1$). Now, as previously argued, the sum of all trek monomials for treks from $S^1_t$ to some $S^{l}_t$ equals the $l$-th component of $\text{vec}(\Gamma_S(0))$, and thus the $(l,1)$-th entry of $\Gamma_S(0)$. Also, the subsequent concatenation with all directed paths starting with a lag-1 edge from some $S^{l}_t$ to $S^j_{t+h}$ corresponds to multiplying (from the left) the $j$-th row of $(B^{(1)})^h$ with the first column
of $\Gamma_S(0)$. This multiplication exactly results in the sum of all trek monomials from $S^1_t$ going to $S^j_{t+h}$. Mathematically,
\begin{align*}
    \Gamma_{S^1_tS^j_{t+h}}&=\sum_{l=1}^d\bigr((B^{(1)})^h\bigr)_{j,l}\sum_{m=0}^\infty\sum_{k=1}^d\bigr((B^{(1)})^mB^{(0)}\bigr)_{1,k}\bigr((B^{(1)})^mB^{(0)}\bigr)_{l,k}\Sigma_{kk}\\
    &=\sum_{l=1}^d\sum_{m=0}^\infty\sum_{k=1}^d\bigr((B^{(1)})^h\bigr)_{j,l}\bigr((B^{(1)})^mB^{(0)}\bigr)_{1,k}\bigr((B^{(1)})^mB^{(0)}\bigr)_{l,k}\Sigma_{kk}\\
    &=\sum_{m=0}^\infty\sum_{k=1}^d\sum_{l=1}^d\bigr((B^{(1)})^h\bigr)_{j,l}\bigr((B^{(1)})^mB^{(0)}\bigr)_{1,k}\bigr((B^{(1)})^mB^{(0)}\bigr)_{l,k}\Sigma_{kk}\\
&=\sum_{m=0}^\infty\sum_{k=1}^d\bigr((B^{(1)})^mB^{(0)}\bigr)_{1,k}\bigr((B^{(1)})^{m+h}B^{(0)}\bigr)_{j,k}\Sigma_{kk}\\
    &=\sum_{m=0}^\infty\sum_{k=1}^d \sum_{\pi\in \mathcal{T}(S^1_t,S^j_{t+h}, S^k_{t-m})}\pi(A,\Sigma)\\
    &=\sum_{v\in V:\;t(v)\leq \min(t_1,t_2)}\;\; \sum_{\pi\in \mathcal{T}(S^1_t,S^j_{t+h}, v)}\pi(A,\Sigma),
\end{align*}
where the second equality follows from a standard rule about the multiplication of a scalar with a converging infinite sum and the third equality  follows from a standard rule about finite sums of converging infinite series and where the last line also makes sense because the order of summation (for the infinite summation) does not matter due to an argument analogous to \eqref{eq_absolute_summability}.
Thus, the Yule-Walker equations give the trek rule for all $h\geq 0$ as well.
For $h<0$, use that one can relabel the components of $\{S_t\}_{t\in\mathbb{Z}}$ such that $S^j$ and $S^1$ now flip roles and then apply the result for $h\geq 0$.

For general SVAR(p) processes: 

First, start by writing $B^{(s_w)}:=(I-A^{(0)})^{-1}A^{(s_w)}$ for all $s_w\in [p]_1$. Next, note that the sum of path monomials of all directed paths from $S^{k_1}_{t-h}$ to $S^{k_2}_{t}$ with $l,k\in [d]_1$ for $h> 0$ equals 
\begin{align}
\label{vec_3}
    \left(\left(\sum_{\substack{1\leq s_1,\ldots,s_r\leq p\\ 
			s_1+\ldots+s_r = h}}B^{(s_1)}\cdots B^{(s_r)}\right)B^{(0)}\right)_{k_2,k_1},
\end{align}
for $h=0$ equals $B^{(0)}_{k_2,k_1}$,
and for $h<0$ equals $0$.
	
One can understand Equation \eqref{vec_3} by using the block-matrix
		\begin{align}
		\label{big_A}
	\mathbf{B}:=
	\begin{pmatrix}
	B^{(1)} & B^{(2)} & \cdots & B^{(p-1)} & B^{(p)}\\
	I_d & 0 & \cdots & 0 & 0 \\
	0 & I_d & \cdots & 0 & 0 \\
	\vdots & \vdots & \ddots & \vdots & \vdots \\
	0 & 0 & \cdots & I_d & 0
	\end{pmatrix}.
	\end{align}
and taking a look at the exponents $\mathbf{B}^m$. We claim that for all $l\in\mathbb{N}_{\geq 0}$ such that $l<\min(m,p)$,
\begin{align}
\label{eq_exponents_big_A}
(\mathbf{B}^m)_{(l\cdot d +1,\ldots,l\cdot d + d),(1,\ldots,d)}=\sum_{\substack{1\leq s_1,\ldots,s_r\leq p\\ 
		s_1+\ldots+s_r = m-l}}B^{(s_1)}\cdots B^{(s_r)},
\end{align}
and if $m<p$, then for all $l$ with $m\leq l <p$
\begin{align}
\label{eq_exponents_big_A_2}
(\mathbf{B}^m)_{(l\cdot d +1,\ldots,l\cdot d + d),(1,\ldots,d)}=I_d\cdot 1_{l=m},
\end{align}
(Here, $(\mathbf{B}^m)_{(l\cdot d +1,\ldots,l\cdot d + d),(1,\ldots,d)}$ is the submatrix of $\mathbf{B}^m$ with rows $(l\cdot d +1,\ldots,l\cdot d + d)$ and columns $(1,\ldots,d)$). The claim from equation \eqref{eq_exponents_big_A_2} follows directly from the zero and unity matrix structure of $\mathbf{B}$. Regarding equation \eqref{eq_exponents_big_A}: One can prove equation \eqref{eq_exponents_big_A} by induction over $m$. Clearly, equation \eqref{eq_exponents_big_A} holds for $m=1$ as just $s_1=1$ in the summation index is possible. Now, because $\mathbf{B}^{m+1}=\mathbf{B}\cdot \mathbf{B}^m$,
we see that for all $l$ with $1\leq l < \min(m+1,p)$ due to the indicator matrix structure in $\mathbf{B}$ that
\begin{align*}
(\mathbf{B}^{m+1})_{(l\cdot d +1,\ldots,l\cdot d + d),(1,\ldots,d)}
	&=(\mathbf{B}^m)_{((l-1)\cdot d +1,\ldots,(l-1)\cdot d + d),(1,\ldots,d)}\\
	&=\sum_{\substack{1\leq s_1,\ldots,s_r\leq p\\ 
			s_1+\ldots+s_r = m-(l-1)}}B^{(s_1)}\cdots B^{(s_r)}\\
&=\sum_{\substack{1\leq s_1,\ldots,s_r\leq p\\ 
		s_1+\ldots+s_r = m+1-l}}B^{(s_1)}\cdots B^{(s_r)},
\end{align*}
where the second equality follows from the induction hypothesis.
For $l=0$ we have
\begin{align*}
(\mathbf{B}^{m+1})_{(l\cdot d +1,\ldots,l\cdot d + d),(1,\ldots,d)}&=\sum_{z=1}^{\min(m+1,p)}\biggr(B^{(z)}\cdot (\mathbf{B}^m)_{((z-1)\cdot d +1,\ldots,(z-1)\cdot d + d),(1,\ldots,d)}\biggr)\\
&=\begin{cases}
B^{(m+1)}+\sum_{z=1}^{m}\biggr(B^{(z)}\cdot \sum_{\substack{1\leq s_1,\ldots,s_r\leq p\\ 
		s_1+\ldots+s_r = m-(z-1)}}B^{(s_1)}\cdots B^{(s_r)}\biggr),\;\text{if }m<p\\
\sum_{z=1}^{p}\biggr(B^{(z)}\cdot \sum_{\substack{1\leq s_1,\ldots,s_r\leq p\\ 
		s_1+\ldots+s_r = m-(z-1)}}B^{(s_1)}\cdots B^{(s_r)}\biggr),\;\text{if }m\geq p
\end{cases}\\
	&=\sum_{\substack{1\leq s_1,\ldots,s_r\leq p\\ 
			s_1+\ldots+s_r = m+1}}B^{(s_1)}\cdots B^{(s_r)},
\end{align*}
where the second equality again follows from the induction hypothesis.

Now, using the block matrix $\mathbf{B}$, one can rewrite each SVAR(p) process as a VAR(1) process. In particular, it holds that $\mathbf{B}$ has all eigenvalues with modulus strictly less than one if and only if the stability condition holds \citep[Section 2.1.1]{{lutkepohl2005new}}. Moreover, one has \citep[Equation (2.1.39)]{{lutkepohl2005new}}
 \begin{align*}
\text{vec}(\Tilde{\Gamma})
    =(I_{(dp)^2}-\mathbf{B}\otimes \mathbf{B})^{-1}\text{vec}(\Tilde{\Sigma}),
 \end{align*}
where $\Tilde{\Sigma}$ is a $(dp)\times(dp)$-matrix for which the $d\times d$ left-upper submatrix is given by $\Tilde{\Sigma}=B^{(0)}\Sigma (B^{(0)})^T$ and all other entries are zero and where 
\begin{align*}
    \Tilde{\Gamma}:=	
     \begin{pmatrix}
		\Gamma_S(0) & \Gamma_S(1) & \ldots & \Gamma_S(p-1)\\
		\Gamma_S(-1) & \Gamma_S(0) & \ldots & \Gamma_S(p-2)\\
		\vdots & \vdots & \ddots & \vdots\\
		\Gamma_S(-(p-1)) & \Gamma_S(-(p - 2)) & \ldots & \Gamma_S(0)
		\end{pmatrix}.
\end{align*}

Now again using the von Neumann series, we get 
 \begin{align*}
 \label{eq_vec_geo2}
\text{vec}(\Tilde{\Gamma})
&=\left(\sum_{m=0}^\infty \left(\mathbf{B}\otimes \mathbf{B}\right)^m\right)\text{vec}(\Tilde{\Sigma})\\
&=\sum_{m=0}^\infty \left(\mathbf{B}\otimes \mathbf{B}\right)^m\text{vec}(\Tilde{\Sigma})\\
    &=\sum_{m=0}^\infty \left(\mathbf{B}^m\otimes \mathbf{B}^m\right)\text{vec}(\Tilde{\Sigma})\\
    &=\sum_{m=0}^\infty \left(\mathbf{B}^m\otimes \mathbf{B}^m\right)_{\cdot, (1,\ldots,d,pd+1,\ldots, pd+d,\ldots\ldots,(d-1)pd+1,\ldots,(d-1)pd+d)}\text{vec}\left(B^{(0)}\Sigma (B^{(0)})^T\right)\\
    &=\sum_{m=0}^\infty \left(\mathbf{B}^m_{\cdot, (1,\ldots,d)}\otimes \mathbf{B}^m_{\cdot, (1,\ldots,d)}\right)\text{vec}\left(B^{(0)}\Sigma (B^{(0)})^T\right)\\
        &=\sum_{m=0}^\infty \left(\mathbf{B}^m_{\cdot, (1,\ldots,d)}B^{(0)}\otimes \mathbf{B}^m_{\cdot, (1,\ldots,d)}B^{(0)}\right)\text{vec}\left(\Sigma\right)\numberthis,
 \end{align*}
where $\mathbf{B}^m_{\cdot, (1,\ldots,d)}$ denotes the submatrix of $\mathbf{B}^m$ consisting of all rows and the first $d$-columns---analogously $\left(\mathbf{B}^m\otimes \mathbf{B}^m\right)_{\cdot, (1,\ldots,d,pd+1,\ldots, pd+d,\ldots\ldots,(d-1)pd+1,\ldots,(d-1)pd+d)}$ denotes the submatrix with all columns of the form $kpd+j$ where $k\in[d-1]_0$ and $j\in[d]_1$ and with all rows---and where the third last equality follows from the zero structure of $\Tilde{\Sigma}$.

Again, let now $t_1=t$ and $t_2=t-(l-1)$ with $l\in[p]_1$ and without loss of generality, $i=1$.
For $j\in[d]_1$, the $j$-th row of $\text{vec}(\Tilde{\Gamma})$ equals $\Gamma_{S^1_{t}S^j_{t}}=\Gamma_{S^j_{t}S^1_{t}}$, and in more generality, the $(l-1)\cdot d+ j$-th row with $j\in[d]_1$ and $l\in[p]_1$ of $\text{vec}(\Tilde{\Gamma})$ equals $\Gamma_{S^1_{t+(l-1)}S^j_{t}}=\Gamma_{S^1_{t}S^j_{t-(l-1)}}$. Similarly, the $(l-1)\cdot d+ j$-th row with $j\in[d]_1$ and $l\in[p]_1$ of 
\begin{align*}
     \left(\mathbf{B}^m_{\cdot, (1,\ldots,d)}B^{(0)}\otimes \mathbf{B}^m_{\cdot, (1,\ldots,d)}B^{(0)}\right)\text{vec}\left(\Sigma\right)
\end{align*}
equals
\begin{align*}
\label{eq_sum_top_nodes_general_p}
 &\sum_{k=1}^d((\mathbf{B}^m)_{\cdot, (1,\ldots,d)}B^{(0)})_{1,k}\cdot(\mathbf{(B}^m)_{\cdot, (1,\ldots,d)}B^{(0)})_{(l-1)\cdot d+ j,k}\cdot\Sigma_{kk}\\
  &=\sum_{k=1}^d(\mathbf{B}^m)_{1, (1,\ldots,d)}(B^{(0)})_{\cdot,k}\cdot(\mathbf{B}^m)_{(l-1)\cdot d+ j,(1,\ldots,d)}(B^{(0)})_{\cdot,k}\cdot\Sigma_{kk}\\
		&=\sum_{k=1}^d\underbrace{\left(\left( \sum_{\substack{1\leq s_1,\ldots,s_r\leq p\\ 
		s_1+\ldots+s_r = m}}B^{(s_1)}\cdots B^{(s_r)}\right)_{1, (1,\ldots,d)}(B^{(0)})_{\cdot,k}1_{m\geq 1}+(B^{(0)})_{1,k}1_{m=0}\right)}_{\text{Sum of all path monomials of all directed paths from $S^k_{t-m}$ to $S^1_t$, see  eq.\ \eqref{vec_3}}}\\
		&\hspace{0.5cm}\cdot    \underbrace{\left(  \left( \sum_{\substack{1\leq s_1,\ldots,s_r\leq p\\ 
		s_1+\ldots+s_r = m-(l-1)}}B^{(s_1)}\cdots B^{(s_r)}\right)_{ j,(1,\ldots,d)}(B^{(0)})_{\cdot,k}1_{l-1 < m}+(B^{(0)})_{j,k}\cdot 1_{l-1=m}\right)}_{\text{Sum of all path monomials of all directed paths from $S^k_{t-m}$ to $S^j_{t-(l-1)}$, see  eq.\ \eqref{vec_3}}}\\
		&\hspace{0.5cm}\cdot\Sigma_{kk}.\numberthis
\end{align*}
Here, the second equality holds due to equations \eqref{eq_exponents_big_A} and \eqref{eq_exponents_big_A_2}. Furthermore, for the second equality we used that $l\in[p]_1$ by assumption and thus $l-1<p$ holds, and in addition, if $l-1=m$, then because as just argued, $l-1<p$, it holds that $m<p$ as well.

From the same concatenation-of-directed-paths argument as for SVAR(1) processes it now follows that
\begin{align*}
 \text{vec}(\Tilde{\Gamma})_{(l-1)\cdot d + j}=   \Gamma_{S^1_{t}S^j_{t-(l-1)}} &=\sum_{m=0}^\infty \left(\left(\mathbf{B}^m_{\cdot, (1,\ldots,d)}B^{(0)}\otimes \mathbf{B}^m_{\cdot, (1,\ldots,d)}B^{(0)}\right)\text{vec}\left(\Sigma\right)\right)_{(l-1)\cdot d+j}\\
    &=\sum_{m=0}^\infty \sum_{k=1}^d \sum_{\pi\in \mathcal{T}(S^1_{t},S^j_{t-(l-1)}, S^k_{t-m})}\pi(A,\Sigma)\\\\
     &=\sum_{v\in V:\;t(v)\leq \min(t_1,t_2)}\;\; \sum_{\pi\in \mathcal{T}(S^1_t,S^j_{t-(l-1)}, v)}\pi(A,\Sigma),
\end{align*}
where the last line again makes sense due to absolute convergence: Here, absolute convergence follows from an analogous argument as for SVAR(1) processes and the fact that the entries of $\mathbf{B}^m$ are absolutely summable and thus square summable \citep[Sections 2.1.1 and A9.1]{lutkepohl2005new}. Finally, using the general Yule-Walker equations \citep[Equation (2.1.36)]{{lutkepohl2005new}}
\begin{align*}
    \Gamma(h):= B^{(1)}\Gamma(h-1)+\ldots+B^{(p)}\Gamma(h-p),\;\;\;\forall h > 0,
\end{align*}
an analogous argument as for SVAR(1) processes, and the fact that the sum of absolutely convergent infinite series is also absolutely convergent due to the triangle inequality, yields the result.
\end{proof}

\subsection{Example for Remark \ref{remark_counterexample_trek_rule}}
\label{counter_ex_trek_rule}
We here present the example mentioned in Remark \ref{remark_counterexample_trek_rule}.
\addtocounter{theorem}{-20}
\begin{example}
We adapt this example from Section 2.3 in \cite{hochsprungglobal}.
	Consider the process $\{S_t\}_{t\in\mathbb{Z}}:=\{(S^1_t,S^2_t)\}_{t\in\mathbb{Z}}$ defined for all $t\in\mathbb{Z}$ by
	\begin{align*}
	S^1_t&=-\sum_{i=1}^\infty \phi^{-i}(\epsilon^1_{t+i}+\epsilon^2_{t+i-1}),\;\textnormal{and}\\
	S^2_t&=\epsilon^2_t
	\end{align*}
	where $\phi\in\mathbb{R}$ with $|\phi|>1$. Note that $\{\phi^{-i}\}_{i\in\mathbb{N}_{\geq 1}}$ is absolutely summable and hence $S^1_t$, interpreted as a limit in mean-square, exists and is uniquely defined up to $P$-nullsets (see Section C.3 in \cite{lutkepohl2005new}).
	
	It can be shown that $\{S_t\}_{t\in\mathbb{Z}}$ (see Section 2.3 in \cite{hochsprungglobal}) satisfies the VAR(1) (and thus SVAR(1)) equations
	\begin{align*}
	S^1_t&=\phi \cdot S^1_{t-1}+S^2_{t-1}+\epsilon^1_t,\;\textnormal{and} \\
	S^2_t&=\epsilon^2_t.
	\end{align*}
	Note that this VAR(1) process does not satisfy the stability assumption: We have
	\begin{align*}
	A^{(1)}=\begin{pmatrix}
	\phi & 1 \\
	0 & 0\\
	\end{pmatrix}
	\end{align*}
	and thus, $\det(I_2-A^{(1)}\lambda) = (1-\lambda \phi)$. Because $(1-\lambda \phi)$ has root $\lambda = 1/\phi$ and because $|\phi|>1$ by assumption, this VAR(1) process is not stable.
	
	Now, in the full time graph, there is no trek between $S^1_t$ and $S^2_{t+1}$. However, the covariance between $S^1_t$ and $S^2_{t+1}=\epsilon^2_{t+1}$ equals $\phi^{-2}\neq 0$.
\end{example}		
	\addtocounter{theorem}{19}
\subsection{Proof of Lemma \ref{lemma_parent_decomp}}
\label{sec_proof_parent_decomp}
This section contains the proof of Lemma \ref{lemma_parent_decomp}.
	\begin{proof}
		By the trek rule (Lemma \ref{lemma_trek_rule_var_processes}), we have 
		\begin{align*}
		\Gamma_{S^{i}_{t_1}S^{j}_{t_2}} = \Gamma_{S^{j}_{t_2}S^{i}_{t_1}} = \sum_{v\in V:\;t(v)\leq \min(t_1,t_2)}\;\;\sum_{\pi \in \mathcal{T}(S^{i}_{t_1},S^{j}_{t_2},v)}\pi(A, \Sigma).
		\end{align*}
		Similarly for all $q\in \textnormal{pa}(S^{j}_{t_2})$,
		\begin{align*}
		\Gamma_{qS^{i}_{t_1}}=\Gamma_{S^{i}_{t_1}q}&=\sum_{v\in V:\;t(v)\leq \min(t_1,t(q))}\;\;\sum_{\pi' \in \mathcal{T}(S^{i}_{t_1},q,v)}\pi'(A, \Sigma)\\
		&=\sum_{v\in V:\;t(v)\leq \min(t_1,t_2)}\;\;\sum_{\pi' \in \mathcal{T}(S^{i}_{t_1},q,v)}\pi'(A, \Sigma),
		\end{align*}
		where the last equality follows because if $t(q)<t_2$, then no trek from $S^i_{t_1}$ to $q$ has a top node at time $t_2$ as every non-top node in a trek must have a larger or equal time index than the top node.
		
		Because $S^{i}_{t_1}$ is a non-descendant of $S^{j}_{t_2}$ (so in particular $S^{i}_{t_1}\neq S^{j}_{t_2}$), every trek from $S^{i}_{t_1}$ to $S^{j}_{t_2}$ ends with an edge $q\rightarrow  S^{j}_{t_2}$ for some $q\in \textnormal{pa}(S^{j}_{t_2})$. Thus, every trek in $\mathcal{T}(S^{i}_{t_1},S^{j}_{t_2})$ can be decomposed into a trek from $\mathcal{T}(S^{i}_{t_1}, q)$ and the corresponding edge  $q\rightarrow  S^{j}_{t_2}$ for some $q\in \textnormal{pa}(S^{j}_{t_2})$.

		Thus, for every trek $\pi \in \mathcal{T}(S^{i}_{t_1},S^{j}_{t_2})$ we can decompose the trek monomial $\pi(A,\Sigma)$ into a product of the trek monomial of some trek $\pi'$ from some $\mathcal{T}(S^{i}_{t_1}, q_{\pi})$ and the corresponding edge coefficient from $A^{(k_{q_\pi})}_{S^j_{t_2}q_\pi}$, that is, 
		\begin{align*}
		\pi(A,\Sigma) = \pi'(A,\Sigma) \cdot A^{(k_{q_\pi})}_{S^j_{t_2}q_\pi}.
		\end{align*}

		Vice versa, every trek from some $\mathcal{T}(S^{i}_{t_1}, q)$ can be concatenated with the corresponding edge $q\rightarrow S^{j}_{t_2}$ to make a trek in $\mathcal{T}(S^{i}_{t_1},S^{j}_{t_2})$. 
		Therefore, there is a one-to-one correspondence between treks in $\mathcal{T}(S^{i}_{t_1},S^{j}_{t_2})$ and treks in $\bigcup_{q\in \textnormal{pa}(S^{j}_{t_2})}\mathcal{T}(S^{i}_{t_1}, q)$   concatenated with the edges from $\bigcup_{q\in \textnormal{pa}(S^{j}_{t_2})}\{q\rightarrow S^{j}_{t_2}\}$. 
		
		Hence, we can write
			\begin{align*}
	&\sum_{q\in \textnormal{pa}(S^{j}_{t_2})} A^{(k_q)}_{S^{j}_{t_2}q}\Gamma_{S^{i}_{t_1}q}\\
		&=	\sum_{q\in \textnormal{pa}(S^{j}_{t_2})}\;\;A^{(k_q)}_{S^{j}_{t_2}q}\sum_{v\in V:\;t(v)\leq \min(t_1,t_2)}\;\;\sum_{\pi' \in \mathcal{T}(S^{i}_{t_1},q,v)}\pi'(A, \Sigma)\\
	&=	\sum_{q\in \textnormal{pa}(S^{j}_{t_2})}\;\;\sum_{v\in V:\;t(v)\leq \min(t_1,t_2)}\;\;\sum_{\pi' \in \mathcal{T}(S^{i}_{t_1},q,v)}A^{(k_q)}_{S^{j}_{t_2}q}\pi'(A, \Sigma)\\
	&=	\sum_{v\in V:\;t(v)\leq \min(t_1,t_2)}\;\;\sum_{q\in \textnormal{pa}(S^{j}_{t_2})}\;\;\sum_{\pi' \in \mathcal{T}(S^{i}_{t_1},q,v)}A^{(k_q)}_{S^{j}_{t_2}q}\pi'(A, \Sigma)\\
	&=\sum_{v\in V:\;t(v)\leq \min(t_1,t_2)}\;\;\sum_{\pi \in \mathcal{T}(S^{i}_{t_1},S^{j}_{t_2},v)}\pi(A, \Sigma),
		\end{align*}
	where the second and third equality follow from the absolute summability part in Lemma \ref{lemma_trek_rule_var_processes}, and a standard rule about the multiplication of a scalar with an absolutely converging infinite sum and the fact that finite sums of absolutely converging infinite series are also absolutely converging due to the triangle inequality, respectively (absolute convergence instead of convergence is required here, as the order of summation in each of the previous equations must not matter for the expressions to make sense). 
	\end{proof}
	
\subsection{Lemma \ref{lemma_determinant}}
\label{sec_proof_lemma_determinant}
\begin{lemma} [Calculating determinants of matrices $\Gamma_{R,C}$]
	\label{lemma_determinant}
		Assume a stable SVAR process satisfying Assumptions \ref{assumption_no_instantaneous_self_edges} and \ref{assumption_acyclicity}. Let $R$ and $S$ be finite sets of vertices from some underlying stable SVAR process such that $|R|=|C|<\infty$. Let $v_1,v_2,\ldots$ be any enumeration of the vertices in the countable set $\{v\in V:\;t(v)\leq t_{\textnormal{max}}(R\cup S)\}$. Then, it holds that
	\begin{align}
	\label{eq_determinant}
		\sum_{l=1}^\infty\sum_{\substack{l_1+\ldots+l_n=l\\l_1,\ldots,l_n\in\mathbb{N}_{\geq 1}}}\;\;\sum_{\substack{\Pi=\{\pi_1,\ldots,\pi_n\}: R\rightrightarrows S\\\text{s.t.\ }\pi_i \text{ has top node }v_{l_i}\\
		\text{and s.t.\ }\Pi \text{ has no sided intersection}} }\textnormal{sgn}(\Pi)\Pi(A,\Sigma).
	\end{align}
	Here, the sum on the right-hand side of equation \eqref{eq_determinant} has typically infinitely many summands and is meant as a power series of the diagonal elements of $\Sigma$ and the non-zero entries of $A^{(0)},\ldots,A^{(p)}$. Furthermore,
	\begin{align*}
	   \sum_{l=1}^\infty\left|\sum_{\substack{l_1+\ldots+l_n=l\\l_1,\ldots,l_n\in\mathbb{N}_{\geq 1}}}\;\;\sum_{\substack{\Pi=\{\pi_1,\ldots,\pi_n\}: R\rightrightarrows S\\\text{s.t.\ }\pi_i \text{ has top node }v_{l_i}\\
		\text{and s.t.\ }\Pi \text{ has no sided intersection}} }\textnormal{sgn}(\Pi)\Pi(A,\Sigma)\right| < \infty.
	\end{align*}
\end{lemma}
\begin{proof}
	Let $n:=|R|=|C|$, let $\mathcal{S}_n$ denote the symmetric group of degree $n$ and let $\sigma\in \mathcal{S}_n$ denote a permutation in that symmetric group. From the trek rule it follows
	that
	\begin{align*}
	    \Gamma_{r_i,s_{\sigma(i)}}&=\sum_{v\in V:\;t(v)\leq \min(t(r_i),t(s_{\sigma(i)}))}\;\;\sum_{\pi \in \mathcal{T}(r_i,s_{\sigma(i)},v)}\pi(A, \Sigma)\\
	    &=\sum_{v\in V:\;t(v)\leq t_{\textnormal{max}}(R\cup S)}\;\;\sum_{\pi \in \mathcal{T}(r_i,s_{\sigma(i)},v)}\pi(A, \Sigma)\\
	    &=\sum_{l=1}^\infty\sum_{\pi \in \mathcal{T}(r_i,s_{\sigma(i)},v_l)}\pi(A, \Sigma),
	\end{align*}
	where the second equation is true because if $\min(t(r_i),t(s_{\sigma(i)}))<t_{\textnormal{max}}(R\cup S)$, then no further treks are included because every non-top node of a trek has a time index larger than or equal to the time index of the respective top node. Also here, $v_1,v_2,\ldots$ is any enumeration of the vertices in $\{v\in V:\;t(v)\leq t_{\textnormal{max}}(R\cup S)\}$ (the third equality here holds true for \emph{every} enumeration due to absolute convergence as stated in Lemma \ref{lemma_trek_rule_var_processes}).
	By definition of the determinant and by the trek rule (Lemma \ref{lemma_trek_rule_var_processes}), we have
	\begin{align*}
	\det(\Gamma_{R,C})&=\sum_{\sigma\in S_n}\biggr (\textnormal{sgn}(\sigma)\prod_{i=1}^n\Gamma_{r_i,s_{\sigma(i)}}\biggr)\\
	&=\sum_{\sigma\in S_n}\textnormal{sgn}(\sigma)\prod_{i=1}^n\sum_{l=1}^\infty\;\;\sum_{\pi \in \mathcal{T}(r_i,s_{\sigma(i)},v_l)}\pi(A, \Sigma)\\
	&=\sum_{\sigma\in S_n}\textnormal{sgn}(\sigma)\sum_{l=1}^\infty\sum_{\substack{l_1+\ldots+l_n=l\\l_1,\ldots,l_n\in\mathbb{N}_{\geq 1}}}\;\;\prod_{i=1}^n\;\;\sum_{\pi \in \mathcal{T}(r_i,s_{\sigma(i)},v_{l_i})}\pi(A, \Sigma)\\
		&=\sum_{\sigma\in S_n}\textnormal{sgn}(\sigma)\sum_{l=1}^\infty\sum_{\substack{l_1+\ldots+l_n=l\\l_1,\ldots,l_n\in\mathbb{N}_{\geq 1}}}\;\;\sum_{\substack{\pi_1 \in \mathcal{T}(r_1,s_{\sigma(1)},v_{l_1})\\\ldots\\\pi_n \in \mathcal{T}(r_n,s_{\sigma(n)},v_{l_n})}}\;\;\prod_{i=1}^n\pi_i(A, \Sigma).
	\end{align*}
	Here, the third equality is due to Cauchy's product formula [e.g., Theorem 3.50 in \citealp{rudin1976principles}] and absolute convergence as stated in Lemma \ref{lemma_trek_rule_var_processes}  and the last equality follows due to a standard exchange of finite product and sum. Also note that due to Cauchy's product formula and absolute convergence [e.g., Exercise 13 in Chapter 3 in \citealp{rudin1976principles}],
	\begin{align*}
	\label{abs_sum_cauchy}
	    &\sum_{l=1}^\infty\left|\sum_{\substack{l_1+\ldots+l_n=l\\l_1,\ldots,l_n\in\mathbb{N}_{\geq 1}}}\;\;\sum_{\substack{\pi_1 \in \mathcal{T}(r_1,s_{\sigma(1)},v_{l_1})\\\ldots\\\pi_n \in \mathcal{T}(r_n,s_{\sigma(n)},v_{l_n})}}\;\;\prod_{i=1}^n\pi_i(A, \Sigma)\right|\\
	    &=\sum_{l=1}^\infty\left|\sum_{\substack{l_1+\ldots+l_n=l\\l_1,\ldots,l_n\in\mathbb{N}_{\geq 1}}}\;\;\prod_{i=1}^n\;\;\sum_{\pi \in \mathcal{T}(r_i,s_{\sigma(i)},v_{l_i})}\pi(A, \Sigma)\right|<\infty \numberthis
	\end{align*}
	
	Each combination of treks $\pi_1\in \mathcal{T}(r_1, s_{\sigma(1)}),\ldots,\pi_n\in \mathcal{T}(r_n, s_{\sigma(n)})$ induces a system of treks from $R$ to $S$ with permutation $\sigma$. Moreover, all systems of treks with permutation $\sigma$ can be decomposed into treks $\pi_1\in \mathcal{T}(r_1, s_{\sigma(1)}),\ldots,\pi_n\in \mathcal{T}(r_n, s_{\sigma(n)})$. Therefore, for all $\sigma\in \mathcal{S}_n$,
	\begin{align*}
\sum_{\substack{\pi_1 \in \mathcal{T}(r_1,s_{\sigma(1)},v_{l_1})\\\ldots\\\pi_n \in \mathcal{T}(r_n,s_{\sigma(n)},v_{l_n})}}\;\;\prod_{i=1}^n\pi_i(A, \Sigma)&=\sum_{\substack{\Pi=\{\pi_1,\ldots,\pi_n\}: R\rightrightarrows (s_{\sigma(1)},\ldots,s_{\sigma(n)})\\\text{s.t.\ }\pi_i \text{ has top node }v_{l_i}} }\Pi(A,\Sigma).
	\end{align*}
	Thus,
	\begin{align*}
	\label{formula_sys_treks}
&	\sum_{\sigma\in S_n}\textnormal{sgn}(\sigma)\sum_{l=1}^\infty\sum_{\substack{l_1+\ldots+l_n=l\\l_1,\ldots,l_n\in\mathbb{N}_{\geq 1}}}\;\;\sum_{\substack{\pi_1 \in \mathcal{T}(r_1,s_{\sigma(1)},v_{l_1})\\\ldots\\\pi_n \in \mathcal{T}(r_n,s_{\sigma(n)},v_{l_n})}}\;\;\prod_{i=1}^n\pi_i(A, \Sigma)\\
	&=	\sum_{\sigma\in S_n}\textnormal{sgn}(\sigma)\sum_{l=1}^\infty\sum_{\substack{l_1+\ldots+l_n=l\\l_1,\ldots,l_n\in\mathbb{N}_{\geq 1}}}\;\;\sum_{\substack{\Pi=\{\pi_1,\ldots,\pi_n\}: R\rightrightarrows (s_{\sigma(1)},\ldots,s_{\sigma(n)})\\\text{s.t.\ }\pi_i \text{ has top node }v_{l_i}} }\Pi(A,\Sigma)\\
	&=	\sum_{l=1}^\infty\sum_{\substack{l_1+\ldots+l_n=l\\l_1,\ldots,l_n\in\mathbb{N}_{\geq 1}}}\;\;\sum_{\sigma\in S_n}\;\;\sum_{\substack{\Pi=\{\pi_1,\ldots,\pi_n\}: R\rightrightarrows (s_{\sigma(1)},\ldots,s_{\sigma(n)})\\\text{s.t.\ }\pi_i \text{ has top node }v_{l_i}} }\textnormal{sgn}(\sigma)\Pi(A,\Sigma)\\	
		&=	\sum_{l=1}^\infty\sum_{\substack{l_1+\ldots+l_n=l\\l_1,\ldots,l_n\in\mathbb{N}_{\geq 1}}}\;\;\sum_{\substack{\Pi=\{\pi_1,\ldots,\pi_n\}: R\rightrightarrows S\\\text{s.t.\ }\pi_i \text{ has top node }v_{l_i}} }\textnormal{sgn}(\Pi)\Pi(A,\Sigma).\numberthis
	\end{align*}
	Here, for the second equality we used standard rules for finite sums of converging infinite series and for a constant factor multiplied to an infinite series.

	 Now, note that if $\Pi$ has sided-intersection, then one can always find another system of treks $\Pi'$ where two targets and the corresponding subtreks from that intersection to these targets are exchanged. This exchange exactly changes the sign of $\Pi'$ relative to the sign of $\Pi$ by the factor $-1$, however, that leaves the monomial $\Pi(A,\Sigma)$ unchanged, that is, $\Pi(A,\Sigma) = \Pi'(A, \Sigma)$. This exchange also either does not or does flip the two top nodes of the corresponding two treks, and hence, both $\Pi$ and $\Pi'$ occur in the sum 
	 \begin{align*}
	     \sum_{\substack{l_1+\ldots+l_n=l\\l_1,\ldots,l_n\in\mathbb{N}_{\geq 1}}}\;\;\sum_{\substack{\Pi=\{\pi_1,\ldots,\pi_n\}: R\rightrightarrows S\\\text{s.t.\ }\pi_i \text{ has top node }v_{l_i}} }\textnormal{sgn}(\Pi)\Pi(A,\Sigma).
	 \end{align*}
Therefore,
	 one can restrict the sum in equation \eqref{formula_sys_treks} to systems of treks with no sided intersection as terms for treks with sided intersection cancel out.
	 
	 From equations \eqref{abs_sum_cauchy} and \eqref{formula_sys_treks} it also follows that
	 \begin{align*}
	     	&\sum_{l=1}^\infty\left|\sum_{\substack{l_1+\ldots+l_n=l\\l_1,\ldots,l_n\in\mathbb{N}_{\geq 1}}}\;\;\sum_{\substack{\Pi=\{\pi_1,\ldots,\pi_n\}: R\rightrightarrows S\\\text{s.t.\ }\pi_i \text{ has top node }v_{l_i}} }\textnormal{sgn}(\Pi)\Pi(A,\Sigma)\right|\\
	     	&=\sum_{l=1}^\infty\left|\sum_{\substack{l_1+\ldots+l_n=l\\l_1,\ldots,l_n\in\mathbb{N}_{\geq 1}}}\;\;\sum_{\sigma\in S_n}\;\;\sum_{\substack{\Pi=\{\pi_1,\ldots,\pi_n\}: R\rightrightarrows (s_{\sigma(1)},\ldots,s_{\sigma(n)})\\\text{s.t.\ }\pi_i \text{ has top node }v_{l_i}} }\textnormal{sgn}(\sigma)\Pi(A,\Sigma)\right|\\
	     	&=\sum_{l=1}^\infty\left|\sum_{\sigma\in S_n}\;\;\textnormal{sgn}(\sigma)\sum_{\substack{l_1+\ldots+l_n=l\\l_1,\ldots,l_n\in\mathbb{N}_{\geq 1}}}\;\;\sum_{\substack{\Pi=\{\pi_1,\ldots,\pi_n\}: R\rightrightarrows (s_{\sigma(1)},\ldots,s_{\sigma(n)})\\\text{s.t.\ }\pi_i \text{ has top node }v_{l_i}} }\Pi(A,\Sigma)\right|\\
	     	&\leq \sum_{l=1}^\infty\sum_{\sigma\in S_n}\;\;|\textnormal{sgn}(\sigma)|\cdot\left|\sum_{\substack{l_1+\ldots+l_n=l\\l_1,\ldots,l_n\in\mathbb{N}_{\geq 1}}}\;\;\sum_{\substack{\Pi=\{\pi_1,\ldots,\pi_n\}: R\rightrightarrows (s_{\sigma(1)},\ldots,s_{\sigma(n)})\\\text{s.t.\ }\pi_i \text{ has top node }v_{l_i}} }\Pi(A,\Sigma)\right|\\
	     	&=\sum_{\sigma\in S_n}\sum_{l=1}^\infty\left|\sum_{\substack{l_1+\ldots+l_n=l\\l_1,\ldots,l_n\in\mathbb{N}_{\geq 1}}}\;\;\sum_{\substack{\pi_1 \in \mathcal{T}(r_1,s_{\sigma(1)},v_{l_1})\\\ldots\\\pi_n \in \mathcal{T}(r_n,s_{\sigma(n)},v_{l_n})}}\;\;\prod_{i=1}^n\pi_i(A, \Sigma)\right|<\infty
	 \end{align*}
	 where the second equality follows from standard rules about finite sums, where the ``$\leq$'' follows from the triangle inequality and the multiplicativity of the absolute value, where the last equality again follows from a standard rule about finite sums of converging infinite series and where the "$<\infty$" follows from \eqref{abs_sum_cauchy}.
\end{proof}

\subsection{Lemma \ref{lemma_unique_solvability_new}}	\label{sec_unique_solvability_new}

\begin{lemma} [Deciding invertibility of matrices $\Gamma_{R,C}$]
	\label{lemma_unique_solvability_new}
Assume a stable SVAR process satisfying Assumptions \ref{assumption_no_instantaneous_self_edges} and \ref{assumption_acyclicity}. Let $R$ and $S$ be finite sets of vertices from some underlying stable SVAR process such that $|R|=|C|<\infty$.	If there exists a system of treks $\Psi:R\rightrightarrows C$ with no sided intersection such that  $\Psi(A,\Sigma)\neq \Pi(A,\Sigma)$ for all other system of treks $\Pi:R \rightrightarrows C$ with no sided intersection, then $\Gamma_{R,C}$ is invertible in generic settings.
\end{lemma}
\begin{proof}
	By Lemma \ref{lemma_determinant},
		\begin{align}
		\label{eq_det_gamma}
	\det(\Gamma_{R,C})=\sum_{l=1}^\infty\sum_{\substack{l_1+\ldots+l_n=l\\l_1,\ldots,l_n\in\mathbb{N}_{\geq 1}}}\;\;\sum_{\substack{\Pi=\{\pi_1,\ldots,\pi_n\}: R\rightrightarrows S\\\text{s.t.\ }\pi_i \text{ has top node }v_{l_i}\\
		\text{and s.t.\ }\Pi \text{ has no sided intersection}} }\textnormal{sgn}(\Pi)\Pi(A,\Sigma),
	\end{align}
	where $v_1,v_2,\ldots$ is any enumeration of the vertices in the countable set $\{v\in V:\;t(v)\leq t_{\textnormal{max}}(R\cup S)\}$.
	If there is a system of treks $\Psi$ with no sided intersection and unique  monomial $\Psi(A,\Sigma)$ among all other systems of treks from $R$ to $C$ with no sided intersection, then the coefficient in front of $\Psi(A,\Sigma)$ in equation \eqref{eq_det_gamma} is $\textnormal{sgn}(\Psi)$ and hence, $\det(\Gamma_{R,C})$ is not the zero power series.

Now, note that $\Tilde{\Theta}_{\textnormal{stable}}(G)\times \mathbb{R}^d_{> 0}$ (where 	$\Tilde{\Theta}_{\textnormal{stable}}(G)$ is a superset of 	$\Theta_{\textnormal{stable}}(G)$ representing all possible stable parameters for \emph{subgraphs} of $G$, see Section \ref{sec_proofs_sec_causal_effects} for a formal definition) is a connected open set due to Lemma \ref{lemma_lebesgue} in Section \ref{sec_proofs_sec_causal_effects} of the Appendix and because  $\mathbb{R}^d_{> 0}$ is a connected open set and because the Cartesian product of connected open sets is connected and open. Furthermore, because  for Banach spaces (and thus real-valued vector spaces) absolute convergence implies convergence and because power series with positive radius of convergence are real analytic in the interior of their radius of convergence (see, for example, Proposition 2.2.7 in \citet{krantz2002primer}),
	it follows that the function
	\begin{align*}
	&f:\Tilde{\Theta}_{\textnormal{stable}}(G)\times \mathbb{R}^d_{\geq 0}\rightarrow \mathbb{R},\\
	&\hspace{0.5cm}(A,\Sigma)\mapsto\sum_{l=1}^\infty\sum_{\substack{l_1+\ldots+l_n=l\\l_1,\ldots,l_n\in\mathbb{N}_{\geq 1}}}\;\;\sum_{\substack{\Pi=\{\pi_1,\ldots,\pi_n\}: R\rightrightarrows S\\\text{s.t.\ }\pi_i \text{ has top node }v_{l_i}\\
		\text{and s.t.\ }\Pi \text{ has no sided intersection}} }\textnormal{sgn}(\Pi)\Pi(A,\Sigma),
	\end{align*}
	restricted to $\Tilde{\Theta}_{\textnormal{stable}}(G)\times \mathbb{R}^d_{> 0}$
	is real analytic on a connected open domain.
	From Lemma \ref{lemma_zero_sets_real_analytic} it now follows that the zero set of $f$ restricted to $\Tilde{\Theta}_{\textnormal{stable}}(G)\times \mathbb{R}^d_{> 0}$ has Lebesgue measure zero and thus, that the zero set of $f$ on its entire domain has Lebesgue measure zero. Therefore and because $\Theta_{\textnormal{stable}}(G)$ is itself measurable (see Lemma \ref{lemma_lebesgue} in Section \ref{sec_proofs_sec_causal_effects} for this fact), $\det(\Gamma_{R,C})$ is non-zero Lebesgue almost everywhere in $\Theta_{\textnormal{stable}}(G)\times \mathbb{R}^d_{\geq 0}$.
\end{proof}	
	
\subsection{Lemma \ref{lemma:invertible_lambda_new}}
\label{sec_invertible_lambda_new}
\begin{lemma}
	\label{lemma:invertible_lambda_new}
Assume a stable SVAR process satisfying Assumptions \ref{assumption_no_instantaneous_self_edges}, \ref{assumption_acyclicity} and \ref{assumption_latents_have_no_observed_parents}. Furthermore, let $B_U$ and $F^{\textnormal{obs}}$ be as in Theorem \ref{main_theorem}. If there is a system of directed paths $\Upsilon$ from $B_U$ to $F^{\textnormal{obs}}$ with no intersections such that 
 \begin{itemize}
            \item all vertices in each directed path except the last one are latent, and
            \item for every other system of directed paths $\Pi:B_U\rightarrow F^{\textnormal{obs}}$ with no intersections for which all vertices in each directed path except that last one are latent (and such that $\Pi\neq \Upsilon$) it holds that $\Pi(A)\neq \Upsilon(A)$.
        \end{itemize}
Then, in generic settings, $\Tilde{\Lambda}$ from equation \eqref{eq_matrix_not} is invertible.
\end{lemma}
\begin{proof}
	For the sake of just this proof, define a new graph $\mathcal{M}$ (containing finitely many vertices): 
	This graph contains all the latent vertices in the time interval $[t_{\textnormal{inf}}(B_U\cup \textnormal{pa}^{\textnormal{lat}}(F^{\textnormal{obs}})),t_{\textnormal{sup}}(B_U\cup \textnormal{pa}^{\textnormal{lat}}(F^{\textnormal{obs}}))]$. 
	The graph $\mathcal{M}$ also contains all the edges between these latent vertices from the full time graph except when both endpoints of these edges lie in $B_U$. Furthermore, $\mathcal{M}$ contains the vertices from the set $F^{\textnormal{obs}}$ and just the edges from the latent parents of $F^{\textnormal{obs}}$ to the respective elements in $F^{\textnormal{obs}}$.
	
Now, recall that $\Tilde{\Lambda}_{kj}=\Tilde{\lambda^k}_{b_j}$
is the sum of all path monomials of all directed paths from $b_j$ to $f_k$ in the full time graph that only use latent vertices except for the vertex $f_k$. By construction of $\mathcal{M}$, it then follows that $\Tilde{\Lambda}_{kj}=\Tilde{\lambda^k}_{b_j}$
is the sum of all path monomials of directed paths from $b_j$ to $f_k$ in $\mathcal{M}$.
Thus, we can apply the Lindström-Gessel-Viennot lemma \citep{lindstrom1973vector,gessel1985binomial} (and in particular the formulation as stated in Lemma 1 in \citet{lindstrom1973vector}) on $\Tilde{\Lambda}$ and $\mathcal{M}$  which yields that
	\begin{align}
	\label{det_lambda_tilde}
	\det(\Tilde{\Lambda}) = \sum_{\substack{\Pi:B_U\rightarrow F^{\textnormal{obs}}\\ \Pi\textnormal{ has no intersections}\\(\textnormal{underlying graph is }\mathcal{M})}}\textnormal{sign}(\Pi) \Pi(A).
	\end{align}

	Note that equation \eqref{det_lambda_tilde} is a polynomial (and thus a real-analytic function) of the underlying parameters in $\mathcal{M}$ as there are only finitely many systems of directed paths from $B_U$ to $F^{\textnormal{obs}}$ in $\mathcal{M}.$
	
	Now, the existence of a system of directed paths $\Upsilon$ with respect to the original full time graph satisfying the constraints mentioned in Lemma \ref{lemma:invertible_lambda_new} implies that $\Upsilon$ is a system of directed paths with no intersections in $\mathcal{M}$ by construction of $\mathcal{M}$. Now, by assumption, $\Upsilon$ has a unique monomial $\Upsilon(A)$ among all other systems of directed paths satisfying the constraints in Lemma \ref{lemma:invertible_lambda_new} when the underlying graph is the full time graph. Thus, by construction of $\mathcal{M}$, the system of directed paths $\Upsilon$ also has a unique monomial $\Upsilon(A)$ among all other systems of directed paths from $B_U$ to $F^{\textnormal{obs}}$ with no intersections when the underlying graph is $\mathcal{M}$. Hence, the monomial $\Psi(A)$ in equation \eqref{det_lambda_tilde} has coefficient $\textnormal{sign}(\Upsilon)$. Therefore and because $\det(\Tilde{\Lambda})$ is a polynomial of the underlying parameters in $\mathcal{M}$ and thus a polynomial in terms of $\Theta_{\textnormal{stable}}(G)\times \mathbb{R}_{\geq 0}^d$ and hence in terms of $\Tilde{\Theta}_{\textnormal{stable}}(G)\times \mathbb{R}_{\geq 0}^d$ (where 	$\Tilde{\Theta}_{\textnormal{stable}}(G)$ is a superset of 	$\Theta_{\textnormal{stable}}(G)$ representing all possible stable parameters for \emph{subgraphs} of $G$, see Section \ref{sec_proofs_sec_causal_effects} for a formal definition), it follows that $\det(\Tilde{\Lambda})$ is not the zero polynomial.  From the fact that $\Tilde{\Theta}_{\textnormal{stable}}(G)\times \mathbb{R}_{>0}^d$ is a connected open set due to Lemma \ref{lemma_lebesgue} in the Appendix, because $\mathbb{R}_{>0}^d$ is a connected open set and because the Cartesian product of connected open sets is connected and open, it thus follows that $\det(\Tilde{\Lambda})$ is a real-analytic function on the connected open domain $\Tilde{\Theta}_{\textnormal{stable}}(G)\times \mathbb{R}_{>0}^d$.
	By Lemma \ref{lemma_zero_sets_real_analytic}, it now follows that the zero set of $\det(\Tilde{\Lambda})$ restricted to $\Tilde{\Theta}_{\textnormal{stable}}(G)\times \mathbb{R}_{>0}^d$ has Lebesgue measure zero, and thus, that the zero set of $\det(\Tilde{\Lambda})$ on $\Tilde{\Theta}_{\textnormal{stable}}(G)\times \mathbb{R}_{\geq0}^d$ has Lebesgue measure zero. Therefore and because $\Theta_{\textnormal{stable}}(G)$ is itself measurable (see Lemma \ref{lemma_lebesgue} in Section \ref{sec_proofs_sec_causal_effects} for this fact),  $\det(\Tilde{\Lambda})$ is non-zero Lebesgue almost everywhere in $\Theta_{\textnormal{stable}}(G)\times \mathbb{R}^d_{\geq 0}$.
\end{proof}

\section{Further results for genericity}
\label{sec_proofs_sec_causal_effects}

For proving genericity, we typically use the fact that
a finite union of zero sets of real analytic functions that are not the zero-function have Lebesgue measure zero. For polynomials, the result that zero sets of non-zero polynomials have zero Lebesgue measure also exists in the statistics literature \citep{okamoto1973distinctness}, for the more general result for real analytic functions see for example \cite{mityagin2015zero}; Lemma 5.22 and Remark 5.23 in \cite{kuchment2016overview} or Lemma 1.22 in \cite{dang2015complex}. 
In the following lemma, we state this result.

\begin{lemma}
\label{lemma_zero_sets_real_analytic}
	Let $f$ be a real analytic function  that is not the zero-function on a connected open domain $D\subseteq \mathbb{R}^N$. Then, its zero set $\{x\in D:\;f(x)=0\}$ has Lebesgue measure zero.
\end{lemma}

We conclude the section by stating another important result that we use throughout the paper.

For that, we make use of a superset of $\Tilde{\Theta}_{\textnormal{stable}}(G)\supseteq \Theta_{\textnormal{stable}}(G)$ which is defined by
	\begin{align*}
	\Tilde{\Theta}_{\textnormal{stable}}(G)&:=\{(\theta_1,\ldots,\theta_N)\in \mathbb{R}^{N}:\;\theta_i= A^{(h_i)}_{j_ik_i} \textnormal{ for all $i\in[N]_1$} \\
	& \hspace{1cm}\textnormal{ and } A^{(0)},\ldots,A^{(p)} \textnormal{ satisfy the stability condition}\}.
	\end{align*}
	Note that the only difference of 	$\Tilde{\Theta}_{\textnormal{stable}}$ in comparison to 	$\Theta_{\textnormal{stable}}$ is that entries are also allowed to be zero and thus $\Tilde{\Theta}_{\textnormal{stable}}$ encodes all nonzero entries of stable parameter matrices for \emph{subgraphs} of $G$.
Since by convention every graph is its own subgraph, the inclusion $\Tilde{\Theta}_{\textnormal{stable}}(G)\supseteq \Theta_{\textnormal{stable}}(G)$ indeed holds true.

In Lemma \ref{lemma_lebesgue} below, we show---using $\Tilde{\Theta}_{\textnormal{stable}}$---that $\Theta_{\textnormal{stable}}$ identified as a subset of $\mathbb{R}^N$ is a connected open set with strictly positive Lebesgue measure, which is a result that we use at several points in the paper. Also note that this result implies that $\Theta_{\textnormal{stable}}\times \mathbb{R}^d_{\geq 0}$ has positive Lebesgue measure and that, for example, $\Theta_{\textnormal{stable}}\times \mathbb{R}^d_{>0}$ is open and connected as every finite union of open connected sets is open and connected.
\begin{lemma}
	\label{lemma_lebesgue}
	If Assumptions \ref{assumption_no_instantaneous_self_edges} and \ref{assumption_acyclicity} hold, then $\Theta_{\textnormal{stable}}(G)$ is nonempty, open and has strictly positive Lebesgue (with respect to $\mathbb{R}^{N}$) measure.
	
	Furthermore, if Assumptions \ref{assumption_no_instantaneous_self_edges} and \ref{assumption_acyclicity} hold, then $\Tilde{\Theta}_{\textnormal{stable}}(G)$ is nonempty, open, connected and has strictly positive Lebesgue (with respect to $\mathbb{R}^{N}$) measure.
	
\end{lemma}
\begin{proof}
	The stability condition for SVAR processes can be rephrased. Namely, an SVAR process is stable if and only if the block-matrix $\mathbf{B}$ from equation \eqref{big_A} in Section \ref{sec_trek_proofs} of the Appendix
	has only eigenvalues with modulus strictly less than $1$ (see Chapter 2 in \citet{lutkepohl2005new}).

	Now, note that the complement of $\Tilde{\Theta}_{\textnormal{stable}}(G)$ (in $\mathbb{R}^N$), which we denote by $\Tilde{\Theta}_{\textnormal{stable}}(G)^c$, equals
	\begin{align*}
	\Tilde{\Theta}_{\textnormal{stable}}(G)^c&:=\{(\theta_1,\ldots,\theta_N)\in \mathbb{R}^{N}:\;\theta_i= A^{(h_i)}_{j_ik_i} \textnormal{ for all $i\in[N]_1$} \\
	& \hspace{1cm}\textnormal{ and } A^{(0)},\ldots,A^{(p)} \textnormal{ \emph{do not} satisfy the stability condition}\}.
	\end{align*}
	
Also note that the complement of $\Theta_{\textnormal{stable}}(G)$ (in $\mathbb{R}^N$) which we denote by $\Theta_{\textnormal{stable}}(G)^c$ equals 	\begin{align*}
		\Theta_{\textnormal{stable}}(G)^c&=\Tilde{\Theta}_{\textnormal{stable}}(G)^c \cup \underbrace{\{(\theta_1,\ldots,\theta_N)\in \mathbb{R}^{N}:\;\theta_i=0 \text{ for at least one $i\in [N]_1$}\}}_{=: M_1}.
	\end{align*}

	We next show that $\Tilde{\Theta}_{\textnormal{stable}}(G)^c$ is closed, which then implies that $\Tilde{\Theta}_{\textnormal{stable}}(G)$ is open. The fact that $\Tilde{\Theta}_{\textnormal{stable}}(G)^c$ is closed also implies that $	\Theta_{\textnormal{stable}}(G)^c$ is closed because $M_1$ is closed and a finite union of closed sets is closed, and thus, that $\Theta_{\textnormal{stable}}(G)$ is open.
	
	Let $b_n\in \Tilde{\Theta}_{\textnormal{stable}}(G)^c$ for all $n\in\mathbb{N}_{\geq 0}$ and such that $b_n$ converges (in a componentwise manner) to some $b\in \mathbb{R}^N$. Let $A^{(0)}_n,\ldots, A^{(p)}_n$ be the corresponding sequence of coefficient matrices (note that this sequence of coefficient matrices is uniquely defined as all other entries of these coefficient matrices not encoded in $b_n$ are $0$ by the subgraph condition) and let $\mathbf{B_n}$ be the uniquely defined sequence of corresponding block matrices as defined in equation \eqref{big_A} (note that there might be several elements of $\Tilde{\Theta}_{\textnormal{stable}}(G)^c$ yielding the same block matrices, however, this non-uniqueness issue does not occur in the argument; we only require uniqueness of the block matrices $\mathbf{B}_n$ for given $A^{(0)}_n,\ldots, A^{(p)}_n$ and this fact is satisfied because $I_d-A^{(0)}_n$ is invertible by assumption). With this construction, componentwise convergence of $b_n$ to some $b$ implies that $A^{(0)}_n,\ldots, A^{(p)}_n$ converge in a componentwise manner to some $A^{(0)},\ldots, A^{(p)}$ and $\mathbf{B_n}$ also converges in a componentwise manner to some $\mathbf{B}$ (because the $I_d-A^{(0)}_n$ are invertible by assumption and because matrix inversion and multiplication is continuous, it follows that $(I_d-A^{(0)}_n)^{-1}A^{(i)}_n$ converges to $(I_d-A^{(0)})^{-1}A^{(i)}$ where $A^{(0)}$ is also permutable to a lower triangular matrix, because Assumptions \ref{assumption_no_instantaneous_self_edges} and \ref{assumption_acyclicity} also hold for the limit as no new non-zero entries are introduced, thus justifying invertibility of $I_d-A^{(0)}$). Furthermore, it holds that the block matrix as defined in equation \eqref{big_A} induced by $A^{(0)},\ldots, A^{(p)}$ yields $\mathbf{B}$ (because $I_d-A^{(0)}$ which occurs in the construction in equation \eqref{big_A} is invertible as just argued) and that the $(h_1,j_1,k_1),\ldots,(h_N,j_N,k_N)$-entries of $A^{(0)},\ldots, A^{(p)}$ equal $b$ and that no new non-zero entries in $A^{(0)},\ldots, A^{(p)}$ are introduced when taking the limit because taking the limit of constant $0$-entries still yields $0$.
	
We now argue that $\mathbf{B}$ has at least one eigenvalue with modulus bigger or equal than $1$: The componentwise convergence of $\mathbf{B}_n$ to $\mathbf{B}$ implies that the eigenvalues of $\mathbf{B}_n$ converge to the eigenvalues of $\mathbf{B}$, which more specifically means that after some permutation of the eigenvalues of each $\mathbf{B}_n$, the vectors of eigenvalues of the $\mathbf{B}_n$ converge to the (in some fashion ordered) vector of the eigenvalues of $\mathbf{B}$. Therefore, one eigenvalue of $\mathbf{B}$ must have modulus bigger or equal than $1$ because otherwise, this convergence of eigenvalues cannot hold because each $\mathbf{B}_n$ has at least one eigenvalue with modulus bigger or equal than $1$. Thus, the matrix $\mathbf{B}$ also gives rise to an SVAR process that as argued in the previous paragraph induces a graph that is a subgraph of $G$, yet an SVAR process that is not stable. Hence, $b\in \Tilde{\Theta}_{\textnormal{stable}}(G)^c$ and hence, $\Tilde{\Theta}_{\textnormal{stable}}(G)^c$ is closed and thus, $\Tilde{\Theta}_{\textnormal{stable}}(G)$ is open. It hence also follows that $	\Theta_{\textnormal{stable}}(G)^c$ is closed and thus that $	\Theta_{\textnormal{stable}}(G)$ is open. Note that this openness also implies that $\Tilde{\Theta}_{\textnormal{stable}}(G)$ and $	\Theta_{\textnormal{stable}}(G)$ are measurable in the first place as open subsets of $\mathbb{R}^N$ are Lebesgue measurable.

Because non-empty open subsets of $\mathbb{R}^N$ have positive Lebesgue measure (and $\Tilde{\Theta}_{\textnormal{stable}}(G)$ clearly is non-empty because $(0,0,\ldots,0)\in \mathbb{R}^N$ is an element of $\Tilde{\Theta}_{\textnormal{stable}}(G)$), it follows that $\Tilde{\Theta}_{\textnormal{stable}}(G)$ has strictly positive Lebesgue measure. 
Because
\begin{align*}
    \Tilde{\Theta}_{\textnormal{stable}}(G) \subseteq 	\Theta_{\textnormal{stable}}(G) \dot{\cup} M_1
\end{align*}
and because $M_1$ has zero Lebesgue measure (and is measurable in the first place) as a proper hyperplane in its ambient space $\mathbb{R}^N$, it follows that $\Theta_{\textnormal{stable}}(G)$ has positive Lebesgue measure as well and thus that $	\Theta_{\textnormal{stable}}(G)$ is non-empty.

Finally, to show that $\Tilde{\Theta}_{\textnormal{stable}}(G)$ is connected, we show that $\Tilde{\Theta}_{\textnormal{stable}}(G)$ is path-connected which then implies that $\Tilde{\Theta}_{\textnormal{stable}}(G)$ is connected. For that, let $b, b'\in \Tilde{\Theta}_{\textnormal{stable}}(G)$ and $A^{(0)},\ldots, A^{(p)}$ respectively $A'^{(0)},\ldots, A'^{(p)}$ be the uniquely defined corresponding coefficient matrices. We construct a path from $b$ to $b'$ by constructing several subpaths individually and then gluing these subpaths together.
Without loss of generality, assume that the first $m_0$-many entries of $b$ and $b'$ correspond to parameters of $A^{(0)}$ respectively $A'^{(0)}$, the second $m_1$-many entries of $b$ and $b'$ to parameters of $A^{(1)}$ respectively $A'^{(1)}$, and so on. We first construct a subpath from $b=(b_1,\ldots, b_N)$ to $(b_1,b_2,\ldots, b_{m_0},0,0\ldots, 0)$ (and analogously from $b'=(b'_1,\ldots, b'_N)$ to $(b'_1,b'_2,\ldots, b'_{m_0},0,0\ldots, 0)$). For that we look at the path
\begin{align*}
    &f_1:[0,1]\rightarrow \Tilde{\Theta}_{\textnormal{stable}}(G),\\
    &\gamma \mapsto (b_1,b_2,\ldots, b_{m_0}, \gamma b_{m_0+1} , \gamma b_{m_0+2}, \ldots, \gamma b_{m_0+m_1}, \gamma^2 b_{m_0+m_1+1}, \gamma^2 b_{m_0+m_1+2},\ldots),
\end{align*}
that is, $f_1$ scales the entries in $b$ corresponding to $A^{(i)}$ by $\gamma^i$. Clearly, $f_1$ is continuous because its individual component functions are monomials. We next show that for all $\gamma \in [0, 1]$ it holds that $f_1(\gamma)\in \Tilde{\Theta}_{\textnormal{stable}}(G)$. Write $A^{(i)}_{\gamma, f_1}$ with $i\in [p]_0$ for the coefficient matrices corresponding to $f_1(\gamma)$. To see that $f_1(\gamma)\in \Tilde{\Theta}_{\textnormal{stable}}(G)$ for all $\gamma\in[0,1]$, we look at the stability condition for $f_1(\gamma)$ which reads as \citep[Chapter 2]{lutkepohl2005new}
\begin{align*}
\label{eq_reverse_char}
    \textnormal{det}(I_d-B^{(1)}_{\gamma, f_1} \lambda-\cdots -B^{(p)}_{\gamma, f_1}\lambda^p)\neq 0 \numberthis
\end{align*}
for all $|\lambda|\leq 1$ where $B^{(i)}_{\gamma, f_1}=(I_d-A^{(0)}_{\gamma, f_1})^{-1}A^{(i)}_{\gamma, f_1}$ for $i\in [p]_1$. Here, note that $I_d-A^{(0)}_{\gamma, f_1}$ is invertible because $A^{(0)}_{\gamma, f_1} = A^{(0)}$ and $I_d-A^{(0)}$ is invertible by Assumptions \ref{assumption_no_instantaneous_self_edges} and \ref{assumption_acyclicity}. From $A^{(0)}_{\gamma, f_1} = A^{(0)}$ it also follows that $B^{(i)}_{\gamma, f_1} = B^{(i)}\cdot \gamma^i$ for all $i\in[p]_1$ where $B^{(i)}=(I_d-A^{(0)})^{(-1)}A^{(i)}$. Therefore, 
\begin{align*}
\textnormal{det}(I_d-B^{(1)}_{\gamma, f_1} \lambda-\cdots -B^{(p)}_{\gamma, f_1})&=    \textnormal{det}(I_d-B^{(1)}\gamma  \lambda-\cdots -B^{(p)}\gamma^p \lambda^p)\\
    &=\textnormal{det}(I_d-B^{(1)}(\gamma  \lambda)-\cdots -B^{(p)}(\gamma \lambda)^p)=:\textnormal{RHS}.
\end{align*}
From the stability condition on $b$ it now follows that
$\textnormal{RHS}\neq 0$
for all $|\gamma \lambda|\leq 1$, and thus, because $\gamma \in [0,1]$ and because $\gamma |\lambda| = |\gamma \lambda|$, it follows that $\textnormal{RHS}\neq 0$ for all $|\lambda|\leq 1$ and $\gamma\in[0,1]$. Thus, $f_1(\gamma) \in \Tilde{\Theta}_{\textnormal{stable}}(G)$, and analogously constructing a path $f_2$ for $b'$ yields an analogous result for $b'$.

As a next subpath $f_3$, we connect $(b_1,b_2,\ldots, b_{m_0},0,0,\ldots, 0)$ to $(0,0,\ldots, 0)$ (and analogously we connect $(b'_1,b'_2,\ldots, b'_{m_0},0,0,\ldots, 0)$ to $(0,0,\ldots, 0)$). Define $f_3$ by
\begin{align*}
    f_3:[0,1]\rightarrow \Tilde{\Theta}_{\textnormal{stable}}(G), \gamma \mapsto (\gamma b_1,\gamma b_2,\ldots, \gamma b_{m_0},0,0,\ldots, 0).
\end{align*}
Clearly, $f_3$ is continuous as its individual component functions are linear functions. Also note that $f_3(\gamma) \in \Tilde{\Theta}_{\textnormal{stable}}(G)$ for all $\gamma\in[0,1]$: First, because for $(b_1,b_2,\ldots, b_{m_0},0,0,\ldots, 0)$ Assumptions \ref{assumption_no_instantaneous_self_edges} and \ref{assumption_acyclicity} hold and thus, Assumptions \ref{assumption_no_instantaneous_self_edges} and \ref{assumption_acyclicity} also hold for $(\gamma b_1,\gamma b_2,\ldots, \gamma b_{m_0},0,0,\ldots, 0)$ thus making $I_d-A^{(0)}_{\gamma, f_3}$, where $A^{(0)}_{\gamma, f_3}$ is the coefficient matrix corresponding to $(\gamma b_1,\gamma b_2,\ldots, \gamma b_{m_0},0,0,\ldots, 0)$, invertible. Second, because the stability condition always holds if $I_d-A^{(0)}_{\gamma, f_3}$ is invertible and if $p=0$ because the stability condition reads as
\begin{align*}
    \textnormal{det}(I_d)\neq 0
\end{align*}
for all $|\lambda|\leq 1$ which is always satisfied.
Analogously constructing a path $f_4$ for $(b'_1,\allowbreak b'_2,\allowbreak\ldots, b'_{m_0},\allowbreak 0,\allowbreak 0,\allowbreak\ldots,\allowbreak 0)$ and gluing $f_1$ to $f_3$ to $f_4$ to $f_2$ together into a path $f$ and using that $f$ is clearly continuous because the individual subpaths are continuous and because the subpaths are glued together at shared points,  shows that (arbitrary) $b, b'\in \Tilde{\Theta}_{\textnormal{stable}}(G)$ are path-connected and thus that $\Tilde{\Theta}_{\textnormal{stable}}(G)$ is path-connected. Hence, we conclude that $\Tilde{\Theta}_{\textnormal{stable}}(G)$ is connected.
\end{proof}

\section{Further details on the real-world example}
\label{sec_app_real_world_ex}

	For total wind generation, electricity demand and prices, we use data from \citet{Bundesnetzagentur}.  For these three time series, we additionally regress out the control variables hour of the week, month of the year and holiday index as explained in Section D of \citet{tiedemann2024identifying}. Hereby, the source for the holiday index is \citet{githubrepo} which itself is based on \citet{destatis} and the python-package \emph{holidays} from \citet{holidays}. We thus in fact do not look at the original total wind generation, electricity demand and price time series, but rather at the corresponding residual time series after regressing out the just mentioned control variables. Note that due to simplicity and accessibility-related issues, we only use a subset of the control variables that \citet{tiedemann2024identifying} use. 
	
	For the semi-synthetic data, we (largely) follow the setup layed out in Section 4 of \citet{tiedemann2024identifying}. In particular, we first fit an AR model of order 20 for the total (residual) wind generation and then use the following settings for the experiments (neglecting units): $\beta^{P}=-100$, $\beta^{P1}=50$, $\gamma^P = 500$, $\gamma^W=1$, $S_0=25000$, $D_0=50000$, $U^D_t\sim \mathcal{N}(0, 2000)$, $U^S_t \sim \mathcal{N}(0, 1)$, $U^A_t\sim \mathcal{N}(0, 2000/\sqrt{2})$, $U^B_t\sim \mathcal{N}(0, 2000/\sqrt{2})$, $B_0=0$, $\beta^{D1}=0.7$ and $\beta^{B1}= 0.9$.

\section{More Examples}
\label{sec_examples}
In this section, we present more examples.

\begin{figure}[h]
	\centering
	\scalebox{0.8}{
		\begin{tikzpicture}
		\node[] (t-5) at (-10,3) {\small$t-5$};
		\node[] (t-4) at (-9,3) {\small$t-4$};
		\node[] (t-3) at (-8,3) {\small$t-3$};
		\node[] (t-2) at (-7,3) {\small$t-2$};
		\node[] (t-1) at (-6,3) {\small$t-1$};
		\node[] (t) at (-5,3) {\small$t$};
		\node[] (t+1) at (-4,3) {\small$t+1$};
		\node[] (t+2) at (-3,3) {\small$t+2$};
		\node[] (t+3) at (-2,3) {\small$t+3$};
		\node[] (t+4) at (-1,3) {\small$t+4$};
		\node[] (t+5) at (0,3) {\small$t+5$};
		
		\node[] (X) at (-12,2) {\small$X$};
		\node[] (U) at (-12,0) {\small$U$};
		\node[] (Y) at (-12,-2) {\small$Y$};
		
		\node[shape=circle, fill] (U-10) at (-10,0) {};
		\node[shape=circle, fill] (X-10) at (-10,2) {};
		\node[shape=circle, fill] (Y-10) at (-10,-2) {};
		
		\node[shape=circle, fill] (U-9) at (-9,0) {};
		\node[shape=circle, fill] (X-9) at (-9,2) {};
		\node[shape=circle, fill] (Y-9) at (-9,-2) {};
		
		\node[shape=circle, fill] (U-8) at (-8,0) {};
		\node[shape=circle, fill] (X-8) at (-8,2) {};
		\node[shape=circle, fill] (Y-8) at (-8,-2) {};
		
		\node[shape=circle, fill] (U-7) at (-7,0) {};
		\node[shape=circle, fill] (X-7) at (-7,2) {};
		\node[shape=circle, fill] (Y-7) at (-7,-2) {};
		
		\node[shape=circle, fill] (U-6) at (-6,0) {};
		\node[shape=circle, fill] (X-6) at (-6,2) {};
		\node[shape=circle, fill] (Y-6) at (-6,-2) {};
		
		\node[shape=circle, fill] (U-5) at (-5,0) {};
		\node[shape=circle, fill] (X-5) at (-5,2) {};
		\node[shape=circle, fill] (Y-5) at (-5,-2) {};
		
		\node[shape=circle, fill] (U-4) at (-4,0) {};
		\node[shape=circle, fill] (X-4) at (-4,2) {};
		\node[shape=circle, fill] (Y-4) at (-4,-2) {};
		
		\node[shape=circle, fill] (U-3) at (-3,0) {};
		\node[shape=circle, fill] (X-3) at (-3,2) {};
		\node[shape=circle, fill] (Y-3) at (-3,-2) {};
		
		\node[shape=circle, fill] (U-2) at (-2,0) {};
		\node[shape=circle, fill] (X-2) at (-2,2) {};
		\node[shape=circle, fill] (Y-2) at (-2,-2) {};
		
		\node[shape=circle, fill] (U-1) at (-1,0) {};
		\node[shape=circle, fill] (X-1) at (-1,2) {};
		\node[shape=circle, fill] (Y-1) at (-1,-2) {};
		
		\node[shape=circle, fill] (U0) at (0,0) {};
		\node[shape=circle, fill] (X0) at (0,2) {};
		\node[shape=circle, fill] (Y0) at (0,-2) {};
		
		\node[] (Udots-past) at (-10.5,0) {$\ldots$};
		\node[] (Udots-future) at (0.5,0) {$\ldots$};
		\node[] (Xdots-past) at (-10.5,2) {$\ldots$};
		\node[] (Xdots-future) at (0.5,2) {$\ldots$};
		\node[] (Ydots-past) at (-10.5,-2) {$\ldots$};
		\node[] (Ydots-future) at (0.5,-2) {$\ldots$};

		\path [->, line width = 0.5mm] (U-10) edge node[left] {} (U-9);
		\path [->, line width = 0.5mm] (U-9) edge node[left] {} (U-8);
		\path [->, line width = 0.5mm] (U-8) edge node[left] {} (U-7);
		\path [->, line width = 0.5mm] (U-7) edge node[left] {} (U-6);
		\path [->, line width = 0.5mm] (U-6) edge node[left] {} (U-5);
		\path [->, line width = 0.5mm] (U-5) edge node[left] {} (U-4);
		\path [->, line width = 0.5mm] (U-4) edge node[left] {} (U-3);
		\path [->, line width = 0.5mm] (U-3) edge node[left] {} (U-2);
		\path [->, line width = 0.5mm] (U-2) edge node[left] {} (U-1);
		\path [->, line width = 0.5mm] (U-1) edge node[left] {} (U0);	
		
		\path [->, line width = 0.5mm, bend left] (U-10) edge node[left] {} (U-8);
		\path [->, line width = 0.5mm, bend left] (U-9) edge node[left] {} (U-7);
		\path [->, line width = 0.5mm, bend left] (U-8) edge node[left] {} (U-6);
		\path [->, line width = 0.5mm, bend left] (U-7) edge node[left] {} (U-5);
		\path [->, line width = 0.5mm, bend left] (U-6) edge node[left] {} (U-4);
		\path [->, line width = 0.5mm, bend left] (U-5) edge node[left] {} (U-3);
		\path [->, line width = 0.5mm, bend left] (U-4) edge node[left] {} (U-2);
		\path [->, line width = 0.5mm, bend left] (U-3) edge node[left] {} (U-1);
	    \path [->, line width = 0.5mm, bend left] (U-2) edge node[left] {} (U0);
		
		\path [->, line width = 0.5mm] (X-10) edge node[left] {} (X-9);
		\path [->, line width = 0.5mm,] (X-9) edge node[left] {} (X-8);
		\path [->, line width = 0.5mm] (X-8) edge node[left] {} (X-7);
		\path [->, line width = 0.5mm] (X-7) edge node[left] {} (X-6);
		\path [->, line width = 0.5mm] (X-6) edge node[left] {} (X-5);
		\path [->, line width = 0.5mm] (X-5) edge node[left] {} (X-4);
		\path [->, line width = 0.5mm] (X-4) edge node[left] {} (X-3);
		\path [->, line width = 0.5mm] (X-3) edge node[left] {} (X-2);
		\path [->, line width = 0.5mm] (X-2) edge node[left] {} (X-1);
        \path [->, line width = 0.5mm] (X-1) edge node[left] {} (X0);
		
		\path [->, line width = 0.5mm, color = blue] (Y-10) edge node[left] {} (Y-9);
		\path [->, line width = 0.5mm, color = blue] (Y-9) edge node[left] {} (Y-8);
		\path [->, line width = 0.5mm, color = blue] (Y-8) edge node[left] {} (Y-7);
		\path [->, line width = 0.5mm, color = blue] (Y-7) edge node[left] {} (Y-6);
		\path [->, line width = 0.5mm, color = blue] (Y-6) edge node[left] {} (Y-5);
		\path [->, line width = 0.5mm, color = blue] (Y-5) edge node[left] {} (Y-4);
		\path [->, line width = 0.5mm, color = blue] (Y-4) edge node[left] {} (Y-3);
		\path [->, line width = 0.5mm, color = blue] (Y-3) edge node[left] {} (Y-2);
		\path [->, line width = 0.5mm, color = blue] (Y-2) edge node[left] {} (Y-1);
		\path [->, line width = 0.5mm, color = blue] (Y-1) edge node[left] {} (Y0);

		\path [->, line width = 0.5mm] (U-10) edge node[left] {} (Y-7);
		\path [->, line width = 0.5mm] (U-9) edge node[left] {} (Y-6);
		\path [->, line width = 0.5mm] (U-8) edge node[left] {} (Y-5);
		\path [->, line width = 0.5mm] (U-7) edge node[left] {} (Y-4);
		\path [->, line width = 0.5mm] (U-6) edge node[left] {} (Y-3);
		\path [->, line width = 0.5mm] (U-5) edge node[left] {} (Y-2);
		\path [->, line width = 0.5mm] (U-4) edge node[left] {} (Y-1);
		\path [->, line width = 0.5mm] (U-3) edge node[left] {} (Y0);
		
		\path [->, line width = 0.5mm] (U-10) edge node[left] {} (Y-8);
		\path [->, line width = 0.5mm] (U-9) edge node[left] {} (Y-7);
		\path [->, line width = 0.5mm] (U-8) edge node[left] {} (Y-6);
		\path [->, line width = 0.5mm] (U-7) edge node[left] {} (Y-5);
		\path [->, line width = 0.5mm] (U-6) edge node[left] {} (Y-4);
		\path [->, line width = 0.5mm] (U-5) edge node[left] {} (Y-3);
		\path [->, line width = 0.5mm] (U-4) edge node[left] {} (Y-2);
        \path [->, line width = 0.5mm] (U-3) edge node[left] {} (Y-1);
        \path [->, line width = 0.5mm] (U-2) edge node[left] {} (Y0);
		
		\path [->, line width = 0.5mm] (U-10) edge node[left] {} (X-9);
		\path [->, line width = 0.5mm] (U-9) edge node[left] {} (X-8);
		\path [->, line width = 0.5mm] (U-8) edge node[left] {} (X-7);
		\path [->, line width = 0.5mm] (U-7) edge node[left] {} (X-6);
		\path [->, line width = 0.5mm] (U-6) edge node[left] {} (X-5);
		\path [->, line width = 0.5mm] (U-5) edge node[left] {} (X-4);
		\path [->, line width = 0.5mm] (U-4) edge node[left] {} (X-3);
		\path [->, line width = 0.5mm] (U-3) edge node[left] {} (X-2);
        \path [->, line width = 0.5mm] (U-2) edge node[left] {} (X-1);
        \path [->, line width = 0.5mm] (U-1) edge node[left] {} (X0);
		
		\path [->, line width = 0.5mm] (U-10) edge node[left] {} (X-8);
		\path [->, line width = 0.5mm] (U-9) edge node[left] {} (X-7);
		\path [->, line width = 0.5mm] (U-8) edge node[left] {} (X-6);
		\path [->, line width = 0.5mm] (U-7) edge node[left] {} (X-5);
		\path [->, line width = 0.5mm] (U-6) edge node[left] {} (X-4);
		\path [->, line width = 0.5mm] (U-5) edge node[left] {} (X-3);
		\path [->, line width = 0.5mm] (U-4) edge node[left] {} (X-2);
		\path [->, line width = 0.5mm] (U-3) edge node[left] {} (X-1);
		\path [->, line width = 0.5mm] (U-2) edge node[left] {} (X0);
		
       \path [->, line width = 0.5mm, color = red] (X-10) edge node[left] {} (Y-5);
		\path [->, line width = 0.5mm, color = red] (X-9) edge node[left] {} (Y-4);
		\path [->, line width = 0.5mm, color = red] (X-8) edge node[left] {} (Y-3);
		\path [->, line width = 0.5mm, color = red] (X-7) edge node[left] {} (Y-2);
		\path [->, line width = 0.5mm, color = red] (X-6) edge node[left] {} (Y-1);
		\path [->, line width = 0.5mm, color = red] (X-5) edge node[left] {} (Y0);

		\end{tikzpicture}
	}
	\caption{Full time graph for Example \ref{ex_app_1}. Here, the \color{red} red \color{black} edges correspond to $A^{(5)}_{YO^2}$ and the \color{blue} blue \color{black} edges to $A^{(1)}_{YY}$.}
	\label{ex_full_time_graph_app_1}
\end{figure}
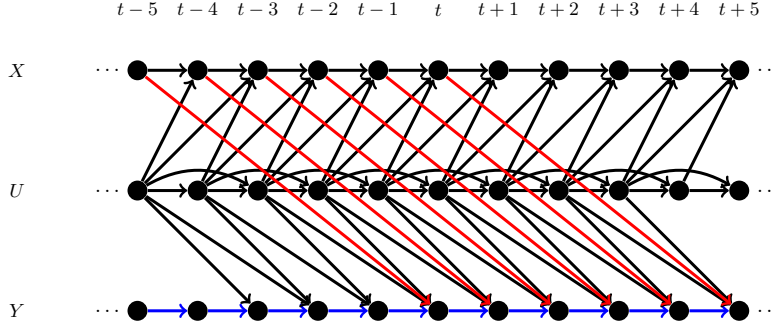
\addtocounter{theorem}{-24}
\begin{example}
\label{ex_app_1}
Consider the full time graph from Figure \ref{ex_full_time_graph_app_1}. Here,
$d_U=1$ and $d_O=2$. Abbreviate $\{X_t\}_{t\in\mathbb{Z}}:=\{O^2_t\}_{t\in\mathbb{Z}}$. Note that $m_U=2$, $l^U_1=1$, $l^U_2=2$, $m_Y=1$, $l^Y_1=1$, $m_{X}=1$, $l^{X}_1=1$, $m_{YX}=1$, $l^{YX}_1=5$, $m_{XY}=0$, $m_{YU}=2$, $l^{YU}_1=2$, $l^{YU}_2=3$, $m_{XU}=2$, $l^{XU}_1=1$, $l^{XU}_2=2$.

First, note that $\textnormal{pa}^{\textnormal{obs}}(Y_t)=\{Y_{t-1}, X_{t-5}\}$ and $\textnormal{pa}^{\textnormal{lat}}(Y_t)=\{U_{t-2}, U_{t-3}\}$. Let
\begin{align*}
    F^{\textnormal{obs}}:=\{X_{t+2}, X_{t+3}\}.
\end{align*}
Then, $\textnormal{pa}^{\textnormal{obs}}(F^{\textnormal{obs}})\setminus F^{\textnormal{obs}}=\{X_{t+1}\}$ and $\textnormal{pa}^{\textnormal{lat}}(F^{\textnormal{obs}})=\{U_{t},U_{t+1},U_{t+2}\}$.
Therefore,
\begin{align*}
    C=\{X_{t-5}, Y_{t-1}, X_{t+2}, X_{t+3}, X_{t+1}\}
\end{align*}
is valid.
Next, consider
\begin{align*}
    B_U:=\{U_{t-3}, U_{t-2}\}.
\end{align*}
Then, $\tau_{Y}=\{2\}$ and $\tau_{X}=\{3\}$ are valid.

Also note that $C^{(1)}_{Y}=\{Y_{t-1}\}$ and $C^{(1)}_{X}=\{X_{t-5}, X_{t+1}\}$. Thus, Corollary \ref{corollary_identifiability} applies and we get identifiability.
One possible $R$ according to Lemma \ref{lemma_resid_class1} is
\begin{align*}
    R:=\{X_{t-5}, Y_{t-2}, X_{t-3}, X_{t-6}, X_{t-7}\}.
\end{align*}
In generic settings, the parameters of interest $A^{(5)}_{YX}$ and $A^{(1)}_{YY}$ are given by
    \begin{align*}
    \label{est_app_1}
        &\begin{pmatrix}
        \color{red}A^{(5)}_{YX}\color{black}\\
        \color{blue}A^{(1)}_{YY}\color{black}\\
        \vdots \\
        \textnormal{other terms}\\
        \vdots
        \end{pmatrix}\\
        &=\begin{pmatrix}
        \Gamma_{X_{t-5}\color{red}X_{t-5}\color{black}} & \Gamma_{X_{t-5}\color{blue}Y_{t-1}\color{black}} & \Gamma_{X_{t-5}X_{t+2}} & \Gamma_{X_{t-5}X_{t+3}} & \Gamma_{X_{t-5}X_{t+1}}\\
        \Gamma_{Y_{t-2}\color{red}X_{t-5}\color{black}} & \Gamma_{Y_{t-2}\color{blue}Y_{t-1}\color{black}} & \Gamma_{Y_{t-2}X_{t+2}} & \Gamma_{Y_{t-2}X_{t+3}} & \Gamma_{Y_{t-2}X_{t+1}}\\
        \Gamma_{X_{t-3}\color{red}X_{t-5}\color{black}} & \Gamma_{X_{t-3}\color{blue}Y_{t-1}\color{black}} & \Gamma_{X_{t-3}X_{t+2}} & \Gamma_{X_{t-3}X_{t+3}} & \Gamma_{X_{t-3}X_{t+1}}\\
        \Gamma_{X_{t-6}\color{red}X_{t-5}\color{black}} & \Gamma_{X_{t-6}\color{blue}Y_{t-1}\color{black}} & \Gamma_{X_{t-6}X_{t+2}} & \Gamma_{X_{t-6}X_{t+3}} & \Gamma_{X_{t-6}X_{t+1}}\\
        \Gamma_{X_{t-7}\color{red}X_{t-5}\color{black}} & \Gamma_{X_{t-7}\color{blue}Y_{t-1}\color{black}} & \Gamma_{X_{t-7}X_{t+2}} & \Gamma_{X_{t-7}X_{t+3}} & \Gamma_{X_{t-7}X_{t+1}}
        \end{pmatrix}^{-1}\cdot
        \begin{pmatrix}
        \Gamma_{X_{t-5}Y_{t}}\\
        \Gamma_{Y_{t-2}Y_{t}}\\
        \Gamma_{X_{t-3}Y_{t}}\\
        \Gamma_{X_{t-6}Y_{t}}\\
        \Gamma_{X_{t-7}Y_{t}}\\
        \end{pmatrix}.\numberthis
    \end{align*}
    Here, the red and blue colours indicate which component corresponds to which column.
    
    Besides, in Figure \ref{Figure_num_validation_ex_app}a we present a numerical validation for this example.
    \demo

\end{example}

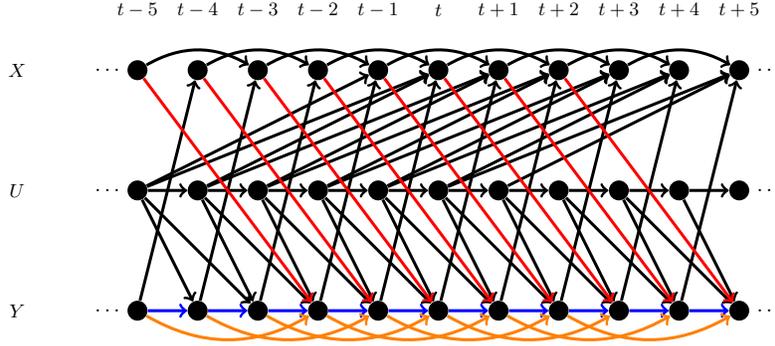
\begin{figure}[h]
	\centering
	\scalebox{0.8}{
		\begin{tikzpicture}
		\node[] (t-5) at (-10,3) {\small$t-5$};
		\node[] (t-4) at (-9,3) {\small$t-4$};
		\node[] (t-3) at (-8,3) {\small$t-3$};
		\node[] (t-2) at (-7,3) {\small$t-2$};
		\node[] (t-1) at (-6,3) {\small$t-1$};
		\node[] (t) at (-5,3) {\small$t$};
		\node[] (t+1) at (-4,3) {\small$t+1$};
		\node[] (t+2) at (-3,3) {\small$t+2$};
		\node[] (t+3) at (-2,3) {\small$t+3$};
		\node[] (t+4) at (-1,3) {\small$t+4$};
		\node[] (t+5) at (0,3) {\small$t+5$};
		
		\node[] (X) at (-12,2) {\small$X$};
		\node[] (U) at (-12,0) {\small$U$};
		\node[] (Y) at (-12,-2) {\small$Y$};
		
		\node[shape=circle, fill] (U-10) at (-10,0) {};
		\node[shape=circle, fill] (X-10) at (-10,2) {};
		\node[shape=circle, fill] (Y-10) at (-10,-2) {};
		
		\node[shape=circle, fill] (U-9) at (-9,0) {};
		\node[shape=circle, fill] (X-9) at (-9,2) {};
		\node[shape=circle, fill] (Y-9) at (-9,-2) {};
		
		\node[shape=circle, fill] (U-8) at (-8,0) {};
		\node[shape=circle, fill] (X-8) at (-8,2) {};
		\node[shape=circle, fill] (Y-8) at (-8,-2) {};
		
		\node[shape=circle, fill] (U-7) at (-7,0) {};
		\node[shape=circle, fill] (X-7) at (-7,2) {};
		\node[shape=circle, fill] (Y-7) at (-7,-2) {};
		
		\node[shape=circle, fill] (U-6) at (-6,0) {};
		\node[shape=circle, fill] (X-6) at (-6,2) {};
		\node[shape=circle, fill] (Y-6) at (-6,-2) {};
		
		\node[shape=circle, fill] (U-5) at (-5,0) {};
		\node[shape=circle, fill] (X-5) at (-5,2) {};
		\node[shape=circle, fill] (Y-5) at (-5,-2) {};
		
		\node[shape=circle, fill] (U-4) at (-4,0) {};
		\node[shape=circle, fill] (X-4) at (-4,2) {};
		\node[shape=circle, fill] (Y-4) at (-4,-2) {};
		
		\node[shape=circle, fill] (U-3) at (-3,0) {};
		\node[shape=circle, fill] (X-3) at (-3,2) {};
		\node[shape=circle, fill] (Y-3) at (-3,-2) {};
		
		\node[shape=circle, fill] (U-2) at (-2,0) {};
		\node[shape=circle, fill] (X-2) at (-2,2) {};
		\node[shape=circle, fill] (Y-2) at (-2,-2) {};
		
		\node[shape=circle, fill] (U-1) at (-1,0) {};
		\node[shape=circle, fill] (X-1) at (-1,2) {};
		\node[shape=circle, fill] (Y-1) at (-1,-2) {};
		
		\node[shape=circle, fill] (U0) at (0,0) {};
		\node[shape=circle, fill] (X0) at (0,2) {};
		\node[shape=circle, fill] (Y0) at (0,-2) {};
		
		\node[] (Udots-past) at (-10.5,0) {$\ldots$};
		\node[] (Udots-future) at (0.5,0) {$\ldots$};
		\node[] (Xdots-past) at (-10.5,2) {$\ldots$};
		\node[] (Xdots-future) at (0.5,2) {$\ldots$};
		\node[] (Ydots-past) at (-10.5,-2) {$\ldots$};
		\node[] (Ydots-future) at (0.5,-2) {$\ldots$};

		\path [->, line width = 0.5mm] (U-10) edge node[left] {} (U-9);
		\path [->, line width = 0.5mm] (U-9) edge node[left] {} (U-8);
		\path [->, line width = 0.5mm] (U-8) edge node[left] {} (U-7);
		\path [->, line width = 0.5mm] (U-7) edge node[left] {} (U-6);
		\path [->, line width = 0.5mm] (U-6) edge node[left] {} (U-5);
		\path [->, line width = 0.5mm] (U-5) edge node[left] {} (U-4);
		\path [->, line width = 0.5mm] (U-4) edge node[left] {} (U-3);
		\path [->, line width = 0.5mm] (U-3) edge node[left] {} (U-2);
		\path [->, line width = 0.5mm] (U-2) edge node[left] {} (U-1);
		\path [->, line width = 0.5mm] (U-1) edge node[left] {} (U0);

		\path [->, line width = 0.5mm, bend left] (X-10) edge node[left] {} (X-8);
		\path [->, line width = 0.5mm, bend left] (X-9) edge node[left] {} (X-7);
		\path [->, line width = 0.5mm, bend left] (X-8) edge node[left] {} (X-6);
		\path [->, line width = 0.5mm, bend left] (X-7) edge node[left] {} (X-5);
		\path [->, line width = 0.5mm, bend left] (X-6) edge node[left] {} (X-4);
		\path [->, line width = 0.5mm, bend left] (X-5) edge node[left] {} (X-3);
		\path [->, line width = 0.5mm, bend left] (X-4) edge node[left] {} (X-2);
		\path [->, line width = 0.5mm, bend left] (X-3) edge node[left] {} (X-1);
		\path [->, line width = 0.5mm, bend left] (X-2) edge node[left] {} (X0);

		\path [->, line width = 0.5mm, color = blue] (Y-10) edge node[left] {} (Y-9);
		\path [->, line width = 0.5mm, color = blue] (Y-9) edge node[left] {} (Y-8);
		\path [->, line width = 0.5mm, color = blue] (Y-8) edge node[left] {} (Y-7);
		\path [->, line width = 0.5mm, color = blue] (Y-7) edge node[left] {} (Y-6);
		\path [->, line width = 0.5mm, color = blue] (Y-6) edge node[left] {} (Y-5);
		\path [->, line width = 0.5mm, color = blue] (Y-5) edge node[left] {} (Y-4);
		\path [->, line width = 0.5mm, color = blue] (Y-4) edge node[left] {} (Y-3);
		\path [->, line width = 0.5mm, color = blue] (Y-3) edge node[left] {} (Y-2);
		\path [->, line width = 0.5mm, color = blue] (Y-2) edge node[left] {} (Y-1);
		\path [->, line width = 0.5mm, color = blue] (Y-1) edge node[left] {} (Y0);

		\path [->, line width = 0.5mm, color = black] (Y-10) edge node[left] {} (X-9);
		\path [->, line width = 0.5mm, color = black] (Y-9) edge node[left] {} (X-8);
		\path [->, line width = 0.5mm, color = black] (Y-8) edge node[left] {} (X-7);
		\path [->, line width = 0.5mm, color = black] (Y-7) edge node[left] {} (X-6);
		\path [->, line width = 0.5mm, color = black] (Y-6) edge node[left] {} (X-5);
		\path [->, line width = 0.5mm, color = black] (Y-5) edge node[left] {} (X-4);
		\path [->, line width = 0.5mm, color = black] (Y-4) edge node[left] {} (X-3);
		\path [->, line width = 0.5mm, color = black] (Y-3) edge node[left] {} (X-2);
		\path [->, line width = 0.5mm, color = black] (Y-2) edge node[left] {} (X-1);
		\path [->, line width = 0.5mm, color = black] (Y-1) edge node[left] {} (X0);
		
		\path [->, line width = 0.5mm, bend right, color = orange] (Y-10) edge node[left] {} (Y-7);
		\path [->, line width = 0.5mm, bend right, color = orange] (Y-9) edge node[left] {} (Y-6);
		\path [->, line width = 0.5mm, bend right, color = orange] (Y-8) edge node[left] {} (Y-5);
		\path [->, line width = 0.5mm, bend right, color = orange] (Y-7) edge node[left] {} (Y-4);
		\path [->, line width = 0.5mm, bend right, color = orange] (Y-6) edge node[left] {} (Y-3);
		\path [->, line width = 0.5mm, bend right, color = orange] (Y-5) edge node[left] {} (Y-2);
		\path [->, line width = 0.5mm, bend right, color = orange] (Y-4) edge node[left] {} (Y-1);
		\path [->, line width = 0.5mm, bend right, color = orange] (Y-3) edge node[left] {} (Y0);
		
		\path [->, line width = 0.5mm] (U-10) edge node[left] {} (X-5);
		\path [->, line width = 0.5mm] (U-9) edge node[left] {} (X-4);
		\path [->, line width = 0.5mm] (U-8) edge node[left] {} (X-3);
		\path [->, line width = 0.5mm] (U-7) edge node[left] {} (X-2);
		\path [->, line width = 0.5mm] (U-6) edge node[left] {} (X-1);
		\path [->, line width = 0.5mm] (U-5) edge node[left] {} (X0);
		
		\path [->, line width = 0.5mm] (U-10) edge node[left] {} (X-6);
		\path [->, line width = 0.5mm] (U-9) edge node[left] {} (X-5);
		\path [->, line width = 0.5mm] (U-8) edge node[left] {} (X-4);
		\path [->, line width = 0.5mm] (U-7) edge node[left] {} (X-3);
		\path [->, line width = 0.5mm] (U-6) edge node[left] {} (X-2);
		\path [->, line width = 0.5mm] (U-5) edge node[left] {} (X-1);
		\path [->, line width = 0.5mm] (U-4) edge node[left] {} (X0);

		\path [->, line width = 0.5mm] (U-10) edge node[left] {} (Y-9);
		\path [->, line width = 0.5mm] (U-9) edge node[left] {} (Y-8);
		\path [->, line width = 0.5mm] (U-8) edge node[left] {} (Y-7);
		\path [->, line width = 0.5mm] (U-7) edge node[left] {} (Y-6);
		\path [->, line width = 0.5mm] (U-6) edge node[left] {} (Y-5);
		\path [->, line width = 0.5mm] (U-5) edge node[left] {} (Y-4);
		\path [->, line width = 0.5mm] (U-4) edge node[left] {} (Y-3);
		\path [->, line width = 0.5mm] (U-3) edge node[left] {} (Y-2);
		\path [->, line width = 0.5mm] (U-2) edge node[left] {} (Y-1);
		\path [->, line width = 0.5mm] (U-1) edge node[left] {} (Y0);
		
		\path [->, line width = 0.5mm] (U-10) edge node[left] {} (Y-8);
		\path [->, line width = 0.5mm] (U-9) edge node[left] {} (Y-7);
		\path [->, line width = 0.5mm] (U-8) edge node[left] {} (Y-6);
		\path [->, line width = 0.5mm] (U-7) edge node[left] {} (Y-5);
		\path [->, line width = 0.5mm] (U-6) edge node[left] {} (Y-4);
		\path [->, line width = 0.5mm] (U-5) edge node[left] {} (Y-3);
		\path [->, line width = 0.5mm] (U-4) edge node[left] {} (Y-2);
		\path [->, line width = 0.5mm] (U-3) edge node[left] {} (Y-1);
		\path [->, line width = 0.5mm] (U-2) edge node[left] {} (Y0);
		
		\path [->, line width = 0.5mm, color = red] (X-10) edge node[left] {} (Y-7);
		\path [->, line width = 0.5mm, color = red] (X-9) edge node[left] {} (Y-6);
		\path [->, line width = 0.5mm, color = red] (X-8) edge node[left] {} (Y-5);
		\path [->, line width = 0.5mm, color = red] (X-7) edge node[left] {} (Y-4);
		\path [->, line width = 0.5mm, color = red] (X-6) edge node[left] {} (Y-3);
		\path [->, line width = 0.5mm, color = red] (X-5) edge node[left] {} (Y-2);
		\path [->, line width = 0.5mm, color = red] (X-4) edge node[left] {} (Y-1);
		\path [->, line width = 0.5mm, color = red] (X-3) edge node[left] {} (Y0);

		\end{tikzpicture}
	}
	\caption{Example full time graph. Here, the \color{red} red \color{black} edges correspond to $A^{(3)}_{YX}$, the \color{blue} blue \color{black} edges to $A^{(1)}_{YY}$ and the \color{orange} orange \color{black} edges to $A^{(3)}_{YY}$.}
	\label{ex_full_time_graph_app_2}
\end{figure}

\begin{example}
\label{ex_app_2}
Consider the full time graph from Figure \ref{ex_full_time_graph_app_2}. Here,
$d_U=1$ and $d_O=2$. Abbreviate $\{X_t\}_{t\in\mathbb{Z}}:=\{O^2_t\}_{t\in\mathbb{Z}}$. Note that $m_U=1$, $l^U_1=1$, $m_Y=2$, $l^Y_1=1$, $l^Y_2=3$, $m_{X}=1$, $l^{X}_1=2$, $m_{YX}=1$, $l^{YX}_1=3$, $m_{XY}=1$, $l^{XY}_1=1$, $m_{YU}=2$, $l^{YU}_1=1$, $l^{YU}_2=2$, $m_{XU}=2$, $l^{XU}_1=4$, $l^{XU}_2=5$.

First, note that $\textnormal{pa}^{\textnormal{obs}}(Y_t)=\{Y_{t-1}, Y_{t-3}, X_{t-3}\}$ and $\textnormal{pa}^{\textnormal{lat}}(Y_t)=\{U_{t-1}, U_{t-2}\}$. Let
\begin{align*}
    F^{\textnormal{obs}}:=\{X_{t+2}\}.
\end{align*}
Then, $\textnormal{pa}^{\textnormal{obs}}(F^{\textnormal{obs}})\setminus F^{\textnormal{obs}}=\{X_{t}, Y_{t+1}\}$ and $\textnormal{pa}^{\textnormal{lat}}(F^{\textnormal{obs}})=\{U_{t-2},U_{t-3}\}$. 
Therefore,
\begin{align*}
    C=\{X_{t-3}, Y_{t-1}, Y_{t-3}, X_{t+2}, X_{t}, Y_{t+1}\}
\end{align*}
is valid.
Next, consider
\begin{align*}
    B_U:=\{U_{t-3}\}.
\end{align*}
Then, $\tau_{Y}=\{3\}$ and $\tau_{X}=\{1\}$ are valid.

Also note that $C^{(1)}_{Y}=\{Y_{t-1}, Y_{t-3}, Y_{t+1}\}$ and $C^{(1)}_{X}=\{X_{t-3}, X_{t}\}$. Thus, Corollary \ref{corollary_identifiability} applies and we get identifiability.
One possible $R$ according to Lemma \ref{lemma_resid_class1} is
\begin{align*}
    R:=\{X_{t-3}, Y_{t-5}, Y_{t-4}, Y_{t-3}, X_{t-2}, X_{t-4}\}.
\end{align*}
In generic settings, the parameters of interest $A^{(3)}_{YX}$ and $A^{(1)}_{YY}$ and $A^{(3)}_{YY}$ are given by
    \begin{align*}
    \label{est_app_2}
        &\begin{pmatrix}
        \color{red}A^{(3)}_{YX}\color{black}\\
        \color{blue}A^{(1)}_{YY}\color{black}\\
        \color{orange}A^{(3)}_{YY}\color{black}\\
        \vdots \\
        \textnormal{other terms}\\
        \vdots
        \end{pmatrix}\\
        &=\begin{pmatrix}
        \Gamma_{X_{t-3}\color{red}X_{t-3}\color{black}} & \Gamma_{X_{t-3}\color{blue}Y_{t-1}\color{black}} & \Gamma_{X_{t-3}\color{orange}Y_{t-3}\color{black}} & \Gamma_{X_{t-3}X_{t+2}} & \Gamma_{X_{t-3}X_{t}} & \Gamma_{X_{t-3}Y_{t+1}}\\
        \Gamma_{Y_{t-5}\color{red}X_{t-3}\color{black}} & \Gamma_{Y_{t-5}\color{blue}Y_{t-1}\color{black}} & \Gamma_{Y_{t-5}\color{orange}Y_{t-3}\color{black}} & \Gamma_{Y_{t-5}X_{t+2}} & \Gamma_{Y_{t-5}X_{t}} & \Gamma_{Y_{t-5}Y_{t+1}}\\
        \Gamma_{Y_{t-4}\color{red}X_{t-3}\color{black}} & \Gamma_{Y_{t-4}\color{blue}Y_{t-1}\color{black}} &
        \Gamma_{Y_{t-4}\color{orange}Y_{t-3}\color{black}} & \Gamma_{Y_{t-4}X_{t+2}} & \Gamma_{Y_{t-4}X_{t}} & \Gamma_{Y_{t-4}Y_{t+1}}\\
        \Gamma_{Y_{t-3}\color{red}X_{t-3}\color{black}} & \Gamma_{Y_{t-3}\color{blue}Y_{t-1}\color{black}} & \Gamma_{Y_{t-3}\color{orange}Y_{t-3}\color{black}} & \Gamma_{Y_{t-3}X_{t+2}} & \Gamma_{Y_{t-3}X_{t}} & \Gamma_{Y_{t-3}Y_{t+1}}\\
        \Gamma_{X_{t-2}\color{red}X_{t-3}\color{black}} & \Gamma_{X_{t-2}\color{blue}Y_{t-1}\color{black}} &
        \Gamma_{X_{t-2}\color{orange}Y_{t-3}\color{black}} &
        \Gamma_{X_{t-2}X_{t+2}} & \Gamma_{X_{t-2}X_{t}} & \Gamma_{X_{t-2}Y_{t+1}}\\
        \Gamma_{X_{t-4}\color{red}X_{t-3}\color{black}} & \Gamma_{X_{t-4}\color{blue}Y_{t-1}\color{black}} & \Gamma_{X_{t-4}\color{orange}Y_{t-3}\color{black}} & \Gamma_{X_{t-4}X_{t+2}} & \Gamma_{X_{t-4}X_{t}} & \Gamma_{X_{t-4}Y_{t+1}}
        \end{pmatrix}^{-1}\cdot
        \begin{pmatrix}
        \Gamma_{X_{t-3}Y_{t}}\\
        \Gamma_{Y_{t-5}Y_{t}}\\
        \Gamma_{Y_{t-4}Y_{t}}\\
        \Gamma_{Y_{t-3}Y_{t}}\\
        \Gamma_{X_{t-2}Y_{t}}\\
        \Gamma_{X_{t-4}Y_{t}}
        \end{pmatrix}.\numberthis
    \end{align*}
    Here, the red, blue and orange colours indicate which component corresponds to which column.
    
     Besides, in Figure \ref{Figure_num_validation_ex_app}b we present a numerical validation for this example.
    \demo

\end{example}
	
	\begin{figure}
	\begin{subfigure}[t]{0.33\textwidth}
		\includegraphics[scale=0.085]{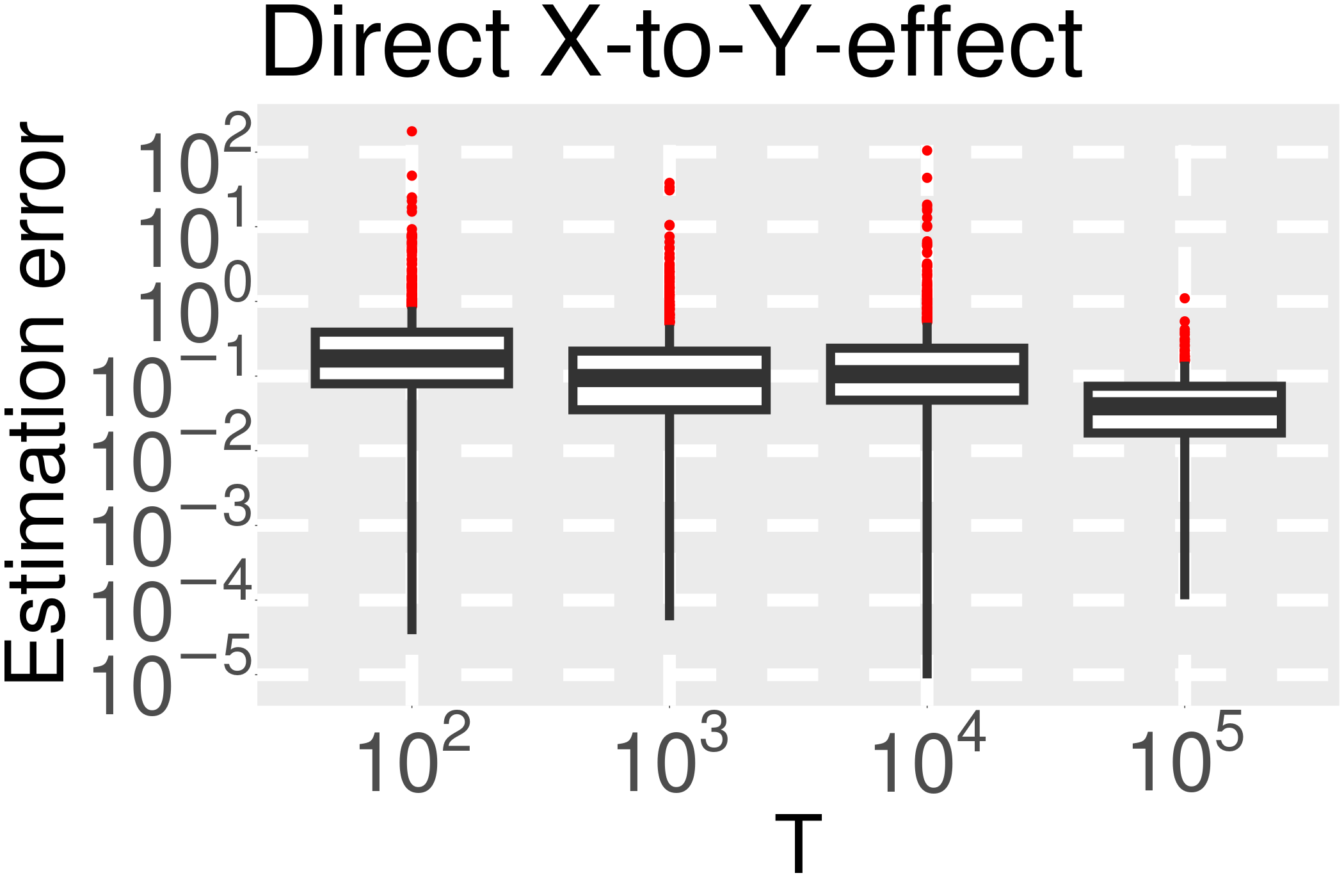}
		\includegraphics[scale=0.085]{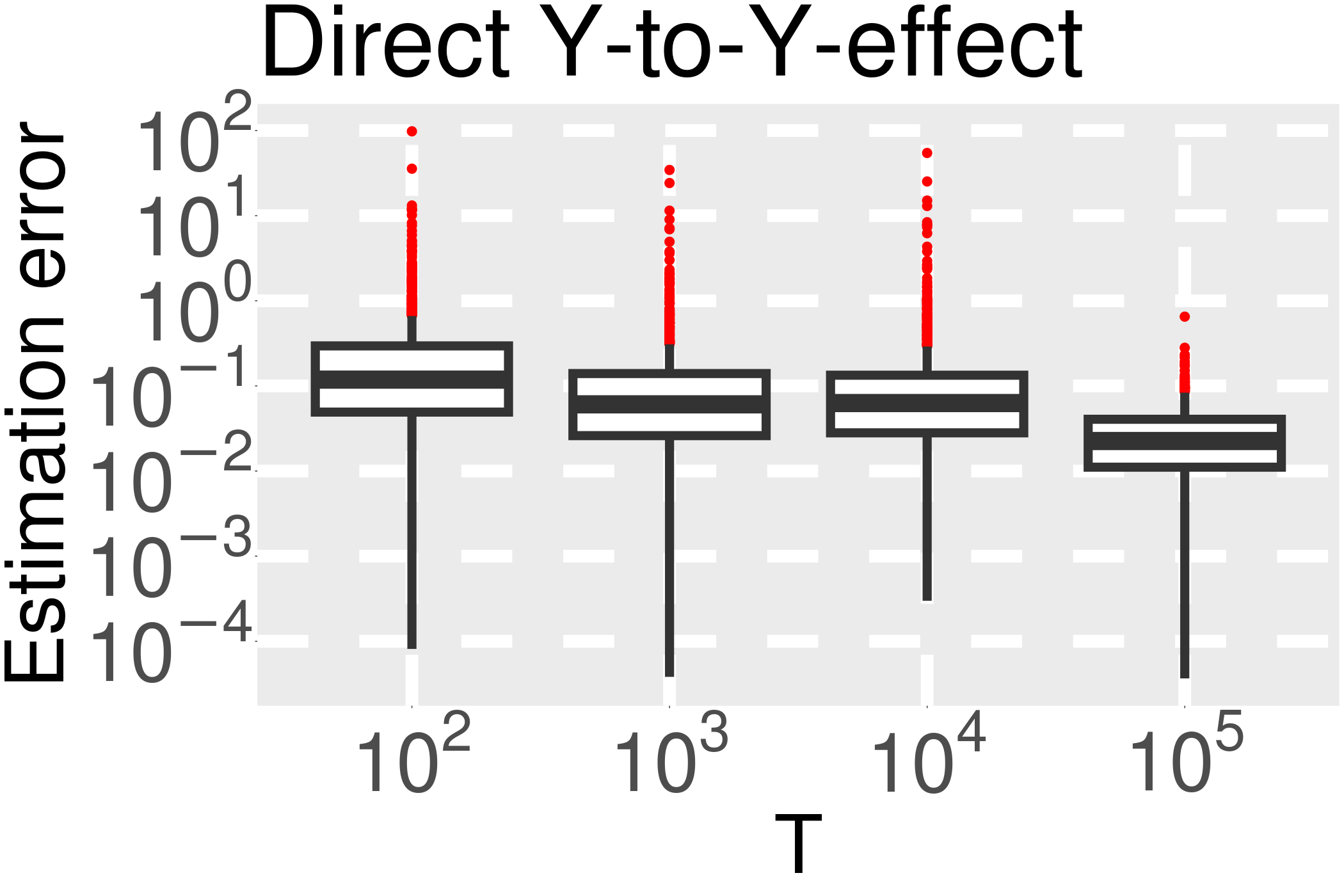}
		\caption{}
	\end{subfigure}
	\begin{subfigure}[t]{0.33\textwidth}
		\includegraphics[scale=0.085]{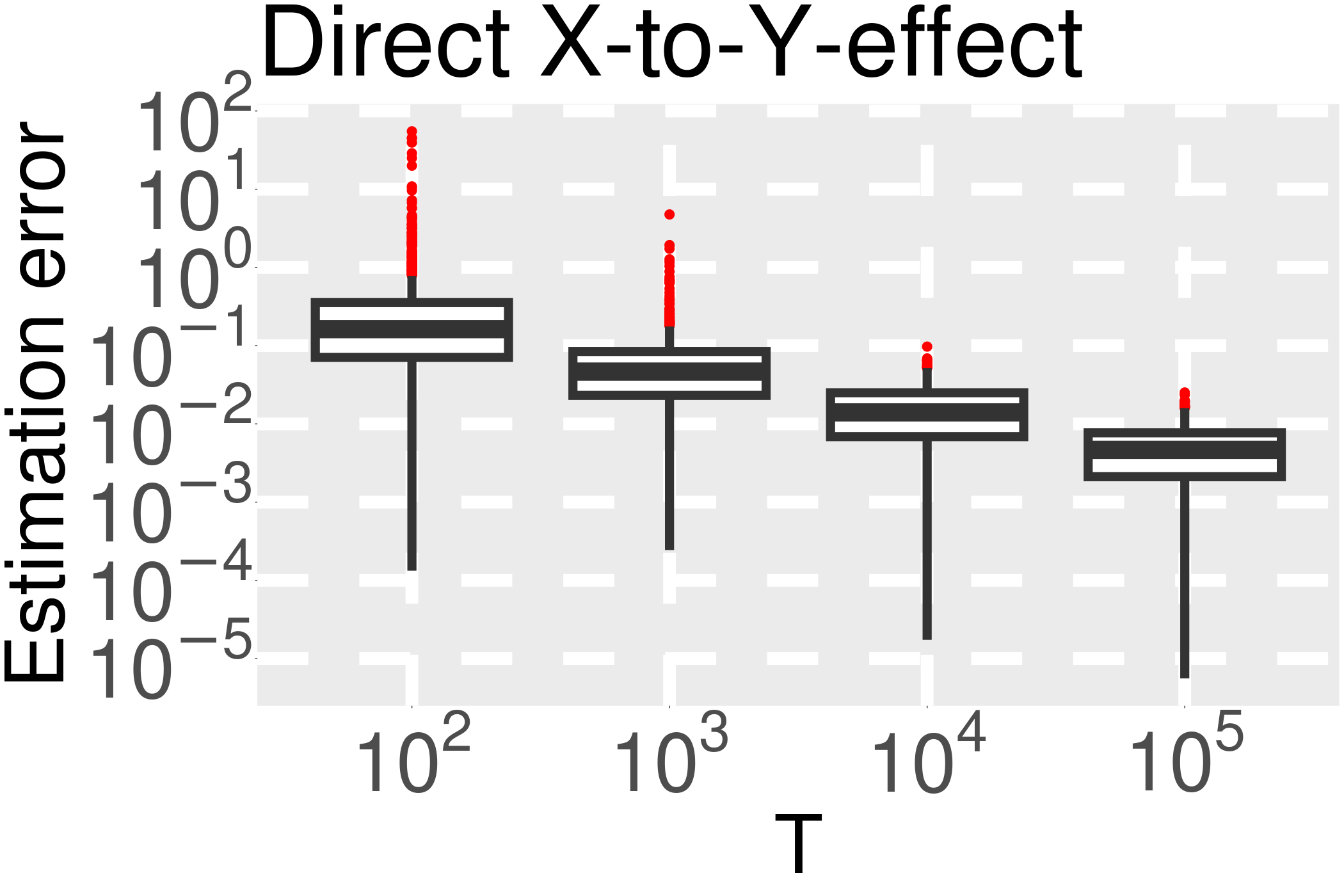}
		\includegraphics[scale=0.085]{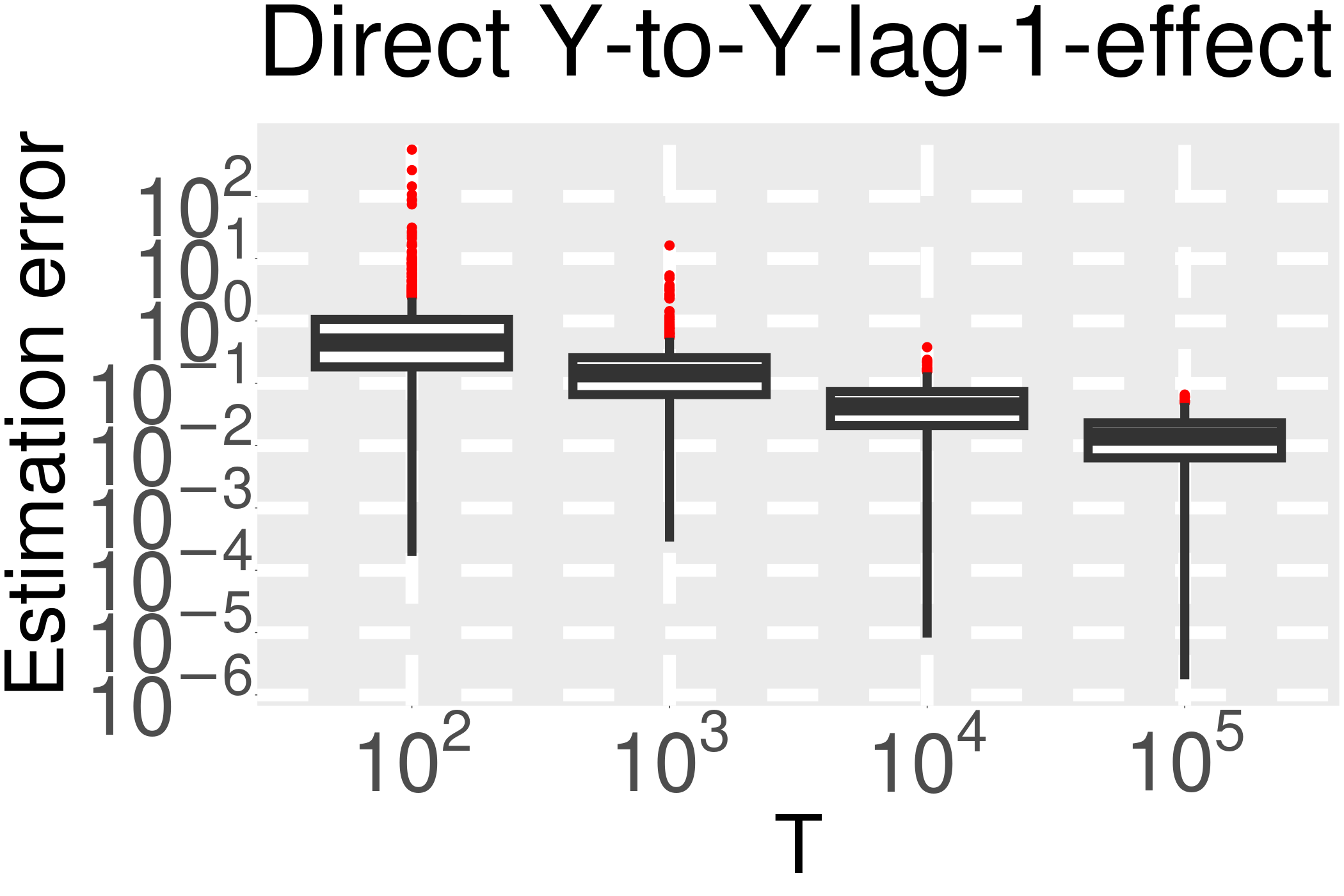}
			\includegraphics[scale=0.085]{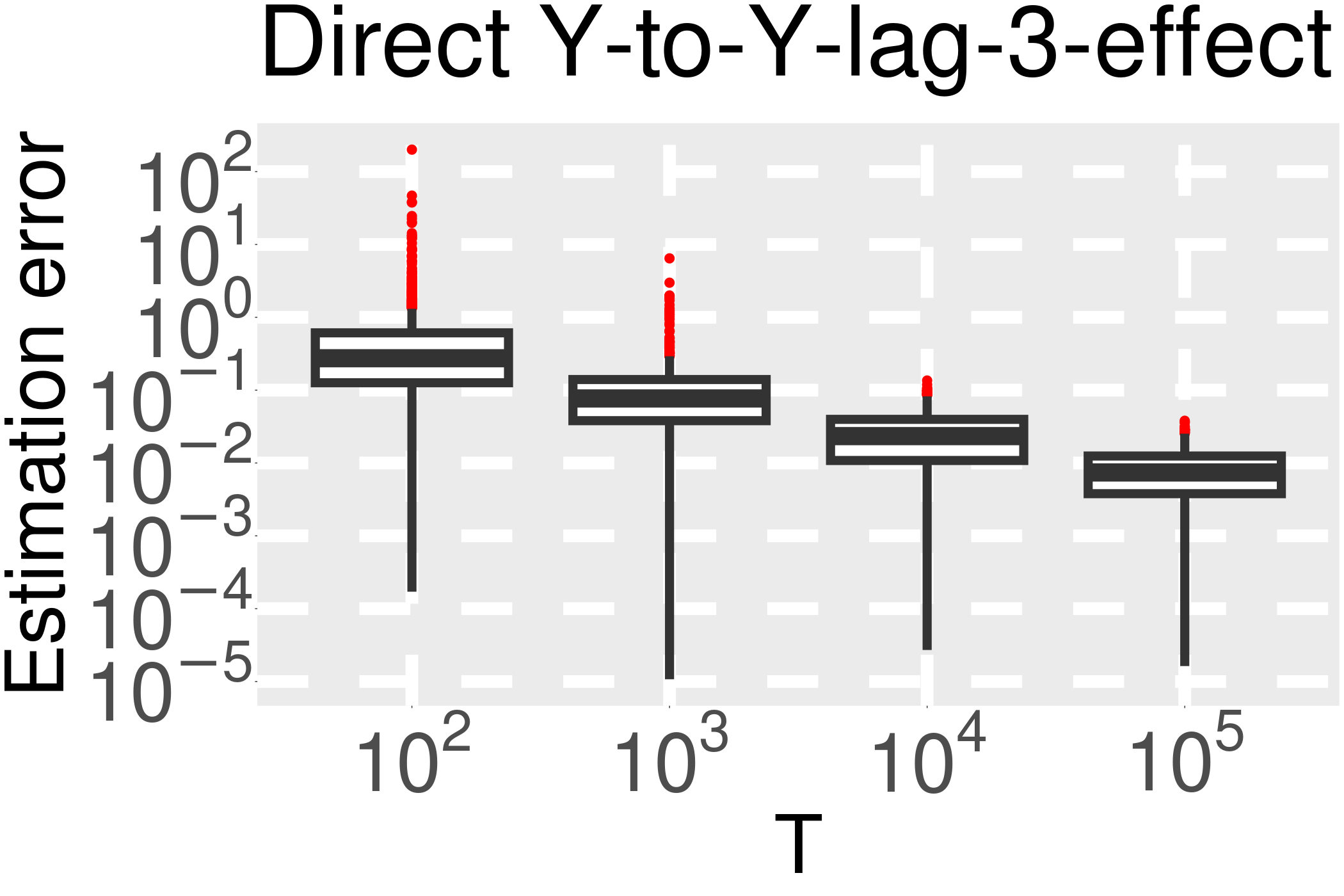}
			\caption{}
		\end{subfigure}
		\begin{subfigure}[t]{0.33\textwidth}
			\includegraphics[scale=0.085]{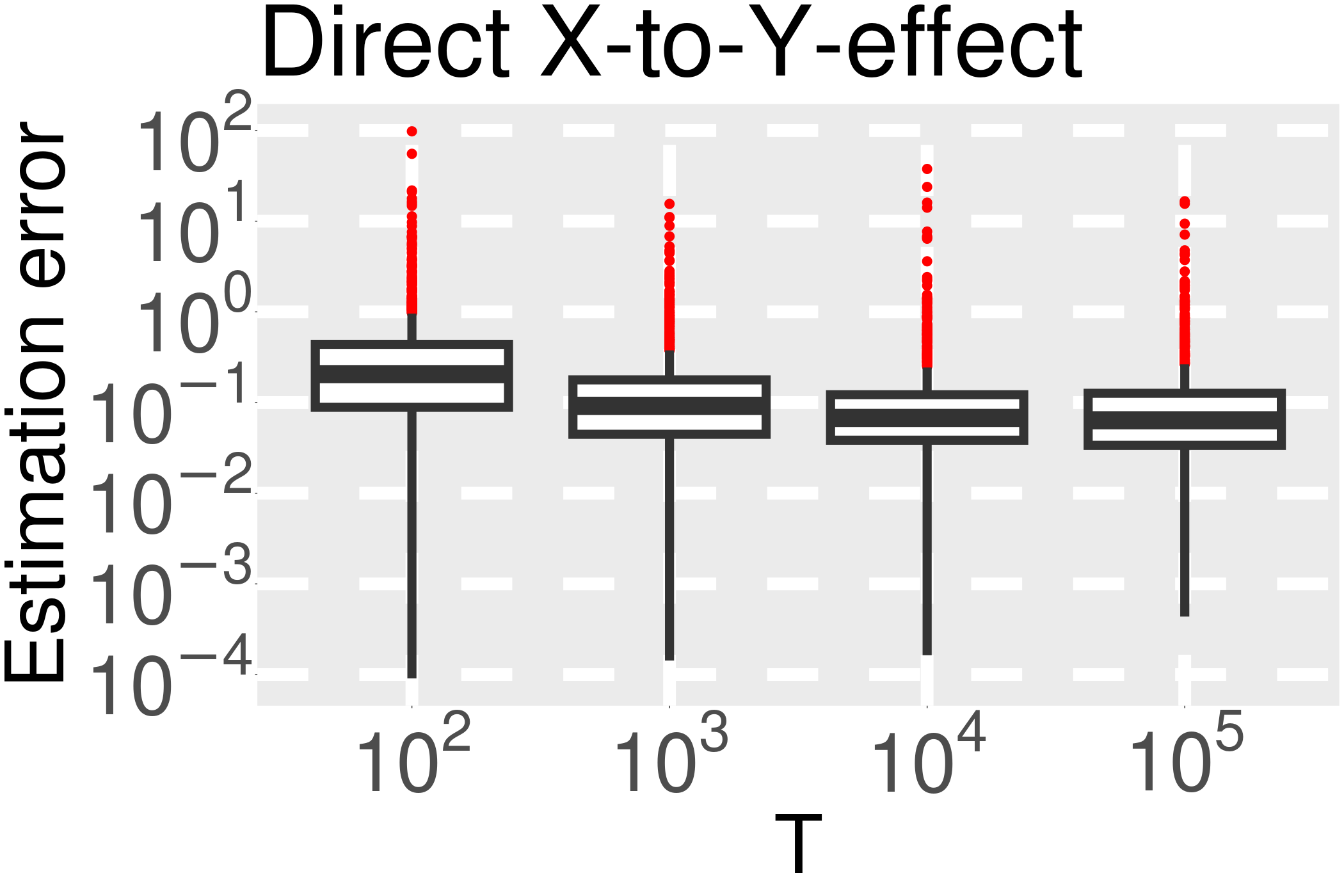}
		\includegraphics[scale=0.085]{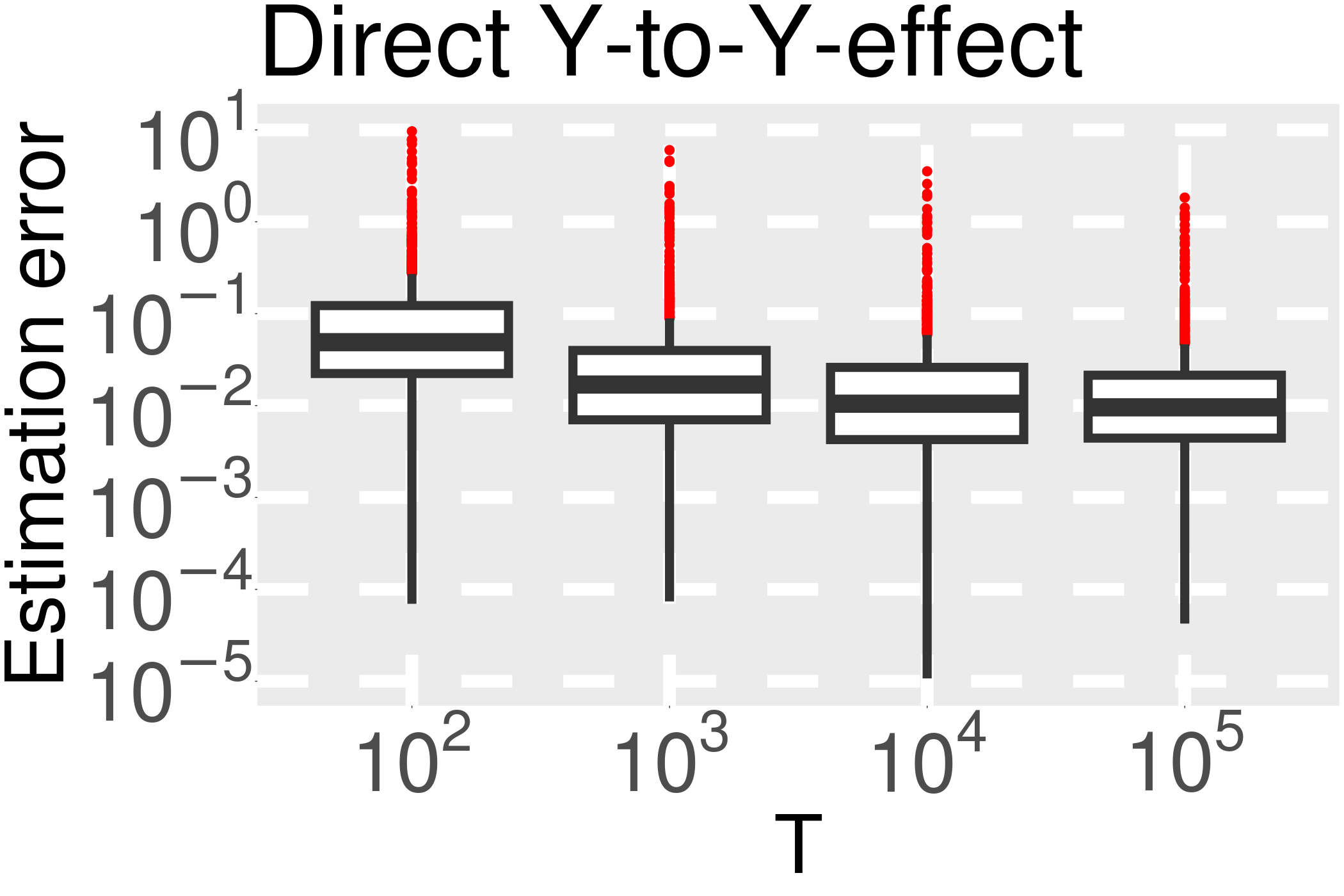}
		\caption{}
		\end{subfigure}
		\caption{A numerical validation of Examples \ref{ex_app_1} (subfigure a, left column), \ref{ex_app_2} (subfigure b, middle column) and \ref{ex_app_3} (subfigure c, right column). For different $1000$ different parameters that yield a stable SVAR process inducing the full time graph from Figure \ref{ex_full_time_graph_app_2} and different time series lengths $T\in\{10^2, 10^3, 10^4, 10^5\}$, we plot the errors of the estimates of $A^{(5)}_{YX}$ and $A^{(1)}_{YY}$ (subfigure a), $A^{(3)}_{YX}$ and $A^{(1)}_{YY}$ and $A^{(3)}_{YY}$ (subfigure b), and $A^{(5)}_{YX}$ and $A^{(2)}_{YY}$ (subfigure c) to the respective true values. Remark: In these boxplots, the whisker's outside the boxes correspond to the smallest and largest points within the $1.5$-inner quartile range. Outliers are highlighted in red. The ordinate axis is log10-transformed.}
		\label{Figure_num_validation_ex_app}
	\end{figure}

\begin{figure}[h]
	\centering
	\scalebox{0.8}{
		\begin{tikzpicture}
		\node[] (t-5) at (-10,3) {\small$t-5$};
		\node[] (t-4) at (-9,3) {\small$t-4$};
		\node[] (t-3) at (-8,3) {\small$t-3$};
		\node[] (t-2) at (-7,3) {\small$t-2$};
		\node[] (t-1) at (-6,3) {\small$t-1$};
		\node[] (t) at (-5,3) {\small$t$};
		\node[] (t+1) at (-4,3) {\small$t+1$};
		\node[] (t+2) at (-3,3) {\small$t+2$};
		\node[] (t+3) at (-2,3) {\small$t+3$};
		\node[] (t+4) at (-1,3) {\small$t+4$};
		\node[] (t+5) at (0,3) {\small$t+5$};
		
		\node[] (X) at (-12,2) {\small$X$};
		\node[] (U) at (-12,0) {\small$U^1$};
		\node[] (U2) at (-12,-2) {\small$U^2$};
		\node[] (Y) at (-12,-4) {\small$Y$};
		
		\node[shape=circle, fill] (U-10) at (-10,0) {};
		\node[shape=circle, fill] (U2-10) at (-10,-2) {};
		\node[shape=circle, fill] (X-10) at (-10,2) {};
		\node[shape=circle, fill] (Y-10) at (-10,-4) {};
		
		\node[shape=circle, fill] (U-9) at (-9,0) {};
		\node[shape=circle, fill] (U2-9) at (-9,-2) {};
		\node[shape=circle, fill] (X-9) at (-9,2) {};
		\node[shape=circle, fill] (Y-9) at (-9,-4) {};
		
		\node[shape=circle, fill] (U-8) at (-8,0) {};
		\node[shape=circle, fill] (U2-8) at (-8,-2) {};
		\node[shape=circle, fill] (X-8) at (-8,2) {};
		\node[shape=circle, fill] (Y-8) at (-8,-4) {};
		
		\node[shape=circle, fill] (U-7) at (-7,0) {};
		\node[shape=circle, fill] (U2-7) at (-7,-2) {};
		\node[shape=circle, fill] (X-7) at (-7,2) {};
		\node[shape=circle, fill] (Y-7) at (-7,-4) {};
		
		\node[shape=circle, fill] (U-6) at (-6,0) {};
		\node[shape=circle, fill] (U2-6) at (-6,-2) {};
		\node[shape=circle, fill] (X-6) at (-6,2) {};
		\node[shape=circle, fill] (Y-6) at (-6,-4) {};
		
		\node[shape=circle, fill] (U-5) at (-5,0) {};
		\node[shape=circle, fill] (U2-5) at (-5,-2) {};
		\node[shape=circle, fill] (X-5) at (-5,2) {};
		\node[shape=circle, fill] (Y-5) at (-5,-4) {};
		
		\node[shape=circle, fill] (U-4) at (-4,0) {};
		\node[shape=circle, fill] (U2-4) at (-4,-2) {};
		\node[shape=circle, fill] (X-4) at (-4,2) {};
		\node[shape=circle, fill] (Y-4) at (-4,-4) {};
		
		\node[shape=circle, fill] (U-3) at (-3,0) {};
		\node[shape=circle, fill] (U2-3) at (-3,-2) {};
		\node[shape=circle, fill] (X-3) at (-3,2) {};
		\node[shape=circle, fill] (Y-3) at (-3,-4) {};
		
		\node[shape=circle, fill] (U-2) at (-2,0) {};
		\node[shape=circle, fill] (U2-2) at (-2,-2) {};
		\node[shape=circle, fill] (X-2) at (-2,2) {};
		\node[shape=circle, fill] (Y-2) at (-2,-4) {};
		
		\node[shape=circle, fill] (U-1) at (-1,0) {};
		\node[shape=circle, fill] (U2-1) at (-1,-2) {};
		\node[shape=circle, fill] (X-1) at (-1,2) {};
		\node[shape=circle, fill] (Y-1) at (-1,-4) {};
		
		\node[shape=circle, fill] (U0) at (0,0) {};
		\node[shape=circle, fill] (U20) at (0,-2) {};
		\node[shape=circle, fill] (X0) at (0,2) {};
		\node[shape=circle, fill] (Y0) at (0,-4) {};
		
		\node[] (Udots-past) at (-10.5,0) {$\ldots$};
		\node[] (Udots-future) at (0.5,0) {$\ldots$};
		\node[] (U2dots-past) at (-10.5,-2) {$\ldots$};
		\node[] (U2dots-future) at (0.5,-2) {$\ldots$};
		\node[] (Xdots-past) at (-10.5,2) {$\ldots$};
		\node[] (Xdots-future) at (0.5,2) {$\ldots$};
		\node[] (Ydots-past) at (-10.5,-4) {$\ldots$};
		\node[] (Ydots-future) at (0.5,-4) {$\ldots$};

		\path [->, line width = 0.5mm] (U-10) edge node[left] {} (U-9);
		\path [->, line width = 0.5mm] (U-9) edge node[left] {} (U-8);
		\path [->, line width = 0.5mm] (U-8) edge node[left] {} (U-7);
		\path [->, line width = 0.5mm] (U-7) edge node[left] {} (U-6);
		\path [->, line width = 0.5mm] (U-6) edge node[left] {} (U-5);
		\path [->, line width = 0.5mm] (U-5) edge node[left] {} (U-4);
		\path [->, line width = 0.5mm] (U-4) edge node[left] {} (U-3);
		\path [->, line width = 0.5mm] (U-3) edge node[left] {} (U-2);
		\path [->, line width = 0.5mm] (U-2) edge node[left] {} (U-1);
		\path [->, line width = 0.5mm] (U-1) edge node[left] {} (U0);
		
		\path [->, line width = 0.5mm] (U-10) edge node[left] {} (U2-9);
		\path [->, line width = 0.5mm] (U-9) edge node[left] {} (U2-8);
		\path [->, line width = 0.5mm] (U-8) edge node[left] {} (U2-7);
		\path [->, line width = 0.5mm] (U-7) edge node[left] {} (U2-6);
		\path [->, line width = 0.5mm] (U-6) edge node[left] {} (U2-5);
		\path [->, line width = 0.5mm] (U-5) edge node[left] {} (U2-4);
		\path [->, line width = 0.5mm] (U-4) edge node[left] {} (U2-3);
		\path [->, line width = 0.5mm] (U-3) edge node[left] {} (U2-2);
		\path [->, line width = 0.5mm] (U-2) edge node[left] {} (U2-1);
		\path [->, line width = 0.5mm] (U-1) edge node[left] {} (U20);
		
        \path [->, line width = 0.5mm] (U2-10) edge node[left] {} (U2-9);
		\path [->, line width = 0.5mm] (U2-9) edge node[left] {} (U2-8);
		\path [->, line width = 0.5mm] (U2-8) edge node[left] {} (U2-7);
		\path [->, line width = 0.5mm] (U2-7) edge node[left] {} (U2-6);
		\path [->, line width = 0.5mm] (U2-6) edge node[left] {} (U2-5);
		\path [->, line width = 0.5mm] (U2-5) edge node[left] {} (U2-4);
		\path [->, line width = 0.5mm] (U2-4) edge node[left] {} (U2-3);
		\path [->, line width = 0.5mm] (U2-3) edge node[left] {} (U2-2);
		\path [->, line width = 0.5mm] (U2-2) edge node[left] {} (U2-1);
		\path [->, line width = 0.5mm] (U2-1) edge node[left] {} (U20);	
		
		\path [->, line width = 0.5mm, bend left] (X-10) edge node[left] {} (X-8);
		\path [->, line width = 0.5mm, bend left] (X-9) edge node[left] {} (X-7);
		\path [->, line width = 0.5mm, bend left] (X-8) edge node[left] {} (X-6);
		\path [->, line width = 0.5mm, bend left] (X-7) edge node[left] {} (X-5);
		\path [->, line width = 0.5mm, bend left] (X-6) edge node[left] {} (X-4);
		\path [->, line width = 0.5mm, bend left] (X-5) edge node[left] {} (X-3);
		\path [->, line width = 0.5mm, bend left] (X-4) edge node[left] {} (X-2);
		\path [->, line width = 0.5mm, bend left] (X-3) edge node[left] {} (X-1);
		\path [->, line width = 0.5mm, bend left] (X-2) edge node[left] {} (X0);

		\path [->, line width = 0.5mm, bend right, color = blue] (Y-10) edge node[left] {} (Y-8);
		\path [->, line width = 0.5mm, bend right, color = blue] (Y-9) edge node[left] {} (Y-7);
		\path [->, line width = 0.5mm, bend right, color = blue] (Y-8) edge node[left] {} (Y-6);
		\path [->, line width = 0.5mm, bend right, color = blue] (Y-7) edge node[left] {} (Y-5);
		\path [->, line width = 0.5mm, bend right, color = blue] (Y-6) edge node[left] {} (Y-4);
		\path [->, line width = 0.5mm, bend right, color = blue] (Y-5) edge node[left] {} (Y-3);
		\path [->, line width = 0.5mm, bend right, color = blue] (Y-4) edge node[left] {} (Y-2);
		\path [->, line width = 0.5mm, bend right, color = blue] (Y-3) edge node[left] {} (Y-1);
		\path [->, line width = 0.5mm, bend right, color = blue] (Y-2) edge node[left] {} (Y0);

		\path [->, line width = 0.5mm] (U2-10) edge node[left] {} (Y-7);
		\path [->, line width = 0.5mm] (U2-9) edge node[left] {} (Y-6);
		\path [->, line width = 0.5mm] (U2-8) edge node[left] {} (Y-5);
		\path [->, line width = 0.5mm] (U2-7) edge node[left] {} (Y-4);
		\path [->, line width = 0.5mm] (U2-6) edge node[left] {} (Y-3);
		\path [->, line width = 0.5mm] (U2-5) edge node[left] {} (Y-2);
		\path [->, line width = 0.5mm] (U2-4) edge node[left] {} (Y-1);
		\path [->, line width = 0.5mm] (U2-3) edge node[left] {} (Y0);
		
		\path [->, line width = 0.5mm] (U2-10) edge node[left] {} (Y-8);
		\path [->, line width = 0.5mm] (U2-9) edge node[left] {} (Y-7);
		\path [->, line width = 0.5mm] (U2-8) edge node[left] {} (Y-6);
		\path [->, line width = 0.5mm] (U2-7) edge node[left] {} (Y-5);
		\path [->, line width = 0.5mm] (U2-6) edge node[left] {} (Y-4);
		\path [->, line width = 0.5mm] (U2-5) edge node[left] {} (Y-3);
		\path [->, line width = 0.5mm] (U2-4) edge node[left] {} (Y-2);
        \path [->, line width = 0.5mm] (U2-3) edge node[left] {} (Y-1);
        \path [->, line width = 0.5mm] (U2-2) edge node[left] {} (Y0);
		
		\path [->, line width = 0.5mm] (U-10) edge node[left] {} (X-9);
		\path [->, line width = 0.5mm] (U-9) edge node[left] {} (X-8);
		\path [->, line width = 0.5mm] (U-8) edge node[left] {} (X-7);
		\path [->, line width = 0.5mm] (U-7) edge node[left] {} (X-6);
		\path [->, line width = 0.5mm] (U-6) edge node[left] {} (X-5);
		\path [->, line width = 0.5mm] (U-5) edge node[left] {} (X-4);
		\path [->, line width = 0.5mm] (U-4) edge node[left] {} (X-3);
		\path [->, line width = 0.5mm] (U-3) edge node[left] {} (X-2);
        \path [->, line width = 0.5mm] (U-2) edge node[left] {} (X-1);
        \path [->, line width = 0.5mm] (U-1) edge node[left] {} (X0);
		
		\path [->, line width = 0.5mm] (U-10) edge node[left] {} (X-8);
		\path [->, line width = 0.5mm] (U-9) edge node[left] {} (X-7);
		\path [->, line width = 0.5mm] (U-8) edge node[left] {} (X-6);
		\path [->, line width = 0.5mm] (U-7) edge node[left] {} (X-5);
		\path [->, line width = 0.5mm] (U-6) edge node[left] {} (X-4);
		\path [->, line width = 0.5mm] (U-5) edge node[left] {} (X-3);
		\path [->, line width = 0.5mm] (U-4) edge node[left] {} (X-2);
		\path [->, line width = 0.5mm] (U-3) edge node[left] {} (X-1);
		\path [->, line width = 0.5mm] (U-2) edge node[left] {} (X0);
		
       \path [->, line width = 0.5mm, color = red] (X-10) edge node[left] {} (Y-5);
		\path [->, line width = 0.5mm, color = red] (X-9) edge node[left] {} (Y-4);
		\path [->, line width = 0.5mm, color = red] (X-8) edge node[left] {} (Y-3);
		\path [->, line width = 0.5mm, color = red] (X-7) edge node[left] {} (Y-2);
		\path [->, line width = 0.5mm, color = red] (X-6) edge node[left] {} (Y-1);
		\path [->, line width = 0.5mm, color = red] (X-5) edge node[left] {} (Y0);

		\end{tikzpicture}
	}
	\caption{Full time graph for Example \ref{ex_app_3}. Here, the \color{red} red \color{black} edges correspond to $A^{(5)}_{YX}$ and the \color{blue} blue \color{black} edges to $A^{(2)}_{YY}$.}
	\label{ex_full_time_graph_app_3}
\end{figure}
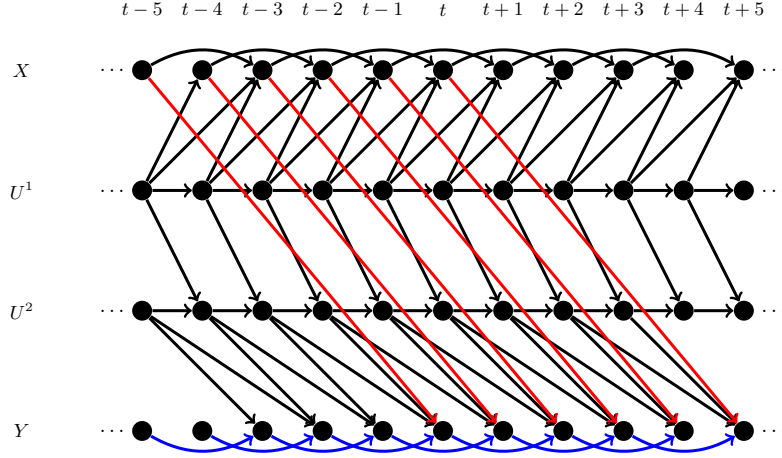

\begin{example}
\label{ex_app_3}
Consider the full time graph from Figure \ref{ex_full_time_graph_app_3}. Here,
$d_U=2$ and $d_O=2$. Abbreviate $\{X_t\}_{t\in\mathbb{Z}}:=\{O^2_t\}_{t\in\mathbb{Z}}$. Note that $m_{U^1}=1$, $l^{U^1}_1=1$, $m_{U^2}=1$, $l^{U^2}_1=1$, $m_{U^2U^1}=1$, $l^{U^2U^1}_1=1$, $m_{U^1U^2}=0$,  $m_Y=1$, $l^Y_1=2$, $m_{X}=1$, $l^{X}_1=2$, $m_{YX}=1$, $l^{YX}_1=5$, $m_{XY}=0$, $m_{YU^2}=2$, $l^{YU^2}_1=2$, $l^{YU^2}_2=3$, $m_{XU^2}=0$, $m_{YU^1}=0$, $m_{XU^1}=2$, $l^{XU^1}_1=1$, $l^{XU^1}_2=2$.

First, note that $\textnormal{pa}^{\textnormal{obs}}(Y_t)=\{Y_{t-2}, X_{t-5}\}$ and $\textnormal{pa}^{\textnormal{lat}}(Y_t)=\{U^2_{t-2}, U^2_{t-3}\}$. Let
\begin{align*}
    F^{\textnormal{obs}}:=\{Y_{t+3}, X_{t+3}\}.
\end{align*}
Then, $\textnormal{pa}^{\textnormal{obs}}(F^{\textnormal{obs}})\setminus F^{\textnormal{obs}}=\{Y_{t+1}, X_{t-2}, X_{t+1}\}$ and $\textnormal{pa}^{\textnormal{lat}}(F^{\textnormal{obs}})=\{U^2_{t-2},U^2_{t-3}, U^1_{t+2}, U^1_{t+1}\}$. 
Therefore,
\begin{align*}
    C=\{X_{t-5}, Y_{t-2}, Y_{t+3}, X_{t+3}, Y_{t+1}, X_{t-2}, X_{t+1}\}
\end{align*}
is valid.
Next, consider
\begin{align*}
    B_U:=\{U^1_{t-3}, U^2_{t-3}\}.
\end{align*}
Then, $\tau_{Y}=\{2\}$ and $\tau_{X}=\{3\}$ are valid.

Also note that $C^{(1)}_{Y}=\{Y_{t-2}, Y_{t+1}\}$ and $C^{(1)}_{X}=\{X_{t-5}, X_{t-2}, X_{t+1}\}$. 
One possible $R$ according to Lemma \ref{lemma_resid_class1} is
\begin{align*}
    R:=\{Y_{t-4}, Y_{t-3},Y_{t-2}, X_{t-6}, X_{t-5}, X_{t-4}, X_{t-3}\}.
\end{align*}
In generic settings, the parameters of interest $A^{(5)}_{YX}$ and $A^{(1)}_{YY}$ and $A^{(3)}_{YY}$ are given by
    \begin{align*}
    \label{est_app_3}
        &\begin{pmatrix}
        \color{red}A^{(5)}_{YX}\color{black}\\
        \color{blue}A^{(2)}_{YY}\color{black}\\
        \vdots \\
        \textnormal{other terms}\\
        \vdots
        \end{pmatrix}\\
        &=\begin{pmatrix}
        \Gamma_{Y_{t-4}\color{red}X_{t-5}\color{black}} & \Gamma_{Y_{t-4}\color{blue}Y_{t-2}\color{black}} & \Gamma_{Y_{t-4}Y_{t+3}} & \Gamma_{Y_{t-4}X_{t+3}} & \Gamma_{Y_{t-4}Y_{t+1}} & \Gamma_{Y_{t-4}X_{t-2}} & \Gamma_{Y_{t-4}X_{t+1}}\\
        \Gamma_{Y_{t-3}\color{red}X_{t-2}\color{black}} & \Gamma_{Y_{t-3}\color{blue}Y_{t-2}\color{black}} & \Gamma_{Y_{t-3}Y_{t+3}} & \Gamma_{Y_{t-3}X_{t+3}} & \Gamma_{Y_{t-3}Y_{t+1}} & \Gamma_{Y_{t-3}X_{t-2}} & \Gamma_{Y_{t-3}X_{t+1}}\\
        \Gamma_{Y_{t-2}\color{red}X_{t-2}\color{black}} & \Gamma_{Y_{t-2}\color{blue}Y_{t-2}\color{black}} &
        \Gamma_{Y_{t-2}Y_{t+3}} & \Gamma_{Y_{t-2}X_{t+3}} & \Gamma_{Y_{t-2}Y_{t+1}} & \Gamma_{Y_{t-2}X_{t-2}} & \Gamma_{Y_{t-2}X_{t-2}}\\
        \Gamma_{X_{t-6}\color{red}X_{t-5}\color{black}} & \Gamma_{X_{t-6}\color{blue}Y_{t-2}\color{black}} & \Gamma_{X_{t-6}Y_{t+3}} & \Gamma_{X_{t-6}X_{t+3}} & \Gamma_{X_{t-6}Y_{t+1}} & \Gamma_{X_{t-6}X_{t-2}} & \Gamma_{X_{t-6}X_{t+1}}\\
        \Gamma_{X_{t-5}\color{red}X_{t-5}\color{black}} & \Gamma_{X_{t-5}\color{blue}Y_{t-2}\color{black}} &
        \Gamma_{X_{t-5}Y_{t+3}} &
        \Gamma_{X_{t-5}X_{t+3}} & \Gamma_{X_{t-5}Y_{t+1}} & \Gamma_{X_{t-5}X_{t-2}} & \Gamma_{X_{t-5}X_{t+1}}\\
        \Gamma_{X_{t-4}\color{red}X_{t-5}\color{black}} & \Gamma_{X_{t-4}\color{blue}Y_{t-2}\color{black}} & \Gamma_{X_{t-4}Y_{t+3}} & \Gamma_{X_{t-4}X_{t+3}} & \Gamma_{X_{t-4}Y_{t+1}} & \Gamma_{X_{t-4}X_{t-2}} & \Gamma_{X_{t-4}X_{t+1}}\\
        \Gamma_{X_{t-3}\color{red}X_{t-5}\color{black}} & \Gamma_{X_{t-3}\color{blue}Y_{t-2}\color{black}} & \Gamma_{X_{t-3}Y_{t+3}} & \Gamma_{X_{t-3}X_{t+3}} & \Gamma_{X_{t-3}Y_{t+1}} & \Gamma_{X_{t-3}X_{t-2}} & \Gamma_{X_{t-3}X_{t+1}}
        \end{pmatrix}^{-1}\cdot
        \begin{pmatrix}
        \Gamma_{Y_{t-4}Y_{t}}\\
        \Gamma_{Y_{t-3}Y_{t}}\\
        \Gamma_{Y_{t-2}Y_{t}}\\
        \Gamma_{X_{t-6}Y_{t}}\\
        \Gamma_{X_{t-5}Y_{t}}\\
        \Gamma_{X_{t-4}Y_{t}}\\
        \Gamma_{X_{t-3}Y_{t}}
        \end{pmatrix}.    \numberthis
    \end{align*}
    Here, the red, blue and orange colours indicate which component corresponds to which column.
    
     Besides, in Figure \ref{Figure_num_validation_ex_app}c we present a numerical validation for this example.
    \demo

\end{example}

\clearpage

\bibliography{main}
\end{document}